\pdfoutput=1
\documentclass[11pt, twoside, table]{predoc2_adjusted}
\usepackage{distoperators}
\usepackage{microtype}
\usepackage{subcaption}
\usepackage{caption}
\usepackage{graphicx}
\usepackage[export]{adjustbox}
\usepackage{tikz}
\usepackage{fancyvrb,newverbs,xcolor} 
\usepackage[shortlabels,inline]{enumitem}
\usepackage{listings}
\usepackage{longtable}
\usepackage{booktabs}
\usepackage{multirow}
\usepackage[title]{appendix}
\usepackage{xfrac}
\usepackage{tabularx}
\usepackage{tabularray}
\usepackage{multirow}
\usepackage[framemethod=TikZ]{mdframed}
\usepackage{amsthm}

\newcounter{theo}[section] 
\renewcommand{\thetheo}{\arabic{theo}}
\newenvironment{theo}[2][]{%
\refstepcounter{theo}%
\ifstrempty{#1}%
{\mdfsetup{%
frametitle={%
\tikz[baseline=(current bounding box.east),outer sep=0pt]
\node[anchor=east,rectangle,fill=lightgray]
{\strut Desideratum~\thetheo};}}
}%
{\mdfsetup{%
frametitle={%
\tikz[baseline=(current bounding box.east),outer sep=0pt]
\node[anchor=east,rectangle,fill=lightgray]
{\strut \sffamily #1};}}%
}%
\mdfsetup{innertopmargin=5pt,linecolor=lightgray,%
linewidth=2pt,topline=true,%
frametitleaboveskip=\dimexpr-\ht\strutbox\relax
}
\begin{mdframed}[]\relax%
\label{#2}}{\end{mdframed}}

\usetikzlibrary{calc}

\theoremstyle{mathdoc}
\newtheorem{proposition}{Proposition}[subsection]
\newtheorem{definition}[proposition]{Definition}
\newtheorem{lemma}[proposition]{Lemma}
\newtheorem{corollary}[proposition]{Corollary}

\title{To select or not to select: predictively consistent priors instead of model selection}
\author[1,*]{Anna Elisabeth Riha}
\author[2]{Leevi Lindgren}
\author{David Kohns}
\author[3]{Paul-Christian Bürkner}
\author[4,1]{Aki Vehtari}
\affil[1]{Department of Computer Science, Aalto University, Finland}
\affil[2]{work done while at Department of Computer Science, Aalto University, Finland}
\affil[3]{Department of Statistics, TU Dortmund University, Germany}
\affil[4]{ELLIS Institute Finland, Finland}
\affil[*]{Corresponding author: anna.riha@aalto.fi}
\makeatletter
\def\runauthor{Riha et al.}
\def\runtitle{To select or not to select}
\makeatother

\begin{document}
\maketitle

\begin{abstract}
    Bayesian modelling workflows often consider multiple candidate models of varying complexity. 
    Model selection is commonly used to navigate potential trade-offs between model complexity and generalisability to new data. 
    We study when model selection is unnecessary or can even be harmful for predictive performance in finite data regimes and find that the need for selecting simpler models can depend on prior choice.
    We formalise predictively consistent priors, which keep prior predictive implications stable as model complexity increases.
    Across examples and numerical experiments, including adding covariates in linear and logistic regression, forward variable selection, and nonlinear modelling,
    flexible models with predictively consistent priors typically match or outperform selected simpler models in out-of-sample predictive performance. 
    When selection helps, it can indicate poor joint prior implications, such as excessive prior mass on implausible predictive values. 
    Based on our findings, we propose replacing the notion of sparsity or parsimony at the level of model components with specifying priors that remain sensible in predictive space as models become more complex.
\end{abstract}


\thispagestyle{empty}

\section{Introduction}
Modellers often start with a simple model that is easy to fit, but, ultimately, want to safely iterate towards more complex models that predict well and are expressive yet computationally feasible \citep{gelman_bayesian_2020}. 
Alternative modelling choices, model checks, or iterative refinements can generate multiple candidate models, motivating transparent strategies for model expansion and comparison 
\citep[see, e.g.,][]{riha_supporting_2024, gelman_garden_2013}.
\citet{dubova_is_2025} emphasise the need for both parsimonious and complex models for scientific progress. 
Complex models are often more realistic \citep{rasmussen_occam_2000}, but preferences or resource limits can motivate choosing a simpler candidate via model selection \citep[for reviews of Bayesian model selection and criteria of model quality, see][, respectively]{vehtari_survey_2012, burkner_models_2023}. 

Statistical folklore says that selecting models with fewer components can avoid overfitting, that is, it helps to balance accuracy on observed (training) data with generalisability on unseen (test) data \citep{hastie_elements_2009}. 
A typical textbook example of this is a polynomial regression where higher-degree polynomials eventually achieve an arbitrarily good fit to the given data, but predict poorly on out-of-sample data \citetext{\citealp[Ch. 7]{mcelreath_statistical_2020}; \citealp{johnson_bayes_2022}}.
Model selection is a commonly suggested remedy, but has itself also been criticised to favour models that are prone to overfitting, for example, when considering stepwise selection from many candidate models \citep{smith_step_2018,ellis_stepwise_2024}. 

To give practical guidance for Bayesian modelling workflows, we examine when model selection is needed and when it may even harm out-of-sample predictive performance. 
Further, we investigate how prior choices affect selection outcomes and the level of model complexity we can reach safely without deteriorating test performance.
To this end, we focus on additive models in finite data regimes, which covers a wide range of use cases across the quantitative sciences.
In particular, we
\begin{enumerate}[nosep,leftmargin=*, label=\textbullet]
  \item Provide general desiderata for good priors and specifically define predictively consistent priors (second-order consistency and predictive parsimony; Sections \ref{sec:good-priors}-\ref{sec:examples-pred-consistent-priors}).
  \item Demonstrate predictively consistent priors in three case studies: adding covariates in linear and logistic models, forward stepwise selection, and increasing nonlinear complexity (Sections \ref{subsec:examples-pred-consistent-priors-illustrative-example-part-3}, \ref{subsec:ex-1-logistic}-\ref{subsec:ex-3-nonlinear}).
  \item Give an example of adjusting a joint prior towards predictive consistency for a randomised-controlled-trial setup with available prior knowledge for the treatment coefficient (Section \ref{subsec:ex-4-rct-studies}).
  \item Summarize our main takeaways and practical recommendations (Section \ref{intro:takeaways-for-modellers}).
\end{enumerate}
Our main recommendation is to replace the notion of sparsity or parsimony at the level of model components with \textit{predictively consistent priors}, that is, proper priors that retain consistent prior predictive distributions on the target space when model complexity increases \citep{gelman_prior_2017, gelman_regression_2020}. 
We operationalise this notion via the implied variance of the predictor distribution (intuition: Section \ref{subsec:methods-intuition-prior-desiderata}; formal definition: Section \ref{subsec:methods-formalisation-prior-desiderata}). 

When correctly integrating over uncertainties and using predictively consistent priors, we find that selecting a simpler model is often unnecessary.
Throughout our experiments in Section \ref{sec:experiments}, complex models with predictively consistent priors typically match or exceed the out-of-sample performance of selected simpler models; and when worse, the loss is usually negligible. 
In some cases, selection within a set of nested models can even harm performance relative to using the biggest model.
If overfitting is a concern, this should be addressed in the initial (joint) priors rather than by selecting among models with predictively inconsistent priors that may favour overfitting.
Importantly, we are not saying  model selection cannot be useful. 
Resource constraints can force reductions in model size and selecting a simpler model can also mitigate the shortcomings of a prior.
In fact, when selection ``helps'', it can indicate a poor prior with too much mass on nonsensical predictive values, motivating better joint prior specification upfront.
  
\subsection{Takeaways for modellers}\label{intro:takeaways-for-modellers}
\paragraph{(1) Joint prior specification should consider independence assumptions and dependence structures between model components.} 
Seemingly weakly informative marginal priors like independent normal priors can imply surprisingly informative joint priors \citep[Figure \ref{fig:illustrative-example-part-3-implied-prior-r2-elpd-test-rho-05-r2-05}; see also][]{efron_discussion_1973}.
Independence is not a safe default as exemplified throughout Section \ref{subsec:ex-1-logistic} to \ref{subsec:ex-2-forward-search}. 
At the same time, when a component is truly independent as in Section \ref{subsec:ex-4-rct-studies}, assuming dependence can also lead to problematic prior predictives and worse predictive performance compared to a prior that reflects the dependence structure.

\paragraph{(2) Prior predictive checks and sensitivity analysis can reveal deficiencies.}
Prior predictive checks often detect too diffuse priors producing nonsensical prior predictive distributions \citep{gabry_visualization_2019, gelman_bayesian_2020} and prior sensitivity analysis highlights influential choices and prior-data conflicts \citep{kallioinen_detecting_2023}.
Checking joint quantities and model summaries can reveal unintended joint effects of individually specified priors and predictive inconsistencies \citep[Figures \ref{fig:illustrative-example-part-3-implied-prior-r2-elpd-test-rho-05-r2-05}, \ref{fig:prior-predictives-logistic-regression} and \ref{fig:ex-1-linear-regression-prior-r2}; see also][]{aguilar_intuitive_2023, aguilar_generalized_2025}.

\paragraph{(3) With predictively consistent priors, model selection is not always needed.}
Predictively consistent priors keep relevant features of the prior predictive distribution stable as model components are added (intuition: Section \ref{subsec:methods-intuition-prior-desiderata}; formal definition: Section \ref{subsec:methods-formalisation-prior-desiderata}).
Across our experiments (Section \ref{sec:experiments}), posterior integration combined with such priors allows increasing model complexity without loss of out-of-sample performance and simpler models rarely outperform more complex ones \citep[see also][]{neal_bayesian_1996,rasmussen_occam_2000}. 
\paragraph{(4) Priors are important, even if the main goal is to reduce the model size.}
When cost or explainability motivate reducing model size, projection predictive inference is recommended \citep{hahn_decoupling_2015, piironen_comparison_2017, piironen_projective_2020, pavone_using_2023,mclatchie_robust_2025}.
Our experiments in Section \ref{subsec:ex-2-forward-search} show that a predictively consistent prior can improve the forward search process of \texttt{projpred} \citep[][]{piironen_projpred_2026}, compared to the case when independent diffuse priors are used.

\paragraph{(5) Model selection being helpful might be an indication of a poor prior.}
If model selection favours a simpler model over a more complex one, this can be caused by a prior placing too much probability mass on nonsensical predictive values favouring overfitting (see Figures \ref{fig:illustrative-example-part-3-implied-prior-r2-elpd-test-rho-05-r2-05} and \ref{fig:ex-2-forward-search-elpd-test-paths-all-rho-n-100-200-r2-05}). 
If selecting a simpler model seems to improve results, there is reason to revisit the chosen priors, especially, if wide default priors have been used or marginal priors have been chosen without considering potential dependence structures between model components. 

\subsection{Related work}\label{subsec:intro-previous-work} 
\subsubsection{In favour of big models}
Building the largest, encompassing model when possible and computationally feasible is a repeatedly advocated idea \citep{lindley_choice_1968, neal_bayesian_1996, rasmussen_occam_2000, vehtari_survey_2012, pavone_using_2023}.
There are many reasons to build complex models, for example, to include all available information and allow flexibility for future applications of the model.
Established concepts like the bias-variance trade-off underline tensions between model complexity and predictive performance on test data \citep{hastie_elements_2009}. 
At the same time, complex yet well-generalising models, for example, Gaussian processes or (Bayesian) neural networks, suggest less straightforward relationships between complexity, bias, and variance \citep{neal_bayesian_1996, rasmussen_gaussian_2006, wilson_position_2025}. 
\cite{rasmussen_occam_2000} argue that parsimony need not imply that a simpler model would perform better, when Ockham's razor is applied at the level of functional complexity.
In the context of generalisation phenomena of deep neural networks like double descent,
\citet{wilson_position_2025} reasons that ``restriction biases'' \citep[][p.3]{wilson_position_2025} are not necessary to avoid overfitting and shows that similarly unintuitive generalisation also occurs in simpler models like linear regression models. 

Even if we do not work with the largest model directly, for example, because it is computationally expensive, it can still be appealing to build a model that is as complex as is feasible. 
Expanding a candidate model is often recommended, especially when the initial model does not satisfy criteria of model quality \citetext{see, e.g., \citealp{gilks_model_1995};
\citealp{ohagan_kendalls_2004}}.
In real-world applications of Bayesian data analysis, 
a typical workflow is to start with a simple, tractable model and iteratively expand toward more realistic models that capture the investigated phenomenon well, provide reliable posterior results and accurate predictions, while also performing model checking and monitoring computation and convergence and revising models if needed \citep{gelman_bayesian_2020}.
As such, Bayesian workflows require strategies for safe model expansion, that is, increasing model complexity without losses in out-of-sample performance.
Moreover, \citet{lindley_choice_1968} suggests performing model selection by identifying models that perform similarly to the largest model encompassing all uncertainty about the modelling problem. 
In continuation and based on ideas by \citet{goutis_model_1998} and \citet{dupuis_variable_2003}, projection predictive inference finds simpler models with comparable predictive performance to a reference model \citetext{\citealp{hahn_decoupling_2015}; \citealp{piironen_projective_2020}}. 
Importantly, the quality of the results depends on the chosen reference model and we need well-performing complex models for successful projection predictive selection \citep[see, e.g.,][]{piironen_comparison_2017,mclatchie_robust_2025, pavone_using_2023}. 

\subsubsection{Prior specification and model complexity}
Prior specification and model selection are separate tasks in Bayesian workflows \citep[see, e.g,][]{bernardo_bayesian_1994,key_bayesian_1999,vehtari_survey_2012, mclatchie_efficient_2024}, but the motivation for employing model selection can help to set better priors from the start.
Particularly, if we care about posterior predictive performance on yet unobserved data, we need to consider the implications of our prior assumptions in predictive space \citep[][]{gelman_prior_2017, simpson_penalising_2017}. 
In practice, limited data availability restricts what we can learn and not all effects in our models can be large \citep{tosh_piranha_2025}, but that does not imply that some or most of the true effects are zero. 
Instead of ``sparsifying'' the model after building it, priors can reflect a notion of sparsity that allows small or close-to-zero effects despite increasing model complexity. 
\citet{wilson_position_2025} recommends that modellers should 
``embrace a flexible hypothesis space, with a soft preference for simpler solutions [...] consistent with the data'' \citep[][p.1]{wilson_position_2025} as a result of investigating generalisation properties of complex models like deep neural networks. 
While targeting different models, this suggestion is conceptually close to our further discussion of desired properties of priors in Section \ref{subsec:methods-intuition-prior-desiderata} and \ref{subsec:methods-formalisation-prior-desiderata}.

The more parameters a model has, the more challenging it can be to define meaningful individual priors that also imply sensible joint priors. 
It is not enough to elicit priors on a parameter-to-parameter basis, as their joint behaviour affects inferences \citep{gelman_prior_2017, mikkola_prior_2024}. 
Sometimes, a simpler model is used only because it is hard to come up with a sensible prior for a more complex version of the model \citep{neal_bayesian_1996}. 
A range of existing approaches tackle aspects of this challenge, for example, by embedding sparsity assumptions with global-local shrinkage priors \citep{polson_shrink_2011}, restricting predictive complexity with reference model priors \citep{nalisnick_predictive_2021}, considering a prior tree structure and a hierarchical decomposition of prior contributions for Bayesian hierarchical models \citep[][]{hem_makemyprior_2024, fuglstad_intuitive_2020}, or setting a prior on a joint quantity like the proportion of variance in the outcome explained by the model \citep[Bayesian $R^2$,][]{gelman_r-squared_2019}, that is propagated back to implied priors for the individual model components \citetext{\citealp{zhang_bayesian_2022}; \citealp{aguilar_intuitive_2023}; \citealp{yanchenko_r2d2_2025}}. 
Another strategy to ensure sensible joint priors is setting priors on function space.
This is an integral part of approaches using Gaussian processes \citep{rasmussen_gaussian_2006} and nonparametric Bayes \citetext{\citealp{hjort_bayesian_2010}; \citealp{ghosal_fundamentals_2017}}, including tree-based function priors such as Bayesian additive regression trees \citep{chipman_bart_2010}. 
It has also been studied in the context of prior specification for Bayesian neural networks \citep[see, e.g., ][]{fortuin_priors_2022, flam-shepherd_mapping_2017, tran_all_2022}.
To enable a more thorough mathematical investigation, we focus primarily on additive models in this paper.

\subsection{Structure of this paper}
Sections \ref{sec:methods} and \ref{sec:good-priors} define key concepts, including our notion of predictively consistent priors. 
Section \ref{sec:examples-pred-consistent-priors} introduces two such priors with illustrative examples, and Section \ref{sec:experiments} evaluates them on simulated data across three settings that are typically framed as prone to overfitting: adding covariates in logistic regression (Section \ref{subsec:ex-1-logistic}), forward variable selection (Section \ref{subsec:ex-2-forward-search}), and increasing complexity in nonlinear models (Section \ref{subsec:ex-3-nonlinear}). 
We also customise a hierarchical shrinkage prior for randomised controlled trials with treatment-specific prior information (Section \ref{subsec:ex-4-rct-studies}). 
We conclude with a discussion of the main findings and directions for future work in Section \ref{sec:discussion}.\footnote{Code to replicate the illustrative examples and experiments in this paper is available at \href{https://github.com/annariha/to-select-or-not}{https://github.com/annariha/to-select-or-not}.} 

\section{Methods}\label{sec:methods} 
\subsection{Key concepts}
We introduce terminology used throughout this paper, including a collection of nested candidate models with an associated notion of complexity, a utility for ranking them, and a selection rule given by an oracle model (see Table \ref{tab:methods-terms} for a summary). 
\begin{table}[t]
    \centering
    {\small 
    \begin{tabularx}{\textwidth}{l|X}
        Term & Definition and/or example \\
        \hline
        True data-generating process & $p_t(y):= p(y \mid \theta=\theta_0)$ with known true $\theta_0$ \\
        True-structure model & $p(\theta, y)$ with parameter space $\Theta$ such that $\theta_0 \in \Theta$ \eqref{eq:methods-true-structure-models} \\
        KL-minimising parameters & $\theta_{\rm KL} =\argmin_{\theta \in \Theta} \mathrm{KL} \left(p_t(y) \ || \ p(y \mid \theta) \right)$ \eqref{eq:methods-theta-KL} \\
        Generative model & $p\left(y,X, \theta\right) = p\left(\theta\right) \ p\left(X\right) \ p\left( y \mid X, \theta\right) $ \\
        Partially generative model & 
        $p(\theta,y\mid X) = p(\theta)\,p(y\mid X,\theta)$, $\rm X$ treated as fixed \\
        Finite collection of models & $\mathcal{M} = \left\{\mathrm{M}_k\right\}_{k=1}^K $ where each $\mathrm{M}_k := p_{\mathrm{M}_k}(\theta_{\mathrm{M}_k},y\mid X)$ with parameter space $\Theta_{\mathrm{M}_k}$ \\
        Nested models & $p_{\mathrm{M}_2}(y \mid \{\theta_{\mathrm{M}_1}, \theta_{\mathrm{M}_2\setminus\mathrm{M}_1} = 0\}) = p_{\mathrm{M}_1}(y \mid \theta_{\mathrm{M}_1})$ where $\Theta_{\mathrm{M}_2}=\Theta_{\mathrm{M}_1}\times\Theta_{\mathrm{M}_2\setminus\mathrm{M}_1}$ and $\Theta_{\mathrm{M}_1}\times\{0\}\subset\Theta_{\mathrm{M}_2}$ \eqref{eq:methods-nested}
        \\ 
        Observation model/data model &  $p(y \mid \theta)$ as a function of $y$\\
        Likelihood & $p(y \mid \theta)$ as a function of $\theta$ \\
        Generative prior & proper, allows to generate coefficient values \\ 
        Predictively consistent prior & generative in predictive space, fulfils second-order predictive consistency \eqref{def:second-order-pc}, often also predictive parsimony \eqref{def:pred-parsimony}  \\
        Oracle model & ``best-in-collection'' $\argmax_{\mathrm{M}\in\mathcal{M}} \ \mathbb{E}_{\tilde{y}} 
        \big[\mathrm{u}\big(\mathrm{M}(y), \tilde{y}\big)\big]$ \eqref{eq:oracle-def}
        \\
    \end{tabularx}}
    \caption{Terms and definitions used in the article.}
    \label{tab:methods-terms}
\end{table}
\paragraph{Candidate models and model complexity}
We work with a finite collection of candidate models $\mathcal{M}=\{\mathrm{M}_k\}_{k=1}^{K}$. 
Each candidate is the joint model $\mathrm{M}_k := p_{\mathrm{M}_k}(\theta_{\mathrm{M}_k},y\mid X)$, with parameter space $\Theta_{\mathrm{M}_k}$ and covariates
$X$. 
After observing training data $y$, we write $\mathrm{M}_k(y) := p_{\mathrm{M}_k}(\,\cdot\mid y,X)$ as the posterior predictive, which is evaluated at new $\tilde{y}$ when assessing test performance. 
Two models are \emph{nested}, written $\mathrm{M}_1\subset\mathrm{M}_2$, when their parameter spaces decompose as $\Theta_{\mathrm{M}_2}=\Theta_{\mathrm{M}_1}\times\Theta_{\mathrm{M}_2\setminus\mathrm{M}_1}$, where $\Theta_{\mathrm{M}_2\setminus\mathrm{M}_1}$ collects parameters present in $\mathrm{M}_2$ but not in $\mathrm{M}_1$, and $\Theta_{\mathrm{M}_1}$ is identified with $\Theta_{\mathrm{M}_1}\times\{0\}\subset\Theta_{\mathrm{M}_2}$ via 
$\theta_{\mathrm{M}_1}\mapsto(\theta_{\mathrm{M}_1},0)$, such that
\begin{align}\label{eq:methods-nested}
p_{\mathrm{M}_2}\!\left(y\mid\theta_{\mathrm{M}_1},\,\theta_{\mathrm{M}_2\setminus\mathrm{M}_1}=0,X\right)
\,=\,p_{\mathrm{M}_1}\!\left(y\mid\theta_{\mathrm{M}_1},X\right).
\end{align}
We omit the explicit conditioning on $X$ in the equations below.
For nested models $\mathrm{M}_1\subset\cdots\subset\mathrm{M}_K$, we identify \emph{model complexity} with a model's position, that is, $\mathrm{M}_K$ is the most complex model and $\mathrm{M}_{k+1}$ is more complex than $\mathrm{M}_k$ (see Figure \ref{fig:illustration-key-concepts-true-structure-model-KL-minimising-parameter} for a simplified illustration of 
$\Theta_{k-1}\subset\Theta_k\subset\Theta_{k+1}\subset\cdots$).

\paragraph{True data-generating process and true-structure models} 
We denote by $p_t(y)$ the true model or \emph{true data-generating process} (DGP) of $y$. 
The DGP does not need to be within the set of candidate models $\mathcal{M}$, but when it is, there exists a model $M_k \in \mathcal{M}$ and a parameter value $\theta_0\in\Theta_{\mathrm{M}_k}$ such that 
\begin{align}\label{eq:methods-true-structure-models} 
    p_t(y)\,=\,p_{\mathrm{M}_k}(y\mid\theta_0),
\end{align}
we call $\mathrm{M}_k$ a \emph{true-structure model} and any such $\theta_0$ \emph{true-structure parameters} \citep{bernardo_bayesian_1994, vehtari_survey_2012}. 
Further, $p_{\mathrm{M}_k}(y\mid\theta_0)$ is then (a parameterisation of) the DGP.
The smallest $\mathrm{M}_k$ in a set of nested models is the \emph{minimal true-structure model} within this set.
Figure \ref{fig:illustration-key-concepts-true-structure-model-KL-minimising-parameter} 
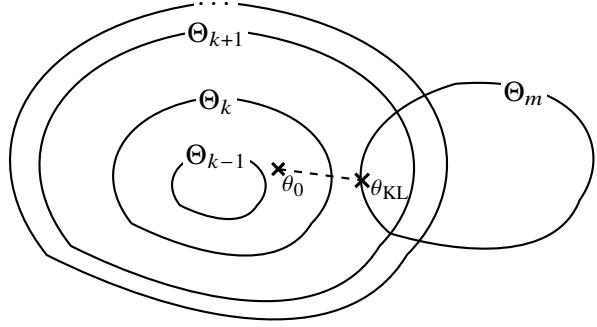
\begin{figure}[t]
    \centering
    \begin{tikzpicture}
    \def\irregular{
    (0,1.8)
        .. controls (2.2,1.6) and (2.6,0.2) .. (1.8,-0.6)
        .. controls (1.2,-1.6) and (-0.4,-1.2) .. (-1.6,-0.6)
        .. controls (-2.4,0.4) and (-1.8,1.6) .. (0,1.8)
    }
    \newcommand{\cross}[2]{%
    \draw[very thick] ($(#1)+(-#2,-#2)$) -- ($(#1)+(#2,#2)$);
    \draw[very thick] ($(#1)+(-#2,#2)$) -- ($(#1)+(#2,-#2)$);
  }
  \foreach \s/\pos/\lab in {
      1.40/1.0/{\cdots},
      1.2/1.0/{\Theta_{k+1}},
      0.7/1.0/{\Theta_{k}},
      0.30/1.0/{\Theta_{k-1}}
  }{
    \draw[thick] [scale=\s] \irregular
    node[pos=\pos, yshift=-2pt, fill=white, inner sep=1pt] {$\lab$};
  }
    \begin{scope}[shift={(3.40,0.10)}, scale=0.75, rotate=10]
        \draw[thick] \irregular
        node[pos=1.2, yshift=-8pt, fill=white, inner sep=1pt] {$\Theta_{m}$};
    \end{scope}
  \coordinate (t0) at (0.82,0.30);
  \cross{t0}{0.08}
  \node[anchor=north] at ($(t0)+(0.2,0.05)$) {\footnotesize{$\theta_0$}};
  \coordinate (tKL) at (1.93,0.15);
  \cross{tKL}{0.09}
  \node[anchor=north] at ($(tKL)+(0.37,0.15)$) {\footnotesize{$\theta_{\rm KL}$}};
  \draw[dashed, thick] (t0) -- (tKL);
  \end{tikzpicture}
  \vspace{-0.8cm}
  \caption{Parameter spaces of true-structure models $\Theta_k$ and $\Theta_{k+1}$ containing $\theta_0$, with $\Theta_k$ being the parameter space of the minimal true-structure model since $\theta_0 \notin \Theta_{k-1}$. The KL-minimising parameter is $\theta_{\rm KL} = \argmin_{\theta \in \Theta} \mathrm{KL}\left(p_t(y \ || \ p(y \mid \theta)\right)$ and $\theta_{\rm KL} = \theta_0$ if the model is a true-structure model.}
  \label{fig:illustration-key-concepts-true-structure-model-KL-minimising-parameter}
\end{figure}
illustrates $\theta_0\in\Theta_k$ but $\theta_0\notin\Theta_{k-1}$. 
This means $\mathrm{M}_k$ is the minimal true-structure model, and every $\mathrm{M}_{k'}$ with $k'\ge k$ is also a true-structure model.

Independently of whether \eqref{eq:methods-true-structure-models} can be satisfied, we define the Kullback–Leibler (KL) projection of $p_t$ onto model $\mathrm{M}_m$ as
\begin{align}\label{eq:methods-theta-KL}
\theta_{\mathrm{KL}}\,=\,\argmin_{\theta\in\Theta_{\mathrm{M}_m}}\mathrm{KL}\!\left(p_t(y)\,\Vert\,p_{\mathrm{M}_m}(y\mid\theta)\right),
\end{align}
that is, the parameter in $\Theta_{m}$ whose implied distribution is closest to $p_t$ in KL divergence. 
If $\mathrm{M}_m$ is a true-structure model, we have $\theta_{\mathrm{KL}}=\theta_0$. 
Figure \ref{fig:illustration-key-concepts-true-structure-model-KL-minimising-parameter} shows the parameter space $\Theta_m$ of a misspecified candidate $\mathrm{M}_m$ with the dashed line pointing to its KL projection $\theta_{\mathrm{KL}}\in\Theta_m$.

\paragraph{Evaluating predictive performance} We use the log score $\mathrm{u}(\mathrm{M}(y),\tilde{y})=\log p_{\mathrm{M}}(\tilde{y}\mid y)$ as a default utility to evaluate a model's predictive performance for new observations \citep{geisser_predictive_1979,bernardo_bayesian_1994,gneiting_strictly_2007}. 
The expected log predictive density ($\mathrm{elpd}$)
is    
\begin{align}
    \mathrm{elpd}(\mathrm{M} \mid y) = \mathbb{E}_{\tilde{y}}\big[\log p_{\mathrm{M}}(\tilde{y} \mid y)\big] 
\end{align} 
where the expectation is with respect to the DGP of $\tilde{y}$ and $p_{\mathrm{M}}(\tilde{y}\mid y) = \int p_{\mathrm{M}}(\tilde{y} \mid \theta) p_{\mathrm{M}}(\theta \mid y) \text{d} \theta$ is the posterior predictive density of $\mathrm{M}$. 
Given independent test data $\{\tilde{y}_i\}_{i=1}^{n_{\text{test}}}$, we can estimate 
$\mathrm{elpd}_{\text{test}}(\mathrm{M} \mid y) = \sum_{i=1}^{n_{\text{test}}} \log p(\tilde{y}_i \mid y)$. 
The in-sample analogue is the log pointwise predictive density $\mathrm{lppd}(\mathrm{M} \mid y) = \sum_{i=1}^n \log p(y_i \mid y)$. 
When no test data are available, leave-one-out cross-validation \citep[LOO-CV, ][]{geisser_predictive_1979} can approximate out-of-sample performance via $\mathrm{elpd}_{\text{loo}} = \sum_{i=1}^{n} \log p(y_i \mid y_{-i})$, where $p(y_i \mid y_{-i}) = \int p(y_i \mid \theta) p(\theta \mid y_{-i}) \text{d} \theta$ is the LOO posterior predictive density obtained by fitting the model to the data without the $i$-th observation. 
LOO-CV with Pareto-smoothed importance sampling \citep[PSIS,][]{vehtari_practical_2017,vehtari_pareto_2024} avoids $n$ refits and provides Pareto-$\hat{k}$ values to diagnose the reliability of the estimates. 
To support comparability, we report $\mathrm{elpd}_{\text{test}}$ on the $\mathrm{elpd}_{\text{loo}}$ scale, that is, scaled with $(n/n_{\text{test}})$.   
On $\mathrm{elpd}_{\text{loo}}$ scale, a difference $\leq 4$ can indicate indistinguishable predictive performance \citep[][]{sivula_uncertainty_2025, mclatchie_efficient_2024}.

\paragraph{Overfitting} Overfitting is a commonly used concept to caution against ``fitting to the data'' without ensuring good test performance \citep{hastie_elements_2009}. 
It is difficult to assess for a single dataset and model, but, for two models $\mathrm{M}_1\subset\mathrm{M}_2$, we say $\mathrm{M}_2$ is \emph{overfitting relative to} $\mathrm{M_1}$ on the training data $y$ if
\begin{align} \mathrm{lppd}(\mathrm{M}_2 \mid y) \geq \mathrm{lppd}(\mathrm{M}_1 \mid y) \quad\text{and}\quad \mathrm{elpd}_{\text{test}}(\mathrm{M}_2 \mid y) < \mathrm{elpd}_{\text{test}}(\mathrm{M}_1 \mid y),  
\label{eq:overfitting-lppd-elpd}
\end{align}
that is, if in-sample predictions improve (or remain the same) but out-of-sample performance reduces.
When a test set is unavailable, the second inequality can be replaced by $\mathrm{elpd}_{\text{loo}}(\mathrm{M}_2\mid y)<\mathrm{elpd}_{\text{loo}}(\mathrm{M}_1\mid y)$, or by corresponding $\mathrm{elpd}$s obtained from any other CV scheme. 
If we can define a unique ordering based on model complexity, we can use the same reasoning for non-nested models.
There is no requirement that the models are true-structure models.
However, misspecified models can sometimes simultaneously overfit and underfit for different areas of observed data, for example, when fitting a linear model to data with nonlinear patterns or vice versa.

\paragraph{Best-predicting models} Following \citet[][, Section 2.2]{arlot_survey_2010}, an \emph{oracle model} is a model in $\mathcal{M}$ that, after being fitted on training data $y=\{y_i\}_{i=1}^{n}$, has the lowest expected loss on unseen data $\tilde{y}$ from the DGP.
With the log-score, this is equivalent to maximising $\mathrm{elpd}$,
\begin{align}\label{eq:oracle-def}
\mathrm{M}_{\mathrm{OR}}(y)\,\in\,\argmax_{\mathrm{M}\in\mathcal{M}}\,
\mathbb{E}_{\tilde{y}}\!\left[\mathrm{u}\!\left(\mathrm{M}(y),\tilde{y}\right)\right]
\,=\,\argmax_{\mathrm{M}\in\mathcal{M}}\,\mathrm{elpd}(\mathrm{M}\mid y).
\end{align}
For simplicity, we will just speak of \textit{the} oracle model in the following, assuming its uniqueness.
The oracle depends on $y$ and averaging additionally over $y$ gives the \emph{oracle in expectation}:
$\argmax_{\mathrm{M}\in\mathcal{M}}\mathbb{E}_{y,\tilde{y}}\!\left[\mathrm{u}(\mathrm{M}(y),\tilde{y})\right]$, which is what we report when summarising over repeated simulations in Section \ref{sec:experiments}. 
The oracle maximises expected utility within $\mathcal{M}$. 
So, for any other model $\mathrm{M}\in\mathcal{M}\setminus\{\mathrm{M}_{\mathrm{OR}}(y)\}$, we have $\mathrm{elpd}(\mathrm{M}\mid y)\le\mathrm{elpd}(\mathrm{M}_{\mathrm{OR}}(y)\mid y)$, with strict inequality whenever the maximiser is unique. 
The oracle is defined only with respect to the collection $\mathcal{M}$, which does not need to contain any true-structure models. 
Even when a true-structure model is included, the oracle does not need to be a true-structure model itself \citep[][p. 46]{arlot_survey_2010}.
Following \eqref{eq:overfitting-lppd-elpd}, models more complex than $\mathrm{M}_{\mathrm{OR}}(y)$ with lower expected predictive performance are overfitting, and models simpler than $\mathrm{M}_{\mathrm{OR}}(y)$ with lower expected predictive performance are underfitting.

\subsection{Illustrative example: Part I}\label{subsec:methods-illustrative-example-part-1} We consider a simple setup with two normal models: an intercept-only model ($\mathrm{M}_0$) with $\mu_{\mathrm{M}_0} = \alpha$ and a single-covariate model ($\mathrm{M}_1$) with $\mu_{\mathrm{M}_1} = \alpha + \beta x_i$.
Importantly, we observe similar phenomena with more complex models, often becoming more pronounced in higher-dimensional settings (see Section \ref{subsec:examples-pred-consistent-priors-illustrative-example-part-3}, \ref{subsec:ex-1-logistic}-\ref{subsec:ex-4-rct-studies}).
We consider $0.01 \leq \beta^* \leq 1$, set  $\alpha^*=0$ and $\sigma^{2*}=1$ and generate data as
\begin{align}\label{eq:illustration-simple-dgp}
       y_i \sim \normal(\alpha^{*} + \beta^* x_i, \sigma^{2*}), \text{ with} \;
       x_i \sim \normal(0, 1). 
\end{align}
Following \eqref{eq:methods-true-structure-models}, $\mathrm{M_1}$ is a true-structure model, while $\mathrm{M}_0$ is not, since we assume $\beta^* \neq 0$.
\begin{figure}[tp]
    \centering
    \includegraphics[]{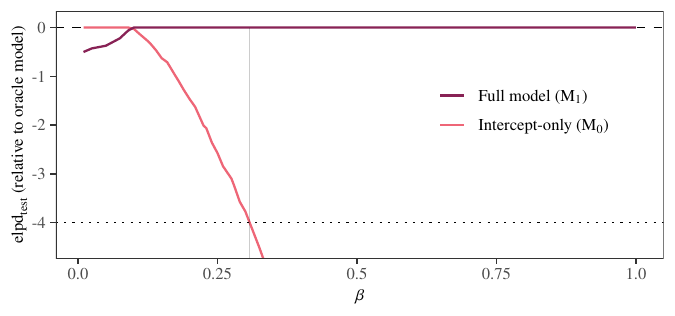}
    \vspace{-0.5cm}
    \caption{Illustrative example: Part 1. Test performance for a model with one covariate ($\mathrm{M}_1$) with prior $\beta \sim \normal(0,1)$ and the intercept-only model ($\mathrm{M}_0$), relative to the oracle model (Equation \eqref{eq:oracle-def}), averaged across $500$ repetitions and rescaled to $\mathrm{elpd}_{\text{loo}}$ scale. 
    We indicate a difference of $4$ to oracle performance with a dotted line.
    }
   \label{fig:intro-single-predictor-only-models-prior-normal-dgp-normal}
\end{figure}
Figure \ref{fig:intro-single-predictor-only-models-prior-normal-dgp-normal} compares $\mathrm{elpd}_{\text{test}}$ of $\mathrm{M_0}$ and $\mathrm{M_1}$ averaged over $500$ repetitions, relative to the oracle model \eqref{eq:oracle-def} where $\mathcal{M}=\{\mathrm{M_0}, \mathrm{M_1}\}$. 
We assume $n=100$, $n_{\text{test}}=2000$, and priors $\alpha \sim \normal(0, 2.5)$, $\sigma^2 \sim \expdist(1)$ and $\beta \sim \normal(0, 1)$.
For $\beta\gtrsim 0.3$, $\mathrm{M}_1$ is the oracle and dominates $\mathrm{M}_0$ by a margin exceeding $4$ on $\mathrm{elpd}_{\text{loo}}$ scale. 
However, for sufficiently small $\beta$, $\mathrm{M}_0$ is the oracle, even though it omits a nonzero true $\beta$, as the variance from estimating a weak $\beta$ with $n=100$ may outweigh the bias from setting it to zero. 
These observations are related to phenomena investigated by \citet{hjort_exact_1994} when comparing ``narrow'' incorrect models with higher bias to ``wide'' true-structure models with higher variability.

We can extend the example by also investigating different strategies for aggregating or selecting $\mathrm{M}_0$ or $\mathrm{M}_1$, as summarised in Table \ref{tab:illu-all-strategies}.
In particular, we consider (1) Bayes factors \citep{kass_bayes_1995}\footnote{A threshold of six for the Bayes factor corresponds to an implicit prior point probability mass of $\approx 0.86$ on the zero effect of $x$ \citep{campbell_bayes_2022}. In practice, this threshold should be chosen such that it matches the prior knowledge at hand.}, and (2) $\mathrm{elpd}_{\text{loo}}$ estimated with PSIS-LOO-CV. 
We also consider aggregating the models using either (3) Bayesian model averaging \citep[BMA; ][]{hoeting_bayesian_1999, raftery_bayesian_1997} weights defined as
    \begin{align}      
        w_k := p(\mathrm{M}_k \mid y) = \frac{p(y \mid \mathrm{M}_k)p(\mathrm{M}_k)}{\sum_{k=1}^K p(y \mid \mathrm{M}_k)p(\mathrm{M}_k)} \label{eq-bma-weight}
    \end{align}
for a model $\mathrm{M}_k \in \mathcal{M}$, or (4) stacking weights  \citep[][]{yao_using_2018}, defined via the optimisation problem\footnote{Bayesian model stacking optimises the averaged predictive distribution to provide the best predictive performance, which can produce quite different weights compared to BMA \citep[see more detailed discussion by][]{yao_using_2018, yao_bayesian_2021}{}.} 
    \begin{align}
        w^* = \argmax_{w_k} \frac{1}{n} \sum_{i = 1}^n \log \sum_{k=1}^K w_k \ p\left( y_i \mid y_{-i}, \mathrm{M}_k\right) \; \; \; \text{such that} \; \; \; w_k \geq 0, \; \sum_{k=1}^K w_k = 1 \label{eq-stack-weights}. 
    \end{align}
\begin{table}[htp]
    \centering
    {\small
    \begin{tblr}{l|l}
        Strategy & Details \\
        \hline
        (1) Bayes factor & Select $\mathrm{M}_1$ if $\text{BF}_{1,0} = \frac{p(\mathrm{M}_1 | y)}{p(\mathrm{M}_0 | y)} > 6$ \\
        (2) Predictive selection & Select model with the highest $\mathrm{elpd}_{\rm loo}$\\
        (3) Bayesian model averaging & Aggregate models with BMA weights \eqref{eq-bma-weight} \\
        (4) Bayesian model stacking & Aggregate models with stacking weights \eqref{eq-stack-weights} \\
    \end{tblr}
    }
    \caption{Illustrative example: Part 1. Different strategies for selecting or combining models.
    Additional results based on model probability, stacking weights, pseudo-BMA weights and LOO-BB weights can be found in Appendix  \ref{app:illustrative-example-part-1}.}
    \label{tab:illu-all-strategies}
\end{table}
\begin{figure}[tp]
    \centering
    \includegraphics[]{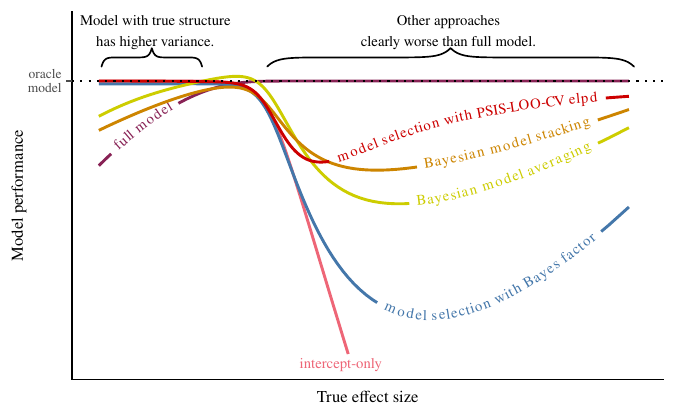}
    \vspace{-0.3cm}
    \caption{Illustrative example: Part 1. Schematic illustration of model performance relative to the oracle model for increasing true effect size in the finite-data regime, comparing the full (true-structure) model ($\mathrm{M}_1$), the intercept-only model ($\mathrm{M}_0$), Bayesian model averaging and stacking, as well as selecting a model with Bayes factor or PSIS $\mathrm{elpd}_{\text{loo}}$. 
    This is a simplified summary based on the results of simulated experiments in Appendix  \ref{app:illustrative-example-part-1}. 
    For small true effect sizes, the true-structure model can overfit, but, over the considered range, the full model has the best expected performance. 
    }\label{fig:illustration-oracle-model-comparison}
\end{figure}
For $n=\{20, 100\}$ and normal DGP, Figure \ref{fig:app-illustrative-example-dgp-normal-elpd-test-models-some-selection-aggregation-strategies} in Appendix \ref{app:illustrative-example-part-1} shows that differences in test performance to the oracle are $< 4$ for most strategies in this illustrative example.
We can still make some interesting observations, summarised in Figure \ref{fig:illustration-oracle-model-comparison}. 
In particular, we find a small range of true $\beta$ values where $\mathrm{M}_0$ and the true-structure model $\mathrm{M}_1$ perform so similar that BMA can slightly outperform either of them.
Importantly, $\mathrm{M}_1$ converges to oracle performance 
most consistently, in comparison to the strategies for selecting or aggregating models which all exhibit a dip in predictive performance after the performance of $\mathrm{M}_0$ drops. 
For example, model selection using Bayes factor $> 6$ is particularly prone to larger losses due to the implicit prior mass this strategy places on the zero effect \citep[see related demonstration on how Bayes factor requires bigger data size before switching to more complex model by][]{erven_catching_2012}.
Additional results for a Student-$t$ DGP (Figure \ref{fig:app-illustrative-example-dgp-stud-t-elpd-test-models-some-selection-aggregation-strategies}) as well as selection based on model probabilities or stacking, pseudo-BMA, and LOO-BB weights (Figures \ref{fig:app-illustrative-example-selection-with-model-prob}-\ref{fig:app-illustrative-example-selection-with-loo-bb-weights}) are in Appendix \ref{app:illustrative-example-part-1}.

\section{Properties of good priors}\label{sec:good-priors}
Given a true-structure model, a ``good'' prior puts prior mass close to the true $\theta_0$ and enables fast convergence of the posterior mass towards $\theta_0$.
There are no universally good priors, as choosing priors is always context-dependent \citep{ohagan_kendalls_2004, gelman_prior_2017}, but we can identify four desired properties
based on existing recommendations and discussions \citep{gelman_prior_2017, mikkola_prior_2024, simpson_penalising_2017}. 
\subsection{Intuition}\label{subsec:methods-intuition-prior-desiderata}
\begin{theo}[Desideratum 1: Proper priors]{desi:1}
    We want proper priors to make models (partially) generative before seeing data and enable well-defined prior predictive checks.
\end{theo}
Good priors should provide a prior predictive distribution that generates data, which we would expect to observe \citep[][]{gelman_prior_2017, gelman_regression_2020}{}{}. 
We can investigate this using prior predictive checking and prior sensitivity analysis \citep{gelman_prior_2017, gelman_regression_2020, gabry_visualization_2019, kallioinen_detecting_2023} and by visualising joint quantities like a-priori Bayesian $R^2$ to reveal unintended joint implications \citep{aguilar_generalized_2025}.
\begin{theo}[Desideratum 2: Sensible prior predictives]{desi:2}
    We want to choose priors such that the model-generated data is sensible, that is, the prior predictive distributions reflect key characteristics of our modelling problem such as scale, variability or tail behaviour. 
\end{theo}
Joint prior behaviour affects inferences \citep{gelman_prior_2017, mikkola_prior_2024}. 
When considering a collection of models, prior predictives should follow prior beliefs even as model complexity increases.
Predictively consistent priors, formally defined in Section \ref{subsec:methods-formalisation-prior-desiderata}, satisfy this logical consistency requirement, for example, by keeping the prior predictive variance constant or restricting it such that it will not explode as the model complexity grows.\footnote{In case of modelling processes with infinite variance some other variability measure would be used.}
\begin{theo}[Desideratum 3: Consistency when increasing model complexity]{desi:3}
    When adding model components or otherwise increasing model complexity, the prior predictives should stay sensible and aligned with prior beliefs. 
\end{theo}
Most often, we can only keep some properties like the first two moments of the predictive distribution similar.
For nested models, \citet{gelman_discussion_1995} recommend using larger models with priors favouring ``the region of the parameter space near the smaller model'' \citep[][p.131]{gelman_discussion_1995}. 
This is also the main principle for penalised complexity priors \citep{simpson_penalising_2017} and motivates a fourth desideratum. 
\begin{theo}[Desideratum 4: Fall back to parsimony]{desi:4} 
    For a more complex version of a simpler model, a good prior should be centred and place most mass on the simpler (base) model, especially when the predictive distribution changes with added model complexity.
\end{theo}
Parsimony is complementary to, and weaker than, second-order predictive consistency (see Section \ref{subsec:methods-formalisation-prior-desiderata}), but still a useful principle for setting priors on less intuitive parameters such as the degrees of freedom of the Student-$t$-distribution. 
While predictive consistency often implies parsimony, exceptions are discussed in Lemma \ref{lem:pc-pp-relationship}.

\subsection{Formalisation}\label{subsec:methods-formalisation-prior-desiderata}
We now formalise Desiderata \ref{desi:3} and \ref{desi:4}.
Consider a sequence of nested models $\{\mathrm{M}_d\}_{d \in \mathbb{N}}$ with increasing complexity.
The complexity index $d$ may refer to number of predictors, varying-effect terms, or group levels.
Each model $\mathrm{M}_d$ has the following hierarchical structure:
\begin{enumerate}[nosep]
    \item \emph{Hyperparameters:} $(\zeta_d, \psi) \sim p(\zeta_d, \psi)$, where $\zeta_d$ are prior hyperparameters for $\theta_d$ (e.g., coefficient scales, group-level variances), and $\psi$ collects likelihood parameters outside any latent hierarchy on $\theta_d$ (e.g., residual variance, dispersion, degrees of freedom).
    \item \emph{Model parameters:} $\theta_d \mid \zeta_d \sim p(\theta_d \mid \zeta_d)$, using the first $d$ model components, including latent group-level effects.
    \item \emph{Systematic component:} $\eta_d := f_d(x_d; \theta_d)$, a deterministic function of $\theta_d$ and inputs $x_d$, including group-level design matrices.
    \item \emph{Observation model:} $y \mid \eta_d, \psi \sim p(y \mid \eta_d, \psi)$, where only $\eta_d$ depends on $\theta_d$ and $\psi$ are shared across all models.
\end{enumerate}
\vspace{0.15cm}
The joint prior is $p_d(\theta_d, \zeta_d, \psi) = p(\theta_d \mid \zeta_d)\,p(\zeta_d, \psi)$.
Thus, the $d$-indexed hyperparameters $\zeta_d$ affect $y$ only through the path $\zeta_d \to \theta_d \to \eta_d \to y$. 
Let $x_d \sim p_d(x_d)$ be a single predictor observation from a reference distribution such that $x_d$ is independent of parameters ($x_d \perp (\theta_d, \zeta_d, \psi) \ \forall d$) with $\mathbb{E}[x_d]=0$ and unit marginal variance.
The \emph{induced signal distribution} $q_d(\eta_d)$ is the law of the scalar systematic component $\eta_d = f_d(x_d;\theta_d)$, obtained as the pushforward of the joint prior $p(x_d,\theta_d,\zeta_d)$ under $f_d$,
\begin{equation}
    q_d(\eta_d) = \int \mathbb{I}(\eta_d = f_d(x_d;\theta_d)) \, p(x_d) \, p(\theta_d, \zeta_d) \, d x_d \, d \theta_d \, d \zeta_d,
\end{equation}
where $p_d(\theta_d, \zeta_d)$ is the marginal prior of $(\theta_d, \zeta_d)$ obtained from $p_d(\theta_d, \zeta_d, \psi)$ by integrating out $\psi$.
\begin{definition}[Second-order predictive consistency]\label{def:second-order-pc}
A prior sequence $\{p_d\}_{d \in \mathbb{N}}$ is \emph{second-order predictively consistent} if
\begin{enumerate}[(i)]
    \item the prior predictive variance of the systematic component is globally bounded:
    \begin{equation}
        \sup_{d \in \mathbb{N}}\; \Var_{q_d}(\eta_d) \leq C \label{eq:pc-variance-bound}
    \end{equation}
    for some constant $C \in \left[0, \infty\right)$, and
    \item the variance converges to a finite asymptotic limit, $L \in \left[0,\infty\right)$:
    \begin{equation}
        \lim_{d \to \infty} \Var_{q_d}(\eta_d) = L. \label{eq:pc-variance-limit}
    \end{equation}
\end{enumerate}
\end{definition}
Definition \ref{def:second-order-pc} uses the distribution over $\eta_d$ rather than the prior predictive over $y$, since observation models with infinite or undefined marginal variance (e.g., Student-$t_{\nu}$ with $\nu \leq 2$) violate the bound on global boundedness regardless of the prior, while for bounded responses (e.g., binary data) the prior predictive variance is trivially bounded and finite. 
The supremum condition \eqref{eq:pc-variance-bound} allows for non-monotone variance sequences and ensures global boundedness. 
Non-monotonicity can arise when the signal aggregates many components, for example, in Bayesian neural networks \citep{lampinen_bayesian_2001}, or when parameters are negatively correlated.
Note that $C$ may be potentially zero for point priors.
The limit condition \eqref{eq:pc-variance-limit} ensures that the variance stabilises to some $L$, which, due to potential non-monotony, may differ from $C$.\footnote{We acknowledge that the limit condition \eqref{eq:pc-variance-limit} implies a finite global variance. 
However, we would like to highlight explicitly that the variance sequence should be well-behaved and that the global and limiting variance may be two different values.}
Together, this formalises Desideratum \ref{desi:3}: the prior predictive distribution stays sensible as model complexity grows precisely when the signal variance remains bounded and converges.

In Desideratum \ref{desi:4}, parsimony is complementary to, and usually weaker than, second-order predictive consistency: it does not require bounded variance, but controls how much the induced signal distribution changes across model sizes.
A natural formalisation is to require that the induced signal distributions do not diverge from one another as $d\rightarrow \infty$ and that their pairwise divergences accumulate to a finite sum.
\begin{definition}[Predictive parsimony]\label{def:pred-parsimony}
Assume $q_{d+1} \ll q_d$ for all $d \in \mathbb{N}$. 
A prior sequence $\{p_d\}_{d \in \mathbb{N}}$ satisfies \emph{predictive parsimony} if the KL divergence between successive induced signal distributions is summable:
\begin{equation}
    \sum_{d=1}^{\infty} D_{\mathrm{KL}}\!\left(q_{d+1} \,\|\, q_d\right) < \infty. \label{eq:kl-parsimony}
\end{equation}
\end{definition}
The asymmetric direction $D_{\mathrm{KL}}(q_{d+1}\|q_d)$ encodes a preference for the simpler model by penalising mass placed by $q_{d+1}$ where $q_d$ has zero support.
The absolute-continuity assumption $q_{d+1} \ll q_d$ ensures each summand is finite.
For Gaussian signal families $q_d = \normal(0, V_d)$, the KL divergence between successive models conditional on $V$ is
\begin{equation}
    D_{\mathrm{KL}}(q_{d+1} \| q_d) = \frac{1}{2}\left[\frac{V_{d+1}}{V_d} - 1 - \log\frac{V_{d+1}}{V_d}\right], \label{eq:kl-gaussian}
\end{equation}
which depends only on the variance ratio $V_{d+1}/V_d$.

\begin{lemma}[Relationship between predictive consistency and parsimony]\label{lem:pc-pp-relationship}
Let $q_d = \normal(0, V_d)$ with $V_d > 0$ for all $d$. Write $\mathrm{S}$ for sequences satisfying second-order predictive consistency (Definition \ref{def:second-order-pc}) and $\mathrm{F}$ for those satisfying predictive parsimony (Definition \ref{def:pred-parsimony}). Then:
\begin{enumerate}[(i)]
    \item $\mathrm{S} \not\subseteq \mathrm{F}$: bounded and convergent variance need not imply parsimony. For example, the oscillating but convergent series $|V_{d+1} - V_d| = \Omega(d^{-1/2})$ does not fulfil parsimony.
    \item $\mathrm{F} \not\subseteq \mathrm{S}$: parsimony does not necessarily imply bounded variance. For example, the independent normal prior ($V_d = d\overline{\tau}^2$) is in $\mathrm{F}$ but not $\mathrm{S}$.
    \item $\mathrm{S} \cap \mathrm{F}$: a prior sequence is in both $\mathrm{S}$ and $\mathrm{F}$ if and only if $V_d$ is bounded, convergent, $\inf_d V_d > 0$, and $\sum_{d=1}^{\infty}(V_{d+1} - V_d)^2 < \infty$. Monotone convergence to a strictly positive limit is a sufficient condition.
\end{enumerate}
\end{lemma}
The proof is given in Appendix \ref{app:non-monotonic-parsimony}.
A key consequence is that predictive parsimony alone is too weak to exclude independent normal priors: each added predictor changes the signal distribution only marginally ($\overline{\tau}^2$ out of $d\overline{\tau}^2$, giving $D_{\mathrm{KL}}(q_{d+1}\|q_d) \approx 1/(4d^2)$, 
summable), yet the variance increments are multiples of $d$, so $V_d = d\overline{\tau}^2$ diverges for any $\overline{\tau}>0$.

\subsubsection{Predictive consistency and parsimony for global-local priors}
For the linear model \eqref{eq-regression-setup}, the widely used class of global-local scale-mixture priors \citep[][]{polson_shrink_2011} takes the conditional form $\theta_j \mid \zeta_d \sim \normal(0, \tau^2 \phi_j)$, where $\zeta_d = (\tau^2, \phi_1, \ldots, \phi_d)$ contains the global scale $\tau^2$ and local weights $\phi_j \geq 0$. 
Typically, heavy-tailed priors are chosen for the local weights and global scale to simultaneously shrink noise heavily to zero and only weakly regularise signals.
For centred priors and a standardised predictor reference distribution, the decomposition under global local priors becomes:
\begin{equation}
    \Var_{q_d}(\eta_d) = \mathbb{E}_{p_d}\!\left[\tau^2 \sum_{j=1}^{d} \phi_j\right]. \label{eq:ppv-global-local}
\end{equation}
This yields a necessary and sufficient condition for predictive consistency:

\begin{corollary}[Characterisation of predictive consistency for global-local priors]\label{lem:gl-characterisation}
Consider the linear model \eqref{eq-regression-setup} with standardised predictor reference distribution and a centred global-local prior with $\theta_j \mid \zeta_d \sim \normal(0, \tau^2 \phi_j)$.
The prior sequence is second-order predictively consistent (Definition  \ref{def:second-order-pc}) if and only if
\begin{equation}
    \sup_{d \in \mathbb{N}} \;\; \mathbb{E}_{p_d}\!\left[\tau^2 \sum_{j=1}^{d} \phi_j\right] < \infty \quad \text{and} \quad \lim_{d \to \infty} \mathbb{E}_{p_d}\!\left[\tau^2 \sum_{j=1}^{d} \phi_j\right] \;\text{exists}. \label{eq:gl-condition}
\end{equation}
\end{corollary}
\paragraph{Proof.} Both conditions are direct translations of \eqref{eq:pc-variance-bound} and \eqref{eq:pc-variance-limit} via the identity \eqref{eq:ppv-global-local}. \hfill $\square$

For any global-local prior with local weights satisfying $\sum_j \phi_j = 1$, Corollary \ref{lem:gl-characterisation} reduces second order predictive consistency to requiring $\mathbb{E}_{p_d}[\tau^2] < \infty$, independently of model size.
Some global-local priors respect this condition, others can be set up to be predictively consistent, see Appendix \ref{app:section-examples-pred-consistent-priors-other-global-local} for more details. 

\paragraph{Parsimony for global-local priors.} When the local weights are constrained to the simplex, that is $\sum_j\phi_j=1$,
and the prior on $\tau^2$ does not depend on $d$, the conditional variance $\Var(\eta_d \mid \tau^2,\phi) = \tau^2$ is independent of $d$. 
With a Gaussian conditional prior on $\theta_d$, the variance sequence $V_d$ entering Lemma \ref{lem:pc-pp-relationship} is constant in $d$, so $D_{\mathrm{KL}}(q_{d+1}\|q_d) = 0$ in the variance-only view. 
The full mixture distributions $q_d$ still depend weakly on $d$ through the law of $\tau^2 \sum_j \phi_j x_{d,j}^2$, but this dependence vanishes 
as $d \to \infty$ (e.g., by concentration of $\sum_j \phi_j x_{d,j}^2$ around its mean when $\phi_j = 1/d$). 
For unconstrained global-local priors like the horseshoe prior \citep{carvalho_horseshoe_2010}, parsimony additionally requires square-summable relative increments of $S_d := \mathbb{E}_{p_d}[\tau^2 \sum_j \phi_j]$, which must be verified case by case.

\paragraph{Implications of Desiderata~3-4 on shrinkage}
In the context of global-local priors, shrinkage factors $\kappa_j = (1 + \sigma^{-2}\tau^2\phi_j)^{-1}$ are routinely analysed in order to characterise the properties of the chosen hyper-priors on the implied distribution of effective non-zero coefficients defined as $m_{\text{eff}, d} = \sum_{j=1}^d (1 - \kappa_j)$ \citep{piironen_sparsity_2017}\footnote{Whenever the observation model is not normal, one may use the approximating variance $\tilde{\sigma}^2$ as defined in \citet{piironen_sparsity_2017}}. 
Because $\frac{x}{\sigma^2 + x} \le \frac{x}{\sigma^2}$ for any $x \ge 0$, we can establish a strict upper bound on the expected effective number of parameters, conditional on $\sigma^2$:
\begin{equation}
\mathbb{E}_{p_d}[m_{\text{eff}, d}] = \mathbb{E}_{p_d} \left[ \sum_{j=1}^d \frac{\tau^2 \phi_j}{\sigma^2 + \tau^2 \phi_j} \right] \le \mathbb{E}_{p_d} \left[ \sum_{j=1}^d \frac{\tau^2 \phi_j}{\sigma^2} \right] = \frac{1}{\sigma^2} \mathbb{E}_{p_d} \left[ \tau^2 \sum_{j=1}^d \phi_j \right] \leq C/\sigma^2.
\end{equation}
However, predictive parsimony, alone does not guarantee a bound on the effective number of non-zero coefficients, since under the isotropic normal prior with fixed variance $\overline{\tau}^2$: 
\begin{equation}
m_{\text{eff}, d} = \sum_{j=1}^d \left( \frac{\bar{\tau}^2}{\sigma^2 + \bar{\tau}^2} \right) = d \left( \frac{\bar{\tau}^2}{\sigma^2 + \bar{\tau}^2} \right).
\end{equation}
As $d \to \infty$, the effective model size $m_{\text{eff}, d} \to \infty$.

\paragraph{Generality beyond additive normal models.}
Since Definitions \ref{def:second-order-pc} and \ref{def:pred-parsimony} are stated in terms of the induced signal distribution $q_d$ and its variance $\Var_{q_d}(\eta_d)$, they apply not only to additive normal regression but to any model for which this quantity is well-defined, including models where the role of $\theta_d$ may be taken by an infinite-dimensional latent stochastic process, such as in a Gaussian process.
In such cases, the definitions still provide a verifiable criterion and can be evaluated via prior predictive simulations.

\subsection{Illustrative example: Part 2}\label{subsec:good-priors-illustrative-example-part-2}
We extend the normal model from Section \ref{subsec:methods-illustrative-example-part-1} and write it more generally as
\begin{equation}
    \begin{aligned}
        y \sim  \normal\left(f(x_d; \theta_d), \sigma^2\right) \; \; \text{with} \; \; \theta_d, \sigma^2  \sim p_d(\theta_d, \sigma^2),
    \end{aligned}\label{eq-regression-setup}
\end{equation}
where $f(x_d; \theta_d) = x_d \, \theta_d$ is the linear predictor over the first $d$ covariates. In the regression examples below, $\theta_d$ collects the regression coefficients written $\beta$, with $j$-th entry $\beta_j$, so that $\theta_d \equiv \beta_{1:d}$.
We assume the moments $\mu_{\theta,d} = \mathbb{E}_{p_d}[\theta_d]$ and $\Sigma_{\theta,d} = \mathrm{Cov}_{p_d}(\theta_d)$ exist and are finite.
Denoting $\Sigma_{x,d} = \mathrm{Cov}_{p_d}(x_d)$, the signal variance is then given by
\begin{equation}
    \Var_{q_d}(\eta_d) = \mathrm{tr}\!\left(\Sigma_{x,d}\Sigma_{\theta,d}\right) + \mu_{\theta,d}^T \Sigma_{x,d} \mu_{\theta,d}. \label{eq:ppv-decomposition}
\end{equation}
If, in addition, the prior is centred so that $\mu_{\theta,d}=0$ and either $\Sigma_{x,d}$ or $\Sigma_{\theta,d}$ are diagonal, then \eqref{eq:ppv-decomposition} reduces to
\begin{equation}
    \Var_{q_d}(\eta_d) = \mathrm{tr}\!\left(\Sigma_{\theta,d}\right) = \sum_{j=1}^{d} \Var_{p_d}(\theta_j). \label{eq:ppv-standardised}
\end{equation}
For independent normal priors $\theta_d \sim \normal(0, \overline{\tau}^2)$, with fixed value $\overline{\tau}^2 > 0$ (i.e., $\zeta_d$ degenerate at $\overline{\tau}^2$), we have $\Var_{q_d}(\eta_d) = d\overline{\tau}^2$, which diverges as $d \to \infty$, thus violating the predictive consistency requirement \eqref{eq:pc-variance-bound}.

\section{Examples of predictively consistent priors}\label{sec:examples-pred-consistent-priors}
We focus on two examples for predictively consistent priors, introduced below. 
Other ways to specify predictively consistent priors were not investigated here but could provide interesting insights (see, e.g., examples listed in Section \ref{subsec:intro-previous-work} 
and Corollary \ref{lem:gl-characterisation} in Section \ref{subsec:methods-formalisation-prior-desiderata}). 

\subsection{$R^2$-induced Dirichlet decomposition (R2D2) priors}\label{subsubsec:methods-r2d2}
The experiments in Section \ref{subsec:examples-pred-consistent-priors-illustrative-example-part-3} and \ref{subsec:ex-1-logistic} to \ref{subsec:ex-2-forward-search} use R2D2 priors \citep{zhang_bayesian_2022} that set a prior on $R^2$, the proportion of variance in the outcome explained by the model.
For a linear regression model $y_i = \alpha + x_i^T \beta + \epsilon_i$ with $i = 1, \cdots n$ observations, $\epsilon_i \sim \normal(0, \sigma^2)$, and coefficient vector $\beta = (\beta_1, \cdots, \beta_p)^T$, following \citet{gelman_r-squared_2019}, \textit{realised} Bayesian $R^2$ is given by
\begin{align}
    R^2 &:= \frac{\Var\left( x^T \beta \right)}{\Var\left( x^T \beta \right) + \sigma^2}.\label{eq:def-bayesian-r2-gelman}
\end{align} 
When marginalising out $X$, this reduces to $\tau^2 / (\tau^2 + \sigma^2)$ and the R2D2 prior is then defined as:  
\begin{equation}
   \begin{aligned}
        R^2 &\sim \betadist(\mu_{R^2}, \varphi_{R^2}) \text{ and } 
        \tau^2 = \sigma^2\frac{R^2}{1 - R^2} \\
        \phi_j &\sim \Dirichlet(a_1, \ldots, a_p) \to \sum_{j=1}^p \phi_j = 1 \\
        \beta_j &\sim \normal(0, \lambda_j^2) \text{ with } \lambda_j^2 = \tau^2 \phi_j, 
    \end{aligned}\label{eq:r2d2-prior}
\end{equation}
following the investigation of normal kernels by \citet{aguilar_intuitive_2023}. 
A Beta prior on the model's $R^2$ implies a Beta prime prior on the global scale $\tau^2$.
Dirichlet weights for the local-scale components $\phi_j$ allocate a share of global variance to each coefficient via $\lambda_j^2=\tau^2\phi_j$.\footnote{The local variances $\lambda^2_j$ should be rescaled to $\tilde{\lambda}^2_j = \lambda^2_j\sfrac{\sigma^2_y}{\sigma^2_{x_{j}}}$ or all variables standardised to ensure that the scale of $R^2$ is independent of the outcome and the covariates \citep{aguilar_intuitive_2023}.}

The finite-moment condition $\mathbb{E}_{p_d}[\tau^2] = \mathbb{E}_{p_d}[\sigma^2 R^2/(1-R^2)] < \infty$ from Lemma \ref{lem:gl-characterisation} therefore becomes the key requirement for predictive consistency.
This holds whenever $(1-\mu_{R^2})\varphi_{R^2} > 1$, that is, when the R2D2 prior does not place too much mass near one\footnote{Note that this condition on $\tau^2$ excludes the bounded influence property which is achieved for the R2D2 prior when $(1-\mu_{R^2})\varphi_{R^2} \leq 0.5$.}, and $\mathbb{E}_{p_d}[\sigma^2] < \infty$.
Hence, setting a prior directly on $R^2$ is not merely a convenient elicitation device; it is the natural parameterisation for achieving second-order predictive consistency.

\subsection{Illustrative example: Part 3}\label{subsec:examples-pred-consistent-priors-illustrative-example-part-3}
To illustrate the different implications of prior choices, we extend the previous example in Section \ref{subsec:methods-illustrative-example-part-1} and \ref{subsec:good-priors-illustrative-example-part-2} to a collection of models with up to $p=30$ predictors, added in fixed order. 
Section \ref{subsec:ex-2-forward-search} considers the (more realistic) case of order based on model selection.
We compare independent normal priors $\beta_j \sim \normal(0, 1), \; \; \forall j=1, \cdots, p$ and an R2D2 prior putting most mass on $R^2$ values close to zero with $\mu_{R^2} = \sfrac{1}{3}$, $\varphi_{R^2} = 3$ and Dirichlet concentration parameters $a_{1:p} = (1, \ldots, 1)$. 
The shape across $R^2$ values can be adjusted through the configuration of the R2D2 prior.\footnote{To support prior elicitation for $R^2$, one can, for example, visualise the R2D2 prior and its components with this useful web-app \href{https://solviro.shinyapps.io/R2D2_shiny/}{https://solviro.shinyapps.io/R2D2\_shiny/} by Sölvi Rögnvaldsson.} 
\begin{figure}[tp]
    \centering
    \includegraphics[]{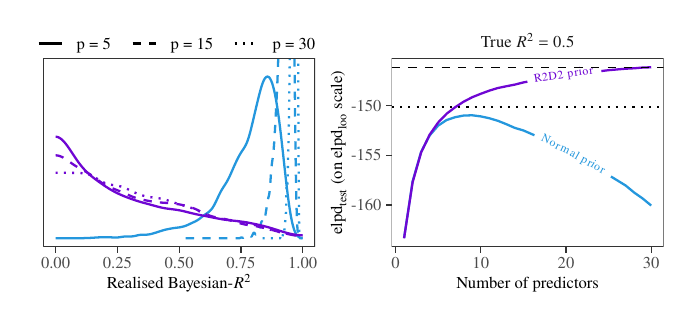}
    \caption{Illustrative example: Part 3. 
    Left: realised prior Bayesian-$R^2$ for $p \in \{5, 15, 30\}$ covariates with $\sigma^2=1$, 
    $\frac{1}{n} X^TX = I_p$ and using $ \Var\left(X^T \beta \mid \beta\right)=||\beta||^2$ with $||\beta||^2 \sim \chi^2(p)$ for normal priors.
    As $p$ increases, the R2D2 prior keeps $R^2$ stable, while it concentrates near 1 with independent normal priors. 
    Right: $\mathrm{elpd}_{\text{test}}$ (with $n=100$, $\rho=0.5$, $R^2=0.5$, $n_{\text{test}}=2000$, averaged over $500$ repetitions), on $\mathrm{elpd}_{\text{loo}}$ scale. 
    A difference of 4 (dotted line) to the best-performing model (dashed line) indicates substantially different predictive performance.  
    Models with R2D2 prior perform best across $p$ and, unlike independent normal priors, allow to increase model size up to the true $p=30$ without a drop in test performance.}
    \label{fig:illustrative-example-part-3-implied-prior-r2-elpd-test-rho-05-r2-05}
\end{figure}
Figure \ref{fig:illustrative-example-part-3-implied-prior-r2-elpd-test-rho-05-r2-05} (left) visualises implied prior $R^2$ for both prior choices and increasing $p \in \{5, 15, 30\}$, illustrating that with higher number of predictors, independent normal priors imply a prior $R^2$ increasingly concentrated at $1$. 
Explaining close to 100\% of the variation in the outcome is in many cases a largely unrealistic assumption.
Relatedly, \citet{efron_discussion_1973} points out that seemingly weakly informative independent normal priors can lead to unexpectedly informative joint priors.
In contrast, the implied prior $R^2$ distribution remains stable with R2D2 prior.  
To add to this, Figure \ref{fig:app-ex-1-adding-covariates-prior-posterior-r2} in Appendix  \ref{app:illustrative-example-part-3-prior-posterior-r2} illustrates that putting most prior mass close to $R^2=1$ can also pull the posterior implied $R^2$ away from the truth for independent normal priors.

Figure \ref{fig:illustrative-example-part-3-implied-prior-r2-elpd-test-rho-05-r2-05} (right) also shows test performance under the two different prior choices averaged over $500$ repetitions for model sizes up to $p=30$. 
We generate training data $\{x_i, y_i\}_{i=1}^n$ with $n=100$ and independent test data with $n_{\text{test}} = 2000$. 
We assume $x_i \sim \normal(0, \Sigma), \; x_i \in \mathbb{R}^{p}$ and $\Sigma\in \mathbb{R}^{p\times p}$ with unit variance and correlation $\rho=0.5$, that is, the true model is the encompassing model.
Following the approach detailed in Appendix  \ref{app:illustrative-example-part-3-data-generation}, we set the true values for the coefficients $\beta_j$ such that the true $R^2 = 0.5$. 
Both priors are different from the DGP.
The chosen R2D2 prior allows to increase model complexity up to the true model ($p=30$) without declining test performance while test performance drops when using independent normal priors. 
Figure \ref{fig:app-illustrative-example-adding-covariates-elpd-paths-r2-0-05-08-rho-0-p-30-100} to \ref{fig:app-illustrative-example-adding-covariates-elpd-paths-r2-0-05-08-rho-09-p-30-100} in Appendix  \ref{app:examples-pred-consistent-priors} show largely similar patterns for $\rho\in \{0, 0.5, 0.9\}$, true $R^2 \in \{0, 0.5, 0.8\}$ and $p\in \{30,100\}$.
Only for uncorrelated predictors, true $R^2 \in \{0.5, 0.8\}$ and $p=30$, the two priors perform similarly.

\subsection{Priors on function space}\label{subsec:methods-function-priors} 
Function-space priors provide yet another framework for constructing predictively consistent priors. 
They allow eliciting prior knowledge directly on properties of functions, like smoothness, periodicity, or stationarity, rather than individual coefficients. 
When these are chosen explicitly to bound the function's marginal variance, they ensure the total signal variance remains controlled as more complexity is added.
In Section \ref{subsec:ex-3-nonlinear}, we use Gaussian process priors which are among the most commonly used function-space priors \citep[e.g.,][]{rasmussen_gaussian_2006}.

The Hilbert space approximate Gaussian process \citep[HSGP, ][]{riutort-mayol_practical_2022, solin_hilbert_2020} used in Section \ref{subsec:ex-3-nonlinear} approximates a Gaussian process by a truncated basis expansion
\begin{align}
   \eta_k = f_k(x) \approx \sum_{j=1}^{k} \left(S_{\theta}\!\left(\sqrt{\lambda_j}\right)\right)^{1/2} \phi_j(x) \,\beta_j,
\end{align}
with $\beta_j \sim \normal(0, 1)$, 
eigenvectors $\phi_j$ and eigenvalues $\lambda_j$ of the covariance operator of the kernel, spectral density $S_{\theta}$, and length-scale $\ell$ and marginal scale of the GP $\sigma^2_{\mathrm{GP}}$ as hyperparameters $\theta = (\ell, \sigma^2_{\mathrm{GP}})$. 
For fixed $\theta$, $\Var\big(f_k(x)\mid\theta\big)=\sum_{j=1}^{k} S_{\theta}\big(\sqrt{\lambda_j}\big)\,\phi_j^2(
x)$ converges to $\kappa_\theta(x,x)=\sigma^2_{\mathrm{GP}}$ as $k\to\infty$.
Integrating over any proper prior on $\theta$ 
and averaging over any reference predictor distribution $p(x)$ gives $\Var_{q_k}(\eta_k)\le \mathbb E[\sigma^2_{\mathrm{GP}}]$ for all $k$, so HSGPs satisfy Definition \ref{def:second-order-pc} with $C=\mathbb E[\sigma^2_{\mathrm{GP}}]$.
In contrast, for raw polynomials $f_k(x)=\sum_{j=0}^{k}\beta_j x^j$ with independent $\beta_j\sim\mathcal \normal(0,\overline{\tau}^2)$ and fixed $\overline{\tau}^2$,
\begin{align}
    \Var_{q_k}\big(\eta_k\big)=\overline{\tau}^2 \sum_{j=0}^{k}\mathbb E[x^{2j}].
\end{align}
If $p(x)$ places non–negligible mass near or beyond $|x|=1$, the partial sums $\sum_{j=0}^{k}\mathbb E[x^{2j}]$ grow with $k$.
For instance, they diverge when $\Pr(|x|\ge 1+\delta)>0$, and even for $x\sim\mathrm{Uniform}[-1,1]$.
\footnote{Note that if $|x|\le c<1$ almost surely, the series converges.}
This violates Definition \ref{def:second-order-pc} and helps to explain the decline in test performance observed in our experiments in Section \ref{subsec:ex-3-nonlinear}.

More generally, for any model $f_d(x)=\sum_{j=1}^{d} w_j(x)\,\beta_j$ with $\beta\sim\mathcal \normal(0,\Sigma_d)$, 
second-order predictive consistency holds if $\sup_{d} \mathbb{E}_{x} \big[w_d(x)^T \Sigma_d\,w_d(x)\big] < \infty $.
When $\Sigma_d=\mathrm{diag}(\sigma_j^2)$, this reduces to $\sup_{d}\sum_{j=1}^{d}\sigma_j^2\,\mathbb E\big[w_j(x)^2\big]<\infty$.
This includes the linear model with $w_j(x) = x_j$ and the bound from $R^2$, the HSGP with $w_j(x) = S_\theta(\sqrt{\lambda_j})^{1/2} \phi_j(x)$ and the bound from $\kappa_{\theta}(x,x)$, and basis-function expansions generally.

\section{Experiments}\label{sec:experiments}
We investigate predictive performance when model complexity increases, comparing independent priors to predictively consistent priors under different true DGPs. 
We focus on modelling scenarios that are typically considered prone to overfitting: 
\begin{enumerate}[nosep, label = \textbullet]
    \item Adding covariates in logistic regression (Section \ref{subsec:ex-1-logistic}) in a dense regime, that is, the true DGP includes all covariates,
    \item Performing forward stepwise selection in a sparse regime, that is, only some covariates have non-zero weights in the true DGP (Section \ref{subsec:ex-2-forward-search}), or
    \item Increasing degree or number of basis functions in nonlinear models where the true DGP is not directly represented by any of the considered models (Section \ref{subsec:ex-3-nonlinear}).
\end{enumerate} 
Additionally, we explore (joint) prior choices for regression models with treatment effect (Section \ref{subsec:ex-4-rct-studies}). 

\subsection{Experiment 1: Adding covariates in logistic regression}\label{subsec:ex-1-logistic}
The setup is similar to the linear regression experiment in Section \ref{subsec:examples-pred-consistent-priors-illustrative-example-part-3}: we fix the DGP to yield a certain true $R^2$ and investigate how test performance changes with the number of predictors and the prior choice.
Many notions of $R^2$ exist for non-normal models \citep[][]{piepho_coefficient_2019}. We use the Binomial adaptation of the Bayesian-$R^2$ as defined by \citet[][]{gelman_r-squared_2019}:

\begin{equation}\label{eq:logistic_bayes_r2}
    \begin{aligned}
        \pi_i &= \frac{1}{1 + \exp(-\eta_{d,i})}, \quad \pi = (\pi_1,\dots,\pi_n)^T, \\
        R^2 &= \frac{\Var(\pi)}{\Var(\pi) + \frac{1}{n} \sum_{i=1}^n \pi_i (1 - \pi_i)},
    \end{aligned}
\end{equation}
where $\eta_{d,i} = \sum_{j=1}^{d} x_{ij}\beta_j$ uses the first $d$ covariates and $x_{ij}$ is the $j$-th covariate of observation $i$.
The predictor term is replaced by the predicted probabilities $\pi_i$, and $\sigma^2$ by the average variance of event occurrence under the Binomial model, conditional on $\beta$. 
In the following, we refer to \eqref{eq:logistic_bayes_r2} as the Bayesian-$R^2$ for the Binomial model. 

Since the logistic link prevents an analytical mapping between a prior on Bayesian-$R^2$ to the variance of the regression coefficients, we place a beta prior on the pseudo-$R^2$ space instead. 
This is based on the observation that many GLMs allow convenient local normal approximations with pseudo-means and pseudo-variances,
building on previous work on pseudo-variance adjustments 
\citep{neal_bayesian_2012, piironen_sparsity_2017}.

To derive the pseudo-variance for the Binomial model, assume that the log likelihood function, $L_i$ for observation $i$ can be written as a member of the exponential family of distributions 
\begin{equation}
    L_i\left(y_i \vert \theta, \phi\right) = \frac{y_i \theta - B(\theta)}{A(\phi)} - C(y_i,\phi),
\end{equation}
where $\theta,\phi$ refer to the natural and dispersion parameter, respectively, for some functions $A(\cdot)$, $B(\cdot)$, $C(\cdot)$, specific to the GLM. 
Denote by $B'(\theta)$ the derivative of $B$ with respect to $\theta$. 
Approximating the Binomial model locally by a normal $z_i \sim \normal(\tilde{f}_i,\tilde{\sigma}^2)$, \citet{piironen_sparsity_2017} show that data-based moments can be used to express the second-order Taylor approximation around the posterior mode by
\begin{equation}
    \tilde{\sigma}^2 = [\mu(1-\mu)]^{-1}, \quad \tilde{f}_i  = \eta_d - \frac{y_i-\mu}{\mu(1-\mu)},
\end{equation}
where $\mathbb{E}[y] = \mu = B'(\theta)$ . 
In a prior specification context, $\mu$ can be set to a prior guess for the mean response (e.g., $\mu = 0.5$ for balanced classes, giving $\tilde{\sigma}^2 = 4$) or estimated as the mean of the response as suggested by \citet[][]{piironen_sparsity_2017}. Conditional on  $\frac{y_i-\mu}{\mu(1-\mu)}$, the variance of the pseudo-response $\tilde{f}_i$ reduces to $\Var(\eta_d)$, rendering a pseudo-$R^2$, which we denote as $\tilde{R}^2$:
\begin{equation} \label{eq:pseudo-r2-def}
    \tilde{R}^2 = \frac{\Var(\eta_{d,i})}{\Var(\eta_{d,i}) + \tilde{\sigma}^2}.
\end{equation}
Similar to the normal model above, one can set a beta prior on $\tilde{R}^2$ defined as the proportion of variance of $z_i$ explained by $\eta_d$.
Although the pseudo-$R^2$ remains only an approximation to the Bayesian-$R^2$, one may also view this definition as an $R^2$ measure on the latent space of the GLM with a variance adjustment given by the pseudo-variance of the GLM.\footnote{Relatedly, \citet{yanchenko_r2d2_2025} propose an $R^2$-based prior by finding an exact push-forward of Bayesian-$R^2\rightarrow \Var(\beta_j)$, conditional on hyper-parameters. A drawback of this approach is that it cannot be analytically found for many GLMs, including the Binomial model.} Under the standardised design, the expected variance of the latent signal is $\Var_{q_d}(\eta_d) = \mathbb{E}_{p_d}\left[ \sum_{j=1}^d \tau^2 \phi_j \tilde{\sigma}^2 \right] = \tilde{\sigma}^2 \mathbb{E}_{p_d}[\tau^2]$, and is predictively consistent (Definition \ref{def:second-order-pc}) under the same conditions as for the normal observation model. Note that while the observation space of the Binomial model is bounded, the latent signal $\eta_d \in \mathbb{R}$ is not. Therefore, second-order predictive consistency (Definition \ref{def:second-order-pc}) does not hold trivially. 

\paragraph{Setup} We simulate $n=100$ observations with $p=30$ candidate predictors drawn from a multivariate normal distribution with constant correlation structure ($\rho = 0.5$). We set the true coefficients $\beta^*$ to generate an underlying data-generating process with a target Bayesian-$R^2 \in \{0, 0.5, 0.8\}$. We then fit a sequence of nested logistic regression models, increasing the number of included covariates from $1$ to $30$. Models are evaluated on $2000$ hold-out samples, averaged over $300$ independent replications.

\paragraph{Prior choices} The pseudo-$R^2$ prior for coefficients $\beta$ can be set similarly as above with $\tilde{R}^2 \sim \betadist(\mu_{R^2},\varphi_{R^2})$, and $\beta_j \sim \normal(0, \tau^2\phi_j\tilde{\sigma}^2/\sigma_{x_j}^2)$, for $j = 1,\dotsc,p$, where $\tau^2 \sim \mathrm{BetaPrime(\mu_{R^2},\varphi_{R^2})}$. The observation model for inference remains the Binomial, the pseudo-$R^2$ definition is merely used to set the prior. 
We compare two alternatives:
\begin{enumerate}[nosep]
    \item Independent normal priors $\beta_j \sim \normal(0, 1), \; \; \forall j=1, \cdots, p$;
    \item Pseudo-$R^2$ prior with $\mu_{R^2} = \sfrac{1}{3}$, $\varphi_{R^2} = 3$ and Dirichlet concentration parameter set to one.
\end{enumerate}

\paragraph{Takeaways}
To illustrate that the Pseudo-$R^2$ is a predictively consistent prior for the Binomial's latent and outcome space, we plot in Figure~\ref{fig:prior-predictives-logistic-regression} the prior predictive probabilities $(1 / (1 + \exp(-\eta_{d,i})))$ along with the prior predictive Bayesian-$R^2$ when generating $\tilde{R}^2 \sim \betadist(\mu_{R^2},\varphi_{R^2})$ and $\beta_j \sim \normal(0, \tau^2\phi_j\tilde{\sigma}^2/\sigma_{x_j}^2)$ for increasing $p$. The shape of the prior predictive Bayesian-$R^2$ is close to the desired $\betadist(1/3,3)$ distribution, independent of $p$, and remarkably similar to the $R^2$ prior distribution when generated from the normal observation model. The effect of strictly bounding the variance of the linear predictor results in the distribution of the prior predictive probabilities to remain stable and invariant to $p$ as well.
\begin{figure}[tp]
    \centering
    \includegraphics[width=\textwidth]{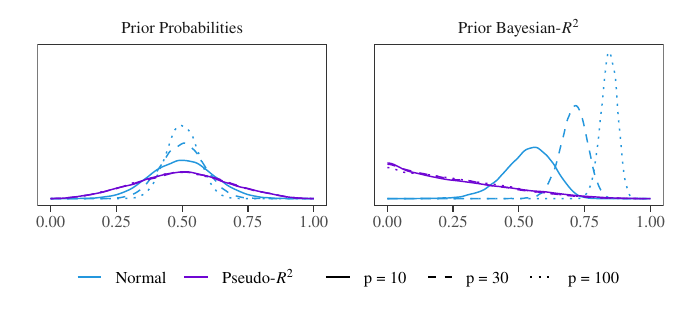}
    \vspace{-1.2cm}
    \caption{Experiment 1: Prior predictive distributions probabilities as well as prior predictive Bayesian-$R^2$ for the logistic regression under the standard normal prior and the pseudo-$R^2$ prior, generating $\tilde{R}^2 \sim \betadist(\mu_{R^2},\varphi_{R^2})$ and $\beta_j \sim \normal(0, \tau^2\phi_j\tilde{\sigma}^2/\sigma_{x_j}^2)$, where $\sigma_{x_j}=1$. 
    } \label{fig:prior-predictives-logistic-regression}
\end{figure}
\begin{figure}[tp]
    \centering
    \includegraphics[width=\textwidth]{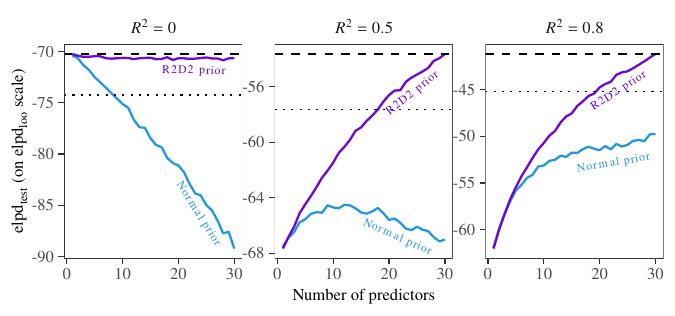}
    \vspace{-0.9cm}
    \caption{Experiment 1: Adding covariates in logistic regression. We compare out-of-sample predictive performance ($\mathrm{elpd}_{\text{test}}$) (averaged over $500$ repetitions) with normal priors or the pseudo-$R^2$ prior with true Bayesian-$R^2 \in \{0, 0.5, 0.8\}$ (columns) and $\rho = 0.5$.
    A difference of 4 on $\mathrm{elpd}_{\text{loo}}$ scale (dotted line) to the best-performing model (dashed line) indicates substantially different predictive performance.
    The pseudo-$R^2$ prior allows to increase model complexity up to the largest model without declining test performance.}
    \label{fig:elpd-path-logistic}
\end{figure}
Under independent standard normal priors, the variance of the linear predictor $x_i^T \beta$ diverges as $p$ increases. Passed through the inverse-logit link function, the prior predictive probabilities $\pi_i$ are pushed toward the boundaries of $0$ and $1$. Consequently, by the continuous mapping theorem \citep[Theorem 2.3]{van_der_vaart_asymptotic_1998}, the expected binomial residual variance vanishes in probability: $\frac{1}{n} \sum_{i=1}^n \pi_i (1 - \pi_i) \xrightarrow{p} 0$ as $p \to \infty$. Simultaneously, the variance of the fitted probabilities, $\Var(\pi)$, converges in distribution to the strictly positive variance of $n$ independent $\text{Bernoulli}(0.5)$ draws. 
Therefore, the Bayesian-$R^2$ for the Binomial model inevitably collapses to $1$.

Results in Figure \ref{fig:elpd-path-logistic} are very similar to Figure \ref{fig:illustrative-example-part-3-implied-prior-r2-elpd-test-rho-05-r2-05}: a predictively consistent prior allows safe use of all predictors without a loss in test performance. 
The best performing models are the ones with all covariates included if the DGP's $R^2$ is larger than 0.
As expected, the difference in performance between priors decreases with an increase in the DGP's $R^2$, since the normal's implied Bayesian-$R^2$ becomes closer to the truth. 

\subsection{Experiment 2: Performing forward selection}\label{subsec:ex-2-forward-search}
We extend the experiment in Section \ref{subsec:examples-pred-consistent-priors-illustrative-example-part-3} and use forward search to identify candidate models of increasing complexity, since it is more realistic that the ordering in which covariates are added matters.
Moreover, we now assume a sparse true DGP where only $15$ out of $50$ covariates have non-zero weights. 

Forward selection is a famous example and widely used cautionary tale for the dangers of overfitting \citep{smith_step_2018, ellis_stepwise_2024, hastie_elements_2009}. 
In the context of Bayesian modelling, \citet{piironen_comparison_2017} and \citet{mclatchie_efficient_2024} illustrate how optimising a utility estimate such as the log score obtained with PSIS-LOO-CV can be liable to finding overfitted models, especially when the data is scarce and the variance of the utility estimate is high, which causes selection-induced bias even if the selection criteria are unbiased estimators of the generalisation utility.
In their experiments, \citet{piironen_projective_2020} consider normal models with an increasing number of predictors and a reference model with a sparsity-assuming normal-Inverse-Gamma prior.
We illustrate that a predictively consistent prior reduces the gap in expected predictive performance between the reference model and the log scores obtained with PSIS-LOO-CV throughout the forward search, and report the projection predictive method's performance to compare with \citet{piironen_comparison_2017}. 

\paragraph{Setup}
We assume a normal DGP as in Section \ref{subsec:examples-pred-consistent-priors-illustrative-example-part-3}, but, similar to \citet[][Section 4.2]{piironen_comparison_2017}, only the first three groups of all $p=50$ covariates have non-zero weights $(\beta^{1:5}, \beta^{6:10}, \beta^{11:15}) = (\xi, 0.5\xi, 0.25\xi)$ while the remaining $(p-15)$ covariates have zero weight. 
Each covariate has zero mean, unit variance and within-group correlation $\rho \in \{0, 0.5, 0.9\}$ and is uncorrelated with covariates in different groups. 
In contrast to Section \ref{subsec:examples-pred-consistent-priors-illustrative-example-part-3}, this implies a block-diagonal structure for the correlation matrix $\Sigma$.
Since we want to keep the block structure as implemented by \citet{piironen_comparison_2017}, we cannot directly use the approach in Section \ref{subsec:examples-pred-consistent-priors-illustrative-example-part-3} to fix a true $R^2$ but instead set $\sigma^2 = \sfrac{1}{R^2} \Var(x^T \beta) (1-R^2)$ with $R^2 = 0.5$ to control the true $R^2$ of the DGP. 

For each sample size $n$, block correlation $\rho$ and $100$ repetitions of the experiment, we generate training data $\{x_i, y_i\}_{i = 1}^n$ and independent test data of size $n_{\text{test}} = 2000$. 
We perform forward search variable selection with the linear model $y_i \sim \normal(\alpha + x_i^T \beta, \sigma^2)$ with (1) independent normal priors or (2) R2D2 prior with selection criterion $\mathrm{elpd}_{\text{loo}}$ estimated with PSIS-LOO-CV, as well as (3) projection predictive selection using the \texttt{projpred} package \citep[][]{piironen_projpred_2026}. 
We use the model with all covariates with R2D2 prior as the reference model.

\paragraph{Prior choices}
Both priors are different from the DGP. 
As before, $\alpha \sim \normal(0, 2.5)$, $\sigma^2 \sim \expdist(1)$ and we compare:
\begin{enumerate}[nosep]
    \item Independent normal priors $\beta_j \sim \normal(0, 1), \; \; \forall j=1, \cdots, p$;
    \item R2D2 prior with $\mu_{R^2} = \sfrac{1}{3}$, $\varphi_{R^2} = 3$ and Dirichlet concentrations $a_{1:p} = (1, \ldots, 1)$. 
\end{enumerate}

\paragraph{Takeaways}

To compare the expected behaviour for the three different approaches, we plot $\mathrm{elpd}_{\text{test}}$ on $\mathrm{elpd}_{\text{loo}}$ scale, averaged over $100$ repetitions for each unique condition of the experiment in Figure \ref{fig:ex-2-forward-search-elpd-test-paths-all-rho-n-100-200-r2-05}. 
\begin{figure}[t]
    \centering
    \includegraphics[]{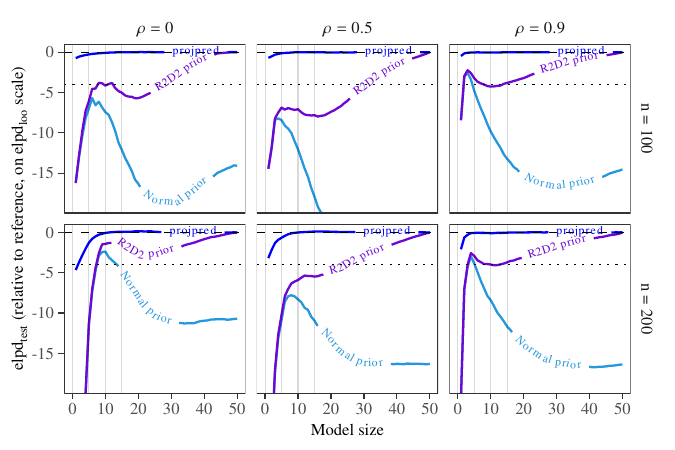}
    \caption{Experiment 2: Forward selection. Predictive performance on independent test data ($\mathrm{elpd}_{\text{test}}$), relative to the reference model rescaled to $\mathrm{elpd}_{\text{loo}}$ scale (with $n_{\text{test}}=2000$, averaged over $100$ repetitions). Data is generated with true $R^2=0.5$ and block correlation $\rho \in \{0, 0.5, 0.9\}$ (columns) and $n \in \{100, 200\}$ (rows). A dotted line marks a difference of 4 relative to the best-performing model and three vertical gray lines at model sizes $5$, $10$ and $15$ indicate the block structure.}
    \label{fig:ex-2-forward-search-elpd-test-paths-all-rho-n-100-200-r2-05}
\end{figure}
When using normal priors, model selection can sometimes help to avoid even worse out-of-sample performance, but, after including $\gtrsim 5$ predictors, we always reach better $\mathrm{elpd}_{\text{test}}$ with R2D2 prior. 
With R2D2 prior, we can safely use all covariates, without fear of deteriorating test performance. 
With independent normal priors, $\mathrm{elpd}_{\text{test}}$ drops and remains clearly below the results achieved with R2D2 prior across $n \in \{100, 200\}$ and $\rho=\{0, 0.5, 0.9\}$. 
If we use a selection method with much smaller model selection variance (like \texttt{projpred}), we can achieve excellent out-of-sample performance also with smaller models. 

Figure \ref{fig:app-ex-2-forward-search-elpd-test-paths-block-structure-with-reps} in Appendix \ref{app:ex-2-forward-search} also shows that forward search with normal priors exhibits larger variation across repetitions than forward search with R2D2 prior.
Variation for both forward search approaches decreases with higher training sample size $n$.
We repeat the experiment for a somewhat more artificial DGP with equally weakly relevant covariates (see Figure \ref{fig:app-ex-2-forward-search-average-elpd-test-paths-flat-structure} and \ref{fig:app-ex-2-forward-search-mean-elpd-test-elpd-loo-paths-data-structure-flat} in Appendix \ref{app:ex-2-forward-search}).
We find the same patterns of considerably larger variance and lower performance for large models with normal priors compared to R2D2 prior, but also that test performance can be similar for uncorrelated predictors ($\rho=0$) and larger training sample size ($n=200$). 
Note that one could improve 
forward search by using the approaches for early-stopping and correction of selection-induced bias by \citet{mclatchie_efficient_2024}, but we focus on comparing ``default'' forward search with PSIS-LOO-CV $\mathrm{elpd}$ for the two priors here.

\subsection{Experiment 3: Increasing the complexity of nonlinear models}\label{subsec:ex-3-nonlinear}
Instead of considering multiple predictors and the behaviour of models from a pre-defined family of models, 
this section considers a single–predictor nonlinear setting  with an increasing polynomial degree or number of basis functions where the exact functional relationship is uncertain, similar to \citet{rasmussen_occam_2000} and \citet[][Chapter 7]{mcelreath_statistical_2020}. 
Building on Section  \ref{subsec:methods-function-priors}, we illustrate how nonlinear modelling choices differ in their implied priors and how this relates to out‑of‑sample performance as model complexity grows.
\paragraph{Setup} 
Data are generated to resemble the ``drowning'' example (number of drownings from 1980-2013) by \citet{gelman_bayesian_2015} available in their accompanying GitHub repository\footnote{See \href{https://github.com/avehtari/BDA_R_demos/blob/master/data/drowning.txt}{https://github.com/avehtari/BDA\_R\_demos/blob/master/data/drowning.txt}}, transforming the predictor to $x\in[-3,3]$.
As such, we cannot assume that any of our candidate models includes the true model.
Note that the DGP is an arbitrary choice. 
Our main goal is to replicate a situation where there is a clear, nonlinear pattern in the data.

\paragraph{Model choices} We assume $y_i \sim \normal\left(f(x_i), \sigma^2\right)$ with $i=1, \cdots, n$. 
We fit models with increasing degree or number of basis functions $k \in \{0, \cdots, 19\}$, starting with the intercept-only model (i.e, $k=0$) with $f(x_i)= \alpha$. 
We investigate \begin{enumerate*}[label = (\arabic*)]
    \item raw, and 
    \item orthogonal polynomial regression with degree $k$,
    \item thin plate spline (TPS) regression \citep[][]{wood_thin_2003, wood_generalized_2017, wood_generalized_2025, hastie_elements_2009}  with rank $k$ and fixed degrees of freedom,
    \item TPS with rank $k$ and 
    non-fixed penalisation,
    \item HSGPs \citep[][]{solin_hilbert_2020, riutort-mayol_practical_2022} with $k$ basis functions as introduced in Section \ref{subsec:methods-function-priors}, 
\end{enumerate*}see a detailed description of model choices and priors in Appendix \ref{app:ex-3-nonlinear-modelling-approaches}.
\begin{figure}[t]
    \centering
    \includegraphics[]{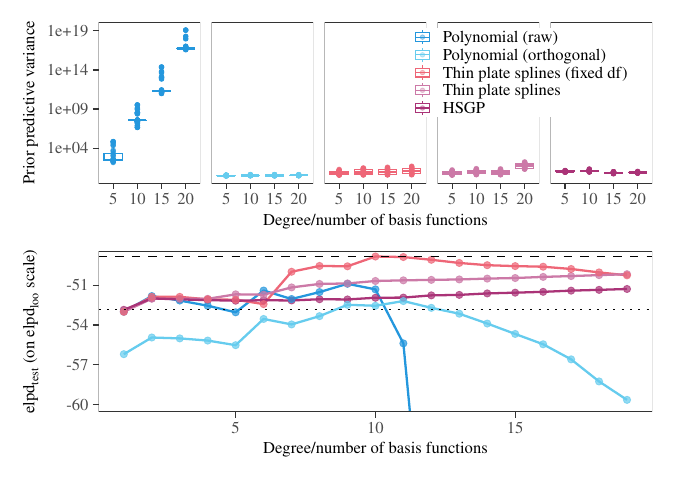}
    \vspace{-0.4cm}
    \caption{Experiment 3: Increasing the complexity of nonlinear models. 
    Top: Prior predictive variance (log10-transformed) for one repetition. 
    Bottom: $\mathrm{elpd}_{\text{test}}$ (with $n=50$, $n_{\text{test}}=2000$, averaged over $500$ repetitions), rescaled to $\mathrm{elpd}_{\text{loo}}$ scale. 
    For raw polynomials, prior predictive variances explode and $\mathrm{elpd}_{\text{test}}$ drops at higher degrees.
    Only thin plate splines and HSGP remain stable as the degree/number of basis functions increases.
    }
    \label{fig:ex-3-nonlinear-basis-functions-prior-predictive-variance-vs-elpd-test-n-50-dgp-BDA-drowning-data}
\end{figure}
\paragraph{Takeaways}
Figure \ref{fig:ex-3-nonlinear-basis-functions-prior-predictive-variance-vs-elpd-test-n-50-dgp-BDA-drowning-data} summarises (top) prior‑predictive variance for one repetition and (bottom) $\mathrm{elpd}_{\text{test}}$ averaged over $500$ repetitions with $n=50$, $n_{\text{test}}=2000$, rescaled to $\mathrm{elpd}_{\text{loo}}$ scale to assess differences.
As opposed to the other approaches, prior predictive variances for raw polynomials explode, consistent with $\sum_j \mathbb E[x^{2j}]$ increasing for $x\in[-3,3]$ (Section \ref{subsec:methods-function-priors}).
As $k$ grows, $\mathrm{elpd}_{\text{test}}$ remains stable only for HSGP and penalised TPS, while raw polynomials, and, at higher degrees, orthogonal polynomials, deteriorate.
For all approaches but HSGPs and penalised TPS, we cannot simply increase model complexity without a drop in test performance.

As illustrated also for $n\in\{20, 100\}$ in Figure \ref{fig:app-ex-3-nonlinear-obs-elpd-paths-n-20-50-100} in Appendix \ref{app:ex-3-nonlinear}, orthogonal polynomials can sometimes delay the drop in test performance and restrict the range of basis values (see Figure \ref{app:fig-ex-3-nonlinear-basis-function-values} in Appendix \ref{app:ex-3-nonlinear}).
In contrast to TPS and HSGPs, they still require specifying individual priors for model components and can suffer rank deficiencies when $k$ approaches the number of unique $x$ values. 
We did not attempt to further fix higher-degree polynomials, since we aim to investigate ``out‑of‑the‑box'' behaviour under independent coefficient priors versus function‑space priors, not to optimise each modelling approach.

\subsection{Experiment 4: Randomised control treatment effect studies}\label{subsec:ex-4-rct-studies}

Previous sections focused on predictive performance, demonstrating that the best predictive performance is obtained using all relevant data combined with predictively consistent priors. 
We now turn to causal inference: estimating the average treatment effect in a randomised controlled trial (RCT), a standard methodology for studying causal effects of interventions from economics to medicine \citep{deaton_understanding_2018}. 
We demonstrate that conditioning on high-dimensional covariate sets, which are predictive of the target and not the treatment probability, can improve inference for the treatment effect. 
However, predictive
consistency alone in the sense of Definition~\ref{def:second-order-pc} is not sufficient to realise these efficiency gains and may even increase bias compared to not using the covariates in the model at all.
We develop theory to support the simulation findings.

We are interested in estimation of the following model:
\begin{equation}
    y_i \sim \normal(\alpha z_i + x_i^T\beta, \sigma^2), \label{eq:rct-model}
\end{equation}
where the treatment assignment $z_i \sim \mathrm{Bernoulli}(0.5)$ for $i=1,\dotsc,n$ and covariates $X \in \mathbb{R}^{n\times p}$ with columns drawn independently from a standard normal distribution.
Here, since $z\perp X$, the inclusion of $X$ into the model are assumed not to have a bias effect and ideally have a variance reduction effect onto the posterior of $\alpha$.\footnote{In the parlance of \citet[][]{cinelli_crash_2024}, $X$ in this DGP design constitute neutral controls in that they neither increase nor decrease asymptotic bias.}
We consider two structural models: $\mathrm{M}_{\text{base}}$ containing only the treatment variable $z$, and $\mathrm{M}_{\text{full}}$ containing both the treatment $z$ and the covariates $X$, where, for simplicity, the correlation among columns of $X$ is assumed to be zero. 

\paragraph{Prior choices} 
For $\mathrm{M}_{\text{full}}$, we compare three prior specifications where a normal prior for coefficient vector, $\beta \sim \normal(0,\Sigma_{\beta})$ and a normal prior for the treatment effect coefficient, $\alpha \sim \normal(0,\tau^2_{\alpha})$ is applied:
\begin{enumerate}[nosep]
    \item \textbf{Normal:} $\Sigma_{\beta} = I_p$ and $\tau_{\alpha} = 2.5$.
    \item \textbf{R2 (joint R2D2):} A joint R2D2 prior over the $p+1$ coefficients $(\beta_1,\dotsc,\beta_p,\alpha)$, so $\Sigma_{\beta} = \tau^2\mathrm{diag}(\phi_1,\dotsc,\phi_p)$ and $\tau^2_{\alpha} = \tau^2\phi_{p+1}$, with $\mu_{R^2} = 1/3$, $\varphi_{R^2} = 3$ and symmetric concentration $a=1$.
    \item \textbf{R2 + Normal (split):} An R2D2 prior for $\beta$ with $\mu_{R^2} = 1/3$, $\varphi_{R^2} = 3$ and symmetric concentration $a=1$ and $\tau^2_{\alpha} = 2.5$ independently.
\end{enumerate}
For $\mathrm{M}_{\text{base}}$, we use $\alpha \sim \normal(0, 2.5)$.

\paragraph{Analytical predictions}
In line with the simulation in the previous section, we focus on the finite-sample regime where $n$ is held fixed and $p$ increases but extend the analysis to $p$ increasing towards $n$. 
Asymptotic results are deferred to the proportional regime $p=\lfloor\gamma n\rfloor$, $n\to\infty$, analysed in Appendix~\ref{app:treatment-proofs-prop}.

Conditional on hyperparameters and $\sigma^2$, the posterior of $\alpha$ marginal of $\beta$ is normal and can be found by the application of the Bayesian Frisch-Waugh-Lovell decomposition \citep{hahn_bayesian_2020, bloniarz_lasso_2016}. 
Via the Schur complement of the joint posterior precision (Proposition \ref{prop:marginal-var}, Appendix \ref{app:treatment-proofs})
\begin{align}
    \mathbb{E}[\alpha \mid \cdot] = s \cdot \hat{\alpha}_{\text{OLS}_\beta}, \quad
     \Var(\alpha \mid \cdot) = \left(d_z + \tau_\alpha^{-2}\right)^{-1}, \label{eq:bayes_fwl}
\end{align}
with $\hat{\alpha}_{\text{OLS}_\beta} = [z^T(I_n-H_\beta)z]^{-1}z^T(I_n-H_\beta)y$ being the residualised least-squares estimate, $H_\beta = X(X^TX + \sigma^2\Sigma_\beta^{-1})^{-1}X^T$ the $\beta$-prior-regularised projection matrix onto $\mathrm{col}(X)$, and $d_z = \sigma^{-2}z^T(I_n - H_\beta)z$ the residualised data precision contributed by the component of $z$ orthogonal to the prior-regularised covariate space.  Since $s = d_z/(d_z + \tau_\alpha^{-2}) \in (0,1)$ controls how much the prior shrinks $\hat{\alpha}_{\text{OLS}_\beta}$ towards zero, it can be interpreted as the shrinkage factor \citep{polson_shrink_2011, piironen_sparsity_2017}.
Inference quality under each prior is governed by how $\tau_\alpha^{-2}$ and $d_z$ behave as $p$ grows.
We simplify the analysis by further marginalising $d_z$ 
(Lemma~\ref{lem:exact-data-prec}, Appendix~\ref{app:treatment-proofs}):
\begin{equation*}
    \mathbb{E}[d_z \mid \cdot] = \frac{1}{4\sigma^2}\bigl[(n - \mathrm{tr}(H_\beta)) + \mathbf{1}_n^T(I_n - H_\beta)\mathbf{1}_n\bigr].
\end{equation*}
Table \ref{tab:rct-rates} summarises the behaviour of the posterior of $\alpha$ as $p$ increases towards $n$, for each prior and whether heavy shrinkage is applied, which we define via the condition,
\begin{equation}\label{eq:heavy-shrinkage-text}
    (p+1)/n \;\gg\; R^2/(1-R^2), \qquad R^2 := \tau^2/(\tau^2+\sigma^2).
\end{equation}
Under this condition, the R2D2 penalty on $\beta$ dominates the spectrum of $X^TX$, so $H_\beta \to 0$ (uniformly over the Dirichlet weights, provided \eqref{eq:heavy-shrinkage-text} holds up to a $\log p$ factor, see Proposition~\ref{prop:rct-joint-approx}) and $\mathbb{E}[d_z \mid \cdot]$ tightens to the parametric scale $n/(2\sigma^2)$. 
Each row is proven with a corresponding supporting proposition in Appendix~\ref{app:treatment-proofs-i}.

\begin{table}[t]
\centering
\small
\caption{Behaviour of the marginal posterior of $\alpha$ in the finite-sample regime (i), with $n$ fixed and $p\to n$. ``Heavy shrinkage'' refers to condition~\eqref{eq:heavy-shrinkage-text}. 
The symbol ``$\downarrow$'' marks stochastic (in-probability) decrease in $p$, and $\Theta_p(\cdot)$ denotes a stochastic scale at fixed $n$ rather than a limit in $p$. 
Quantities written in $n$ are finite-$n$ constants rather than rates. The full table including the proportional regime (ii), $p=\lfloor\gamma n\rfloor$ with $n\to\infty$, is Table~\ref{tab:rct-rates-full} in Appendix~\ref{app:treatment-proofs-prop}.}
\label{tab:rct-rates}
\begin{tabular}{@{}lllll@{}}
\toprule
Prior & $\tau_\alpha^{-2}$ & $\mathbb{E}[d_z \mid \cdot]$ & $s$ & $\Var(\alpha \mid \cdot)$\\
\midrule
\multicolumn{5}{@{}l}{\textit{General regime (i): $n$ fixed, $p \to n$}} \\
Normal & $\mathcal{O}(1)$ & $>0$, $<n/(2\sigma^2)$ & bounded in $(0,1)$ & $\mathcal{O}(1)$, $\le \tau_\alpha^2$\\
Split        & $\mathcal{O}(1)$ & $>0$, $<n/(2\sigma^2)$ & bounded in $(0,1)$ & $\mathcal{O}(1)$, $\le \tau_\alpha^2$\\
Joint R2D2   & $\Theta_p(p/\tau^2)$ & $>0$, $<n/(2\sigma^2)$ & $\downarrow$, min $>0$ & $\downarrow$, min $>0$\\
\midrule
\multicolumn{5}{@{}l}{\textit{Under heavy shrinkage \eqref{eq:heavy-shrinkage-text} (regime (i))}} \\
Split        & $\mathcal{O}(1)$ & $\to n/(2\sigma^2)$ & $1 - \mathcal{O}(\sigma^2/n)$ & $\Theta(\sigma^2/n)$\\
Joint R2D2   & $\Theta_p(p/\tau^2)$ & $\to n/(2\sigma^2)$ & $\Theta_p\!\left(\tfrac{nR^2}{(1-R^2)p}\right)$, $\downarrow$ & $\Theta_p(\tau^2/p)$, $\downarrow$\\
\bottomrule
\end{tabular}

\medskip
    Proofs for each row are in Appendix~\ref{app:treatment-proofs}. 
    The Normal prior is covered by Proposition \ref{prop:rct-normal}, the Split prior by Proposition~\ref{prop:rct-split}, and the Joint R2D2 prior by Proposition \ref{prop:rct-joint} (general regime) and Proposition~\ref{prop:rct-joint-approx} (heavy shrinkage). 
    Frequentist coverages are derived in Proposition~\ref{prop:coverage}.
\end{table}
The main issue revealed in Table \ref{tab:rct-rates} is that the joint R2D2 prior increasingly shrinks $\alpha$ as $p$ grows and can bias its posterior mean.
This comes from the behaviour of $\tau^{-2}_{\alpha}$.
Since, marginally, $\phi_\alpha \sim \mathrm{Beta}(\mu = 1/(p+1), \varphi = a(p+1))$ under symmetric $a$, the conditional prior precision $\tau_\alpha^{-2}$ increases with $p$ at stochastic scale $\Theta_p(p/\tau^2)$.\footnote{This in-probability scale holds for all $a>0$.} 
The Normal and split priors avoid this by fixing $\tau_\alpha^{-2} = \mathcal{O}(1)$, which leaves $s \in (0,1)$ bounded and can be shown to yield bounded posterior variance of $\alpha$. 

Under the heavy shrinkage condition, the split prior dominates the Normal prior on the variance of the posterior of $\alpha$: the R2D2 component on $\beta$ induces an effective per-coordinate penalty $\sigma^2(p+1)/\tau^2$ that grows with $p$, dominating the eigenvalues of $X^TX$ uniformly and tightening $\mathbb{E}[d_z\mid \cdot]$ to the parametric scale $n/(2\sigma^2)$. The Normal prior's fixed penalty $\sigma^2$ never drives $H_\beta$ to zero, so it leaves $H_\beta$ close to the projection $P_X$ and its $\mathbb{E}[d_z\mid \cdot]$ stays strictly below the ceiling $n/(2\sigma^2)$ without tightening to it. 
Under the same condition, the joint R2D2 prior's $\mathbb{E}[d_z\mid \cdot]$ also reaches $n/(2\sigma^2)$, but is overwhelmed by the simultaneously growing $\tau_\alpha^{-2} = \Theta(p/\tau^2)$. 

To see for which dimensionality-signal ratio the joint R2D2 prior strongly shrinks $\alpha$, consider the plug-in approximation to $s$ at the prior mean, $s(\mathbb{E}[\phi_\alpha])$.
Under the heavy-shrinkage condition (Proposition~\ref{prop:rct-joint-approx}) this becomes
\begin{equation}
    s \;\approx\; \frac{n\tau^2}{n\tau^2 + 2\sigma^2(p+1)} \;=\; \frac{n R^2}{n R^2 + 2(1-R^2)(p+1)},\label{eq:shrinkage_ratio}
\end{equation}
for example, when $(p+1)/n = R^2/[2(1-R^2)]$ then the prior and data contributions to the posterior precision of $\alpha$ are equal, giving $s \approx 1/2$. Jensen's inequality guarantees that $\mathbb{E}[s(\phi_\alpha)] < s(\mathbb{E}[\phi_\alpha])$ since $s$ is concave in $\phi_{\alpha}$. Consequently, the closed-form approximation in Equation \ref{eq:shrinkage_ratio} represents an upper bound on the expected shrinkage.

\paragraph{Setup} 
We set the true data-generating process such that the treatment explains a fraction $\delta_T$ of the total explained variance, that is, $\delta_T = \alpha^{*2} z^Tz / (\alpha^{*2} z^Tz + \beta^TX^TX\beta)$, with covariate coefficients drawn from $\beta_j \sim \normal(0,0.05)$. This yields the true treatment coefficient
\begin{equation}
    \alpha^* = \sqrt{\frac{\delta_T \, \beta^T X^T X \beta}{z^T z \, (1 - \delta_T)}}. \label{eq:true-alpha}
\end{equation}
The residual variance $\sigma^2$ is set to control the overall $R^2$.\footnote{Solve $R^2 =  (\alpha^2z^Tz + \beta^TX^TX\beta)/ (\alpha^2z^Tz + \beta^TX^TX\beta + n\sigma^2)$ for $\sigma$.}
We simulate with $R^2 \in \{0.2, 0.8\}$, $\delta_T \in \{0.2,0.8\}$, $n = 150$, and $p = 100$, placing us in the regime $n > p$ but with a high ratio $p/n \approx 0.67$. 
We mostly focus on $R^2=0.2$, $\delta_T = 0.2$ here, results for the other configurations are in Appendix \ref{app:treatment-study}.

\paragraph{Takeaways}
\begin{table}[t]
\centering
\small
\caption{Experiment 4: Posterior recovery metrics for the treatment effect ($\alpha$). Average 90\% credible interval length (Len.), coverage probability (Cov.), and root mean square error (RMSE) over $500$ repetitions. The base model contains only the treatment variable and all other models include $p=100$ covariates. Bold marks lowest RMSE, all coverage values within $\pm 0.03$ of the nominal 0.9, and smallest interval length conditional on correct calibration.}
\begin{tabular}{@{}l ccc c ccc@{}}
\toprule
\multirow{2}{*}{\textbf{Model}} & \multicolumn{3}{c}{$\boldsymbol{\delta_T = 0.2}$} & \phantom{a} & \multicolumn{3}{c}{$\boldsymbol{\delta_T = 0.8}$} \\
\cmidrule{2-4} \cmidrule{6-8}
& \textbf{Len.} & \textbf{Cov.} & \textbf{RMSE} && \textbf{Len.} & \textbf{Cov.} & \textbf{RMSE} \\
\midrule
\multicolumn{8}{@{}l}{\textbf{Low Signal ($R^2 = 0.2$)}} \\
\midrule
Base (Only Treatment) & 0.92 & \textbf{0.92} & \textbf{0.27} && 0.92 & \textbf{0.91} & \textbf{0.27} \\
Normal                & 0.97 & \textbf{0.87} & 0.30          && 0.97 & \textbf{0.87} & 0.31          \\
Joint R2D2            & 0.49 & 0.12 & 0.61          && 0.95 & 0.41          & 0.59          \\
Split (R2 + Normal)   & \textbf{0.87} & \textbf{0.90} & \textbf{0.27} && \textbf{0.87} & \textbf{0.89} & \textbf{0.27} \\
\midrule
\multicolumn{8}{@{}l}{\textbf{High Signal ($R^2 = 0.8$)}} \\
\midrule
Base (Only Treatment) & 0.57 & 0.94          & 0.15          && 0.57 & 0.95          & 0.15          \\
Normal                & 0.51 & \textbf{0.88} & 0.16          && 0.51 & \textbf{0.88} & 0.16          \\
Joint R2D2            & 0.51 & 0.39          & 0.34          && 0.48 & 0.67 & 0.23          \\
Split (R2 + Normal)   & \textbf{0.48} & \textbf{0.90} & \textbf{0.14} && \textbf{0.48} & \textbf{0.89} & \textbf{0.14} \\
\bottomrule
\label{tab:treatment_recovery_metrics}
\end{tabular}
\end{table}Table~\ref{tab:treatment_recovery_metrics} confirms the analytical predictions. 
The joint R2D2 prior shows severe downward bias, consistent with the small shrinkage factor from Equation~\eqref{eq:shrinkage_ratio} at the data-generating $R^2 = 0.2$ ($s \approx 0.16$).\footnote{At this design the heavy-shrinkage condition~\eqref{eq:heavy-shrinkage-text} holds only with a modest margin, so $s \approx 0.16$ is an approximate upper bound on the expected shrinkage rather than an exact value.} 
This leads to narrow but poorly calibrated intervals. 
By contrast, the split prior yields the sharpest recovery of $\alpha^*$, near-nominal frequentist coverage (which we derived analytically, see Proposition~\ref{prop:coverage}, Appendix~\ref{app:treatment-proofs}), and lowest RMSE, competitive with the base model. 
The base model's intervals remain close to calibrated but are wider since it does not benefit from covariate adjustment. 
The Normal model avoids the joint R2D2 prior's bias, since $\tau_\alpha^{-2}=\mathcal{O}(1)$ keeps $s$ bounded away from zero, but its weak fixed penalty leaves $H_\beta$ close to the full projection onto $\mathrm{col}(X)$, leading to larger posterior variance of $\alpha$. 
\begin{figure}[tp]
    \centering
    \includegraphics[width=\textwidth]{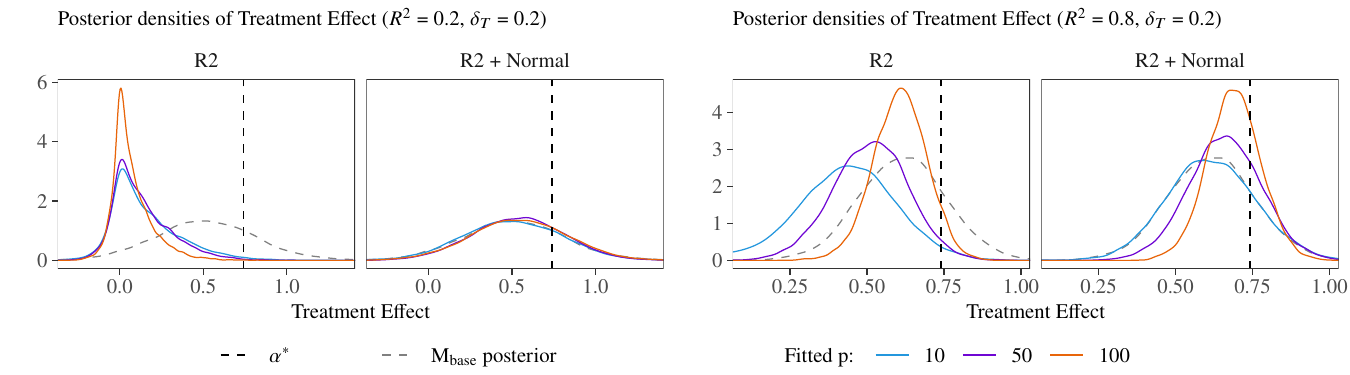}
    \vspace{-0.7cm}
    \caption{Experiment 4: Posterior of the treatment effect for increasing $p$. Grey lines show the treatment-only model $\mathrm{M}_{\text{base}}$; coloured lines show $\mathrm{M}_{\text{full}}$ for different $p$. When the signal is relatively small ($R^2 = 0.2$, left panel of figures), the joint R2D2  model produces increasing bias with $p$. When the signal is large enough ($R^2 = 0.8$, right panel of figures), the R2D2 joint model concentrates on $\alpha^*$ with $p$. In either setting, the split R2 model becomes sharper around $\alpha^*$ with $p$, which is more pronounced for $R^2=0.8$.} \label{fig:posterior-recovery-treatment-increasing-p}
\end{figure}
Figure~\ref{fig:posterior-recovery-treatment-increasing-p} isolates the effect of increasing $p$, where the DGP stays constant, but the treatment posterior is estimated with expanding subsets of the covariates. 
With strong shrinkage ($R^2 = 0.2, \delta_T = 0.2$), the joint R2D2 prior becomes increasingly biased, whereas the split prior gives sharper posteriors without further bias.
However, when the shrinkage is moderate due to high $R^2$ (right panel), and therefore the expected bias smaller, the joint R2D2 prior also concentrates around $\alpha^*$. 

\section{Discussion}\label{sec:discussion}
This work revisits the role of model selection in Bayesian workflows and argues that it is not a substitute for building good models in the first place.
We investigate when selecting a simpler model is necessary, when it can harm out-of-sample predictive performance, and how prior choices affect safe model expansion.
Our main proposal is to use \textit{predictively consistent} priors instead of enforcing parsimony at the level of individual model components through model selection.

We define predictively consistent and parsimonious priors in Section \ref{subsec:methods-formalisation-prior-desiderata} and illustrate how such priors can make model selection unnecessary in settings where it is often recommended to avoid overfitting. 
Our results lead to concrete recommendations for checking joint prior implications and building complex yet well-generalising models (Section \ref{intro:takeaways-for-modellers}). 
Across our examples and experiments in Section \ref{subsec:examples-pred-consistent-priors-illustrative-example-part-3} and \ref{sec:experiments}, 
large models with predictively consistent priors usually reach equal or better out-of-sample predictive performance than selected simpler models.
When simpler models perform better, gains are usually small.
In some cases, selection from a nested set of models even harms test performance relative to using the largest model with a suitable joint prior.
This also matters for iterative modelling workflows. 
A common concern is that expanding models step by step can lead to overfitting.
Our results suggest that predictively consistent priors can safeguard against deteriorating test performance and thereby support safe model expansion.

The relevant question is not whether a prior is specified jointly or component-wise, but whether its joint implications reflect the structure of the modelling problem. 
This is especially important for models with many covariates, where reasonable marginal priors can still imply an unreasonable prior predictive distribution. 
Independence is not a safe default, but imposing dependence can also be harmful. 
Section \ref{subsec:ex-4-rct-studies} further illustrates that prior structure can matter for inference, not only for prediction. 
In the considered setting, splitting the prior on the treatment effect and the other covariate effects can yield sharper posterior intervals without adding bias, whereas a joint prior over treatment and covariate effects can introduce bias.

Priors based on interpretable joint quantities are one practical way to make prior specification more manageable \citep[see, e.g.,][]{aguilar_intuitive_2023, aguilar_generalized_2025}.
We discuss R2D2 priors and function-space priors as two examples for constructing predictively consistent
priors.
Other priors may also satisfy predictive consistency, provided that they control the induced signal distribution as complexity increases. 
Thus, the main message is not tied to a particular prior family, but the question is whether the chosen prior keeps the prior predictive distribution sensible and stable with increasing model complexity.

Our work has several limitations.
First, we focus on additive models.
This allows a clean formalisation, but also limits the scope of our results.
Future work should extend the concept of predictively consistent priors to non-additive modelling scenarios. 
Second, our experiments focus on scenarios with $n \geq p$.
We expect the problematic patterns observed for independent normal priors to get much worse when $n<p$. 
However, a systematic investigation of other $n$-$p$ regimes remains open.
Third, we compare only selected prior configurations.
Results for other priors might fall between the considered comparison points, for example, when scaling normal priors depending on the number of covariates.
Such $p$-scaled normal priors would also restrict prior predictive variance but at the cost of shrinking all coefficients strongly towards zero.
More work is needed to compare a wider range of models and prior choices.

Finally, most of our experiments focus on predictive performance, assuming that the considered models are causally sensible. 
In more realistic cases of causal analysis, we should not blindly include all possible potential predictors, but need to consider the causal structure \citep[e.g., possible confounders, colliders, and backdoor effects, see][]{mcelreath_statistical_2020} to include only predictors that would minimise the bias in the estimation of the effect of interest.
Predictive consistency alone does not guarantee unbiased estimation of a causal effect.
Our example in Section \ref{subsec:ex-4-rct-studies} probes aspects of this in one controlled setting.
More work is needed to define desirable prior properties for causal analyses, where the focus is not only on predictive performance. 

\subsection*{Acknowledgements}

We thank Javier Enrique Aguilar and Yann McLatchie for helpful comments and discussions, and Sölvi Rögnvaldsson for his web application supporting R2D2 prior elicitation. 
We acknowledge the computational resources provided by the Aalto Science-IT project. 
This work was supported by the Research Council of Finland grant (313122) and the Research Council of Finland Flagship programme: Finnish Center for Artificial Intelligence (FCAI).


\begin{refcontext}[sorting=nyt]
\printbibliography
\end{refcontext}

\appendix

\section{Notation}\label{app:notation}
Table~\ref{tab:notation} summarises the notation used throughout the paper. We work directly with realised quantities and do not introduce separate symbols for random variables. Capital Latin letters denote matrices, and lower-case letters denote vectors and scalars, with dimensions clear from context.
\begin{table}[h]
    \centering
    {\small
    \begin{tabularx}{\textwidth}{lX}
        \toprule
        Symbol & Meaning \\
        \midrule
        $n$, \; $i=1,\dots,n$ & Number of observations and observation index \\
        $p$, \; $j=1,\dots,p$ & Number of covariates and covariate (coefficient) index \\
        $y_i$ & Scalar outcome of observation $i$ \\
        $y=(y_1,\dots,y_n)^T$ & Outcome vector \\
        $\tilde{y}$ & Test (out-of-sample) outcome \\
        $X\in\mathbb{R}^{n\times p}$ & Design matrix \\
        $x_i\in\mathbb{R}^p$ & Covariate vector of observation $i$ (the $i$-th row of $X$) \\
        $x_{ij}$ & Scalar covariate entry, the $j$-th covariate of observation $i$ \\
        $\theta$ & Model parameters (general) \\
        $\beta=(\beta_1,\dots,\beta_p)^T$, \; $\beta_j$ & Regression coefficient vector and its $j$-th entry, with $\theta_d\equiv\beta_{1:d}$ in the linear examples \\
        $\eta=f(x;\theta)$ & Systematic component (linear predictor), written $\eta_d$ at complexity $d$ \\
        $\mathcal{M}=\{\mathrm{M}_k\}_{k=1}^{K}$, \; $k$ & Finite collection of candidate models and model index \\
        $d$ & Model-complexity index, counting covariates, basis functions, or group levels, equal to $p$ in the regression examples \\
        $d_z$ & Residualised data precision for the treatment coefficient (Section~\ref{subsec:ex-4-rct-studies}) \\
        \midrule
        $p(\cdot)$ & Generic density, used for the prior, likelihood, posterior, posterior predictive, and marginal distributions, distinguished by their arguments and by model subscripts such as $p_{\mathrm{M}_k}$ and $p_t$ \\
        $\{p_d\}_{d\in\mathbb{N}}$ & Prior sequence, with $p_d$ the joint prior of the model at complexity $d$ \\
        $q_d$ & Induced (pushforward) distribution of the systematic component $\eta_d$ \\
        $\mathbb{E}[\cdot]$, $\Var(\cdot)$, $\mathrm{Cov}(\cdot)$ & Expectation, variance, and covariance. A subscript names the distribution, as in $\mathbb{E}_{p_d}$ under the prior $p_d$ and $\Var_{q_d}$ under $q_d$ \\
        \bottomrule
    \end{tabularx}}
    \caption{Summary of notation.}
    \label{tab:notation}
\end{table}

\clearpage

\section{Relationship between second order predictive consistency and predictive parsimony}\label{app:non-monotonic-parsimony}

This appendix proves Lemma~\ref{lem:pc-pp-relationship} and provides the supporting technical result (Proposition~\ref{prop:non-monotonic-parsimony}). For the derivations below, we assume that the prior predictive signal distributions are conditional on $V_d$
for any $d$.
\begin{proof}[Proof of Lemma~\ref{lem:pc-pp-relationship}]
Throughout, $q_d = \normal(0, V_d)$ with $V_d > 0$.

\emph{Part~(i): $\mathrm{S} \not\subseteq \mathrm{F}$.}
Choose $V_d = L + (-1)^d d^{-1/3}$ for any $L > 1$ (so that $V_1 = L - 1 > 0$ and hence $V_d > 0$ for all $d$).
Then $\sup_d V_d < \infty$ and $V_d \to L$, so this sequence is in $\mathrm{S}$.
Since $(-1)^{d+1} = -(-1)^d$, consecutive terms have opposite signs, so $|V_{d+1} - V_d| = (d+1)^{-1/3} + d^{-1/3} \approx 2d^{-1/3}$ for large $d$. Hence $(V_{d+1} - V_d)^2 \approx 4d^{-2/3}$ and $\sum_{d=1}^{\infty} d^{-2/3} = \infty$ ($p$-series with $p \leq 1$).
By Proposition~\ref{prop:non-monotonic-parsimony}(i), $\sum_{d=1}^{\infty} D_{\mathrm{KL}} = \infty$, so this sequence is not in $\mathrm{F}$.

\emph{Part~(ii): $\mathrm{F} \not\subseteq \mathrm{S}$.}
For the independent normal prior, $V_d = d\overline{\tau}^2 \to \infty$, hence, this prior produces a sequence which is not in $\mathrm{S}$.
Furthermore, $V_{d+1}/V_d = (d+1)/d = 1 + 1/d$, so $D_{\mathrm{KL}}(q_{d+1}\|q_d) = \frac{1}{2}[1/d - \log(1+1/d)]$.
Since $\log(1+x) = x - x^2/2 + O(x^3)$, we get $D_{\mathrm{KL}} = \frac{1}{4d^2} + O(d^{-3})$, and $\sum_{d=1}^{\infty} d^{-2} < \infty$. Hence, this sequence is in F.

\emph{Part~(iii): Characterisation of $\mathrm{S} \cap \mathrm{F}$.}
By definition, $\mathrm{S}$ requires $V_d$ bounded and convergent.
Given boundedness ($0 < V_{\min} \leq V_d \leq C$), Proposition~\ref{prop:non-monotonic-parsimony}(i) shows $\sum_{d=1}^{\infty} D_{\mathrm{KL}} < \infty \iff \sum_{d=1}^{\infty}(V_{d+1}-V_d)^2 < \infty$.
Proposition~\ref{prop:non-monotonic-parsimony}(ii) states the condition on $V_d$  sufficiently small to yield a sequence in the intersection of F and S. Note that any prior which produces a \emph{monotone sequence} (and is in S) fulfils this condition: if $V_d$ is monotone and converges to $L > 0$, then $\sum_{d=1}^{\infty}|V_{d+1}-V_d| = |L - V_1| < \infty$\footnote{Note that this covers both $V_d \uparrow L$ and $V_d\downarrow L$.} by telescoping (all terms have the same sign), so $\sum_{d=1}^{\infty}(V_{d+1}-V_d)^2 \leq \sup_d|V_{d+1}-V_d| \cdot \sum_{d=1}^{\infty}|V_{d+1}-V_d| < \infty$.
\end{proof}

The following proposition, used above, characterises KL summability in terms of squared increments.

\begin{proposition}\label{prop:non-monotonic-parsimony}
Let $q_d = \normal(0, V_d)$ with $0 < V_{\min} \leq V_d \leq C < \infty$ for all $d$, and let $\delta_{d+1} := V_{d+1} - V_d$.
\begin{enumerate}[(i)]
    \item $\sum_{d=1}^{\infty} D_{\mathrm{KL}}(q_{d+1} \| q_d) < \infty$ if and only if $\sum_{d=1}^{\infty} \delta_{d+1}^2 < \infty$.
    \item A sufficient condition for $D_{\mathrm{KL}}$ summability is $|\delta_{d+1}| = \mathcal{O}(d^{-g})$ for some $g > 1/2$. This exponent is sharp: for every $0 < g \leq 1/2$, there exist variance paths satisfying Definition~\ref{def:second-order-pc} with $|\delta_{d+1}| = \Theta(d^{-g})$ for which $\sum_{d=1}^{\infty} D_{\mathrm{KL}} = \infty$.
\end{enumerate}
\end{proposition}

\begin{proof}
\emph{Part (i).}
Define $r_d := \delta_{d+1}/V_d$, so that $V_{d+1}/V_d = 1 + r_d$ and
\begin{equation}
    D_{\mathrm{KL}}(q_{d+1} \| q_d) = \tfrac{1}{2}\bigl[r_d - \log(1 + r_d)\bigr]. \label{eq:kl-relative-increment}
\end{equation}
Consider the function $h(r) := [r - \log(1+r)]/r^2$ for $r \in (-1,\infty) \setminus \{0\}$, extended by continuity to $h(0) = 1/2$.
This function is continuous and strictly positive on $(-1,\infty)$.
Since $V_{d+1}/V_d \in [V_{\min}/C,\; C/V_{\min}]$, the values $r_d$ lie in a compact subset of $(-1,\infty)$. By the extreme value theorem, the continuous function $h$ attains a minimum $c_1 > 0$ and a maximum $c_2 < \infty$ on this set.
Therefore $c_1\, r_d^2 \leq D_{\mathrm{KL}}(q_{d+1}\|q_d) \leq c_2\, r_d^2$ for all $d$, and substituting $r_d = \delta_{d+1}/V_d$ with $V_{\min} \leq V_d \leq C$ gives
\begin{equation}
    \frac{c_1}{C^2}\,\delta_{d+1}^2 \;\leq\; D_{\mathrm{KL}}(q_{d+1}\|q_d) \;\leq\; \frac{c_2}{V_{\min}^2}\,\delta_{d+1}^2,
\end{equation}
establishing the equivalence $\sum_{d=1}^{\infty} D_{\mathrm{KL}} < \infty \iff \sum_{d=1}^{\infty} \delta_{d+1}^2 < \infty$.

\emph{Part (ii): Sufficiency.}
If $|\delta_{d+1}| \leq A\, d^{-g}$ for some $A > 0$ and $g > 1/2$, then $\delta_{d+1}^2 \leq A^2 d^{-2g}$ and $\sum_{d=1}^{\infty} d^{-2g} < \infty$ since $2g > 1$.

\emph{Part (ii): Sharpness.}
Fix $0 < g \leq 1/2$ and choose $V_1$ large enough that $V_d := V_1 + \sum_{k=1}^{d-1}(-1)^k k^{-g} \geq V_{\min} > 0$ for all $d$ (possible since the alternating series converges by Leibniz's test, so the partial sums are bounded).
Then $V_d \to L := V_1 + \sum_{k=1}^{\infty}(-1)^k k^{-g} < \infty$ and $\sup_d V_d < \infty$, so Definition~\ref{def:second-order-pc} is satisfied.
However, $\delta_{d+1} = (-1)^d d^{-g}$, giving $\sum_{d=1}^{\infty} \delta_{d+1}^2 = \sum_{d=1}^{\infty} d^{-2g} = \infty$ since $2g \leq 1$.
By part~(i), $\sum_{d=1}^{\infty} D_{\mathrm{KL}} = \infty$.
\end{proof}

\section{Connections between $\Var(\eta)$, Bayesian $R^2$, and $\mathrm{elpd}$}

We assume a data-generating process $y \sim \normal(x^T \beta^*, \sigma^{2*})$ and the model
\begin{align}
    y \mid \beta, \sigma^2 \sim \normal(x^T \beta, \sigma^2).  
\end{align} 
Our prior assumptions for $\beta$ influence $\Sigma_{\beta} = \Var(\beta)$. 
This affects the scale of the linear predictor $\eta(x) = x^T\beta$, since the marginal variance of the linear predictor is $\Var\left(\eta(x)\right) = x^T \Var(\beta) x = x^T \Sigma_{\beta} x$. 
For data $D = \{x_i, y_i\}_{i=1}^n$, and new observation $\{\tilde{x}, \tilde{y}\}$, we can approximate the posterior predictive distribution $p(\tilde{y} \mid \tilde{x}, D)$ with $\normal(\mu(\tilde{x}), v(\tilde{x}))$, where 
\begin{align}
   \mu(\tilde{x}) &=  \mathbb{E}\left[\eta(\tilde{x}) \mid D\right] = \tilde{x}^T \mathbb{E}\left[\beta \mid D\right] 
   \\
   \label{eq:app-connections-predictive-variance}
   v(\tilde{x}) &= \mathbb{E}\left[\sigma^2 \mid D\right] + \Var(\eta(\tilde{x}) \mid D) = \mathbb{E}\left[\sigma^2 \mid D\right] 
   + \tilde{x}^T \Var(\beta \mid D) \tilde{x}. 
\end{align}
This means that the predictive variance $v(\tilde{x})$ at $\tilde{x}$ is a sum of the  posterior mean of the observation noise variance and the posterior variance of the linear predictor. 
For a normal model with normal prior $\beta \sim \normal(0, \Sigma_{\beta})$, the posterior mean is $\mathbb{E}[\beta \mid D] = (X^TX + \sigma^2 \Sigma_{\beta}^{-1})^{-1} X^T y$ (with design matrix $X$ and outcome vector $y$) and it becomes apparent that $\Sigma_{\beta}$ influences both predictive mean and variance. 

Using $\mu(\tilde{x})$ and $v(\tilde{x})$, the pointwise log predictive density for new observation $\{\tilde{x}, \tilde{y}\}$ can be written as  
\begin{align}
    \mathrm{lpd}(\tilde{x}, \tilde{y}) = \log p(\tilde{y} \mid \tilde{x}, D) \approx - \frac{1}{2} \left[\log (2 \pi v(\tilde{x})) + \frac{(\tilde{y} - \mu(\tilde{x}))^2}{v(\tilde{x})}\right].
    \label{eq:app-connections-lpd-exact}
\end{align} 
Taking the expectation with respect to the true DGP 
\begin{align}
    \mathbb{E}[\mathrm{lpd}(\tilde{x}, \tilde{y})] &= \mathbb{E}\left[ - \frac{1}{2} \left[\log (2 \pi v(\tilde{x})) + \frac{(\tilde{y} - \mu(\tilde{x}))^2}{v(\tilde{x})}\right]\right] \\
    &= - \frac{1}{2} \log (2 \pi v(\tilde{x})) - \frac{1}{2} \frac{\mathbb{E}\left[(\tilde{y} - \mu(\tilde{x}))^2\right]}{v(\tilde{x})} \\
    &= - \frac{1}{2} \log (2 \pi) - \frac{1}{2}\log  (v(\tilde{x})) - \frac{1}{2} \frac{\sigma^2_* + (\tilde{x}^T\beta_* - \mu(\tilde{x})) ^2}{v(\tilde{x})},  
\end{align}
where we explicitly separate the true noise variance $\sigma^2_*$ from the model parameter $\sigma^2$, whose posterior mean affects the predictive variance $v(\tilde{x})$ (see \eqref{eq:app-connections-predictive-variance} above). 
We see that $\mathbb{E}[\mathrm{lpd}(\tilde{x}, \tilde{y})]$ is affected by $(\tilde{x}^T\beta_* - \mu(\tilde{x}))^2$ and $v(\tilde{x})$ through $- \frac{1}{2}\log  (v(\tilde{x})) $ and $ - \frac{1}{2} \frac{\sigma^2_* + (\tilde{x}^T\beta_* - \mu(\tilde{x})) ^2}{v(\tilde{x})}$. 
When $\Var(\eta(x))$ increases, $v(\tilde{x})$ increases. 
Then, $\mathbb{E}[\mathrm{lpd}(\tilde{x}, \tilde{y})]$ is penalised more with $-\frac{1}{2}\log  (v(\tilde{x}))$.

The approximation with a normal distribution, and, therefore, also \eqref{eq:app-connections-lpd-exact}, are only exact for a conjugate model with normal priors and known $\sigma^2$.
Otherwise, we can use posterior draws $\beta^{(s)}$ and $\sigma^{2 (s)}$
to get 
\begin{align}
    \mu(x) &\approx \widehat{\mu}(x) = \frac{1}{S} \sum_{s=1}^S \eta^{(s)}(x) \\ 
    v(x) &\approx \widehat{v}(x) = \frac{1}{S} \sum_{s=1}^S \sigma^{2 (s)} + \frac{1}{S} \sum_{s=1}^S \left(\eta^{(s)}(x) - \widehat{\mu}(x)\right)^2,  
\end{align}
where $\eta^{(s)}(x) = x^T \beta^{(s)}$. 
Given posterior draws, we could again approximate $\mathrm{lpd}$ as in \eqref{eq:app-connections-lpd-exact} by replacing $\mu(x)$ with the predictive mean estimate $\widehat{\mu}(x)$ and $v(x)$ with the predictive variance estimate $\widehat{v}(x)$, but it is more accurate to use exact Monte Carlo 
\begin{align}
    \mathrm{lpd}(x, y) &\approx 
    \log \frac{1}{S} \sum_{s=1}^S \left[ \frac{1}{2 \pi \sigma^{2 (s)}} \exp\left( - \frac{\left(y - \eta^{(s)}(x)\right)^2}{2 \sigma^{2 (s)}}\right) \right].  
\end{align}
Again, the prior choices for $p(\beta)$ influence $\Var\left(\eta^{(s)}\right)$, that is, the dispersion of $\left\{\eta^{(s)}\right\}_{s=1}^S$, through $\Var(\beta \mid D) = \Sigma_{\beta}$. 
This affects the predictive variance via $\frac{1}{S} \sum_{s=1}^S \left(\eta^{(s)}(x) - \widehat{\mu}(x)\right)^2$ and, therefore, Bayesian $R^2$ as defined by \citet{gelman_r-squared_2019}. 
Moreover, it also influences the spread 
in the term $\left(y - \eta^{(s)}(x)\right)^2$ in $\mathrm{lpd}$ which, in turn, could affect computation or approximation of $\mathrm{elpd}_{\mathrm{loo}}$ and $\mathrm{elpd}_{\mathrm{test}}$.

\begin{align}
    p(\beta) \to
    \Var\left(\eta^{(s)}\right) \left[\leftrightarrow R^2\right] \to \mathrm{lpd} \to \mathrm{elpd}
\end{align}

To tie this back to results presented before, see Section \ref{subsec:methods-formalisation-prior-desiderata}, \ref{subsubsec:methods-r2d2} and \ref{subsec:methods-function-priors} where we provide details on how different explicit or implied prior choices can affect $\Var(\eta)$ under increasing model complexity. 

\citet{gelman_r-squared_2019} also define a residual-based Bayesian $R^2$ as 
\begin{equation}
    R^2 = \frac{\Var_{\mu}}{\Var_{\mu} + \Var_{\text{res}}}.
\end{equation} 
In particular, following their notation, when using draws from the residual distribution, we get 
\begin{equation}
    \Var_{\text{res}} = V_{i=1}^n \hat{e}_i^{(s)} = \frac{1}{n-1} \sum_{i=1}^n \left(\hat{e}_i^{(s)}\right)^2 = \frac{1}{n-1} \sum_{i=1}^n \left(y_i - \hat{y}_i^{(s)}\right)^2
\end{equation}
for the residual-based $R^2$.  
At the same time, posterior mean squared error (MSE) values can be computed by evaluating, 
\begin{equation}
    \mathrm{MSE} = \frac{1}{n} \sum_{i=1}^n \left(y_i - \hat{y}_{i}^{(s)}\right)^2
\end{equation}
for each posterior prediction $\hat{y}_i^{(s)}$ and each $y_i$ with $i=1, \cdots n$. 
This directly corresponds to $\Var_{\text{res}}$ for the residual based $R^2$ by \citet{gelman_r-squared_2019} that is discussed also in the corresponding online appendix \href{https://avehtari.github.io/bayes_R2/bayes_R2.html}{https://avehtari.github.io/bayes\_R2/bayes\_R2.html}. 

\section{Illustrative example: Part 1}\label{app:illustrative-example-part-1}

\subsection{Illustrative example: Part 1. Results with normal priors versus R2D2 prior}\label{app:illustrative-example-part-1-normal-r2d2}

Figure \ref{fig:app-illustrative-example-dgp-normal-elpd-test-models-some-selection-aggregation-strategies} shows results for $\alpha \sim \normal(0, 2.5)$, $\sigma_{\epsilon}^2 \sim \expdist(1)$ with different prior choices for $\beta$ in the full model ($\mathrm{M}_1$): 
\begin{enumerate}[nosep]
    \item $\beta \sim \normal(0, 1)$, and 
    \item an R2D2 prior with mean $\mu_{R^2}=\sfrac{1}{3}$ and precision $\varphi_{R^2} = 3$.
\end{enumerate}
For a single covariate, the Dirichlet prior for the variance decomposition parameter collapses to a constant, but the Beta prior on $R^2$ still induces a prior for $\beta$ via $R^2 = \beta^2 / (1 + \beta^2)$. 
We choose the R2D2 prior here despite the low dimensionality of the problem because it retains boundedness while putting most prior mass at small $R^2$ values near zero in our chosen configuration.
The specific shape across the range of $R^2$ values can be adjusted through the configuration of the R2D2 prior.
\begin{figure}[htp!]
    \centering
    \includegraphics[width=0.7\linewidth]{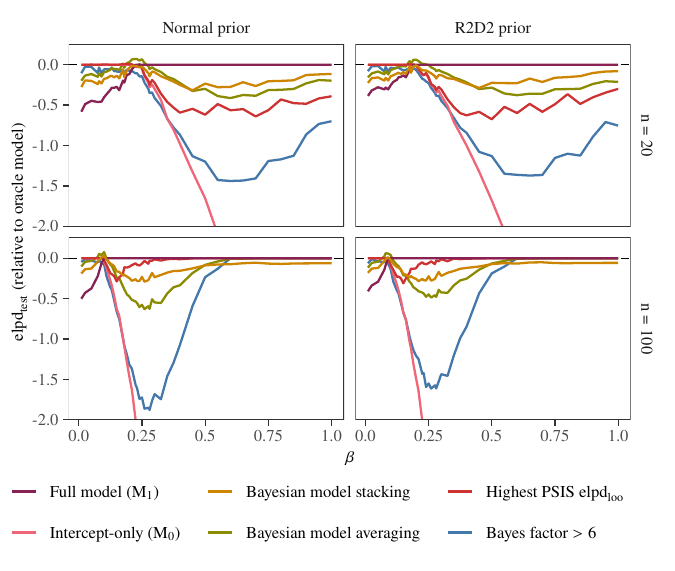}
    \vspace{-0.3cm}
    \caption{Illustrative example: Part 1. We compare predictive performance with $\elpdPlain$ on independent test data ($\mathrm{elpd}_{\text{test}}$) for the full model with one covariate ($\mathrm{M}_1$), the intercept-only model ($\mathrm{M}_0$), and four model selection and averaging strategies relative to the oracle model \eqref{eq:oracle-def} across different training data sizes (rows) and prior choices (columns). 
    The x-axis shows true effect sizes for $\beta$.
    Results are averages across $500$ repetitions and rescaled to the same scale as $\mathrm{elpd}_{\text{loo}}$. 
    Compared to results with normal prior, $\mathrm{elpd}_{\text{test}}$ for $\mathrm{M}_1$ is closer to the oracle model when using the R2D2 prior.
    For effect sizes $>0.25$, $\mathrm{M}_1$ is the oracle model for both prior choices.}
    \label{fig:app-illustrative-example-dgp-normal-elpd-test-models-some-selection-aggregation-strategies}
\end{figure}
We observe that the R2D2 prior can slightly improve the full model's worst-case predictive performance on independent test data in the (small) section of $\beta$ values where the intercept-only model $\mathrm{M}_0$ and other strategies have higher predictive performance.
This improvement for $\mathrm{M}_1$ is bigger when $n$ is smaller. 

\clearpage

\subsection{Illustrative example: Part 1. Comparison to Student-t DGP}\label{app:methods-illustrative-example-part-1-student-t-dgp-normal-vs-r2d2-prior}
To compare the previously presented results in Section \ref{subsec:methods-illustrative-example-part-1}, where $\mathrm{M}_1$ contains the true model, to a scenario where the DGP is similar but not as $\mathrm{M}_1$, we assume 
\begin{align}
    y_i \sim \text{Student-t}(\nu, \alpha + \beta x_i, \sigma), \text{ with } x_i \sim \normal(0,1), \nu = 20, \sigma=1. 
    \label{eq:illustrative-example-dgp-stud-t}
\end{align}
Figure \ref{fig:app-illustrative-example-dgp-stud-t-elpd-test-models-some-selection-aggregation-strategies} shows similar patterns as discussed for the normal DGP in Figure \ref{fig:intro-single-predictor-only-models-prior-normal-dgp-normal} and \ref{fig:app-illustrative-example-dgp-normal-elpd-test-models-some-selection-aggregation-strategies}. 
\begin{figure}[htp!]
    \centering
    \includegraphics[width=0.7\linewidth]{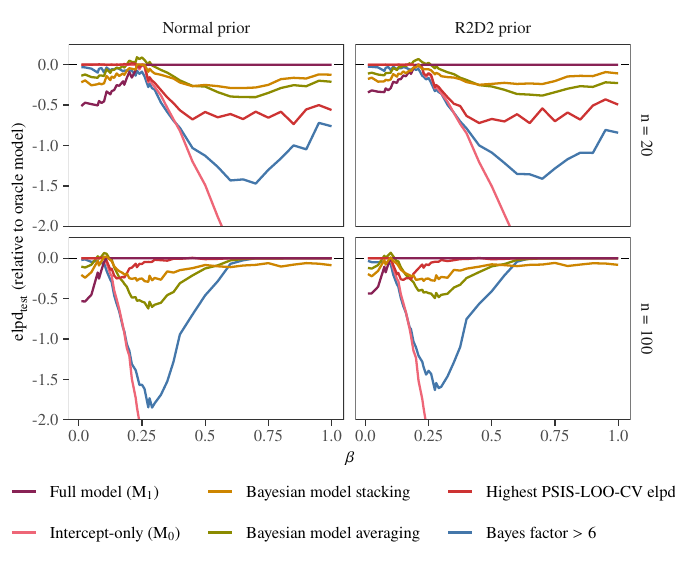}
    \caption{Illustrative example: Part 1. 
    We compare predictive performance on independent test data ($\mathrm{elpd}_{\text{test}}$) for the full model with one covariate ($\mathrm{M}_1$), the intercept-only model ($\mathrm{M}_0$), and four model selection and averaging strategies relative to the oracle model \eqref{eq:oracle-def} across different training data sizes (rows) and prior choices (columns). 
    We employ a Student-t DGP with degrees of freedom $\nu = 20$.
    The x-axis shows the chosen effect sizes for $\beta$.
    Results are averages across $500$ repetitions and rescaled to the same scale as $\mathrm{elpd}_{\text{loo}}$.}
    \label{fig:app-illustrative-example-dgp-stud-t-elpd-test-models-some-selection-aggregation-strategies}
\end{figure}

\subsection{Illustrative example: Part 1. Using model probabilities, pseudo-BMA, LOO-BB and stacking weights for model selection}\label{app:illustrative-example-other-selection-strategies}
The difference in predictive performance between two models $\mathrm{M}_0$ and $\mathrm{M}_1$ can be estimated from the pointwise differences in predictive log densities. 
Given reliably estimated differences $\Delta \widehat{\mathrm{elpd}}_{\text{loo}}(\mathrm{M}_0, \mathrm{M}_1 \mid y)$ between models $\mathrm{M}_0$ and $\mathrm{M}_1$ and well-calibrated standard error for the estimator of the differences $\widehat{\mathrm{se}}_{\text{loo}}(\mathrm{M}_0, \mathrm{M}_1 \mid y)$, we can use a normal approximation to compute the probability that one model is better than another \citep{vehtari_bayesian_2002, sivula_uncertainty_2025} and use this for model selection as illustrated in Figure \ref{fig:app-illustrative-example-selection-with-model-prob}. 
\begin{figure}[htp!]
    \centering
    \includegraphics[width=0.7\linewidth]{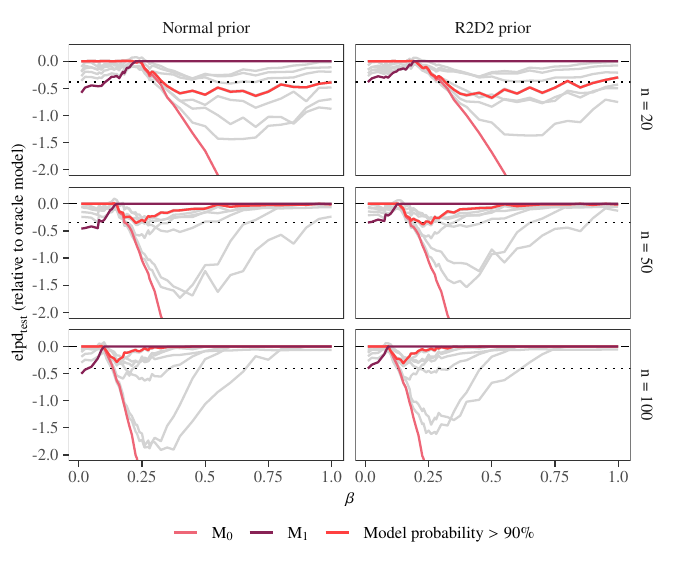}
    \vspace{-0.4cm}
    \caption{Illustrative example: Part 1. Out-of-sample performance for the intercept-only model $\mathrm{M}_0$, the full model $\mathrm{M}_1$ and selection with model probability $>90\%$. The other methods considered in Section 3 are plotted as gray lines in the background. The dotted lines indicate $\mathrm{M}_1$'s worst-case performance with R2D2 prior.}
    \label{fig:app-illustrative-example-selection-with-model-prob}
\end{figure}
\citet{sivula_uncertainty_2025} provide theoretical and experimental results for conditions when the standard error estimate (and corresponding normal approximation) is reasonably well-calibrated. 
This paper is not focusing on properties of Bayesian model stacking or other model weights, but, based on our experiments, weights for stacking, pseudo-BMA and LOO-BB are often non-zero for both models when $n$ is small (see Figure \ref{fig:app-illustrative-example-selection-with-stacking-weights}, \ref{fig:app-illustrative-example-selection-with-pbma-weights}, \ref{fig:app-illustrative-example-selection-with-loo-bb-weights}). 
\begin{figure}[htp!]
    \centering
    \includegraphics[width=0.7\linewidth]{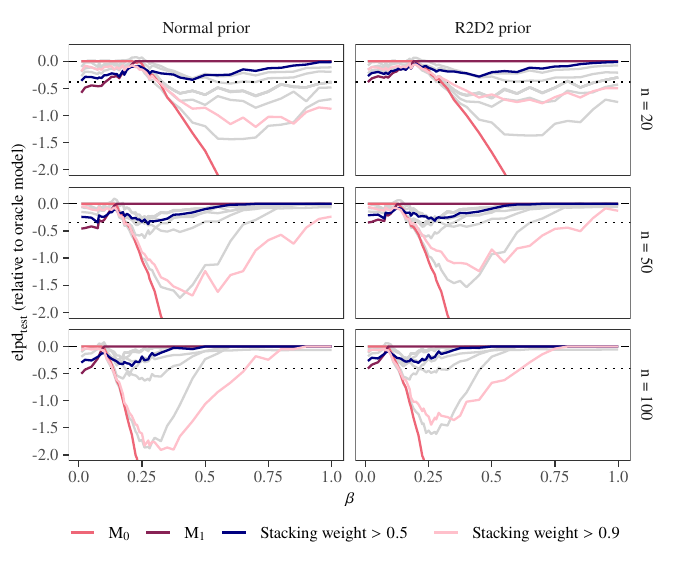}
    \vspace{-0.4cm}
    \caption{Illustrative example: Part 1. Out-of-sample performance for the intercept-only model $\mathrm{M}_0$, the full model $\mathrm{M}_1$ and selection using stacking weights $>0.5$ or $0.9$. The other methods in Section 3 are plotted as gray lines in the background. The dotted lines indicate $\mathrm{M}_1$'s worst-case performance with R2D2 prior.}
    \label{fig:app-illustrative-example-selection-with-stacking-weights}
\end{figure}
Further research is needed on using stacking, pseudo-BMA, or LOO-BB weights for selection.
For example, stacking was not designed for selection but for aggregating the results of several models \citep[][]{yao_using_2018}. 
Non-zero weights might not harm the performance of the aggregated result, but, of course, they can harm selection when the selection threshold is strict and high (see, for example, results for selecting with stacking weights $>0.9$ in Figure \ref{fig:app-illustrative-example-selection-with-stacking-weights}).  
\begin{figure}[htp!]
    \centering
    \includegraphics[width=0.7\linewidth]{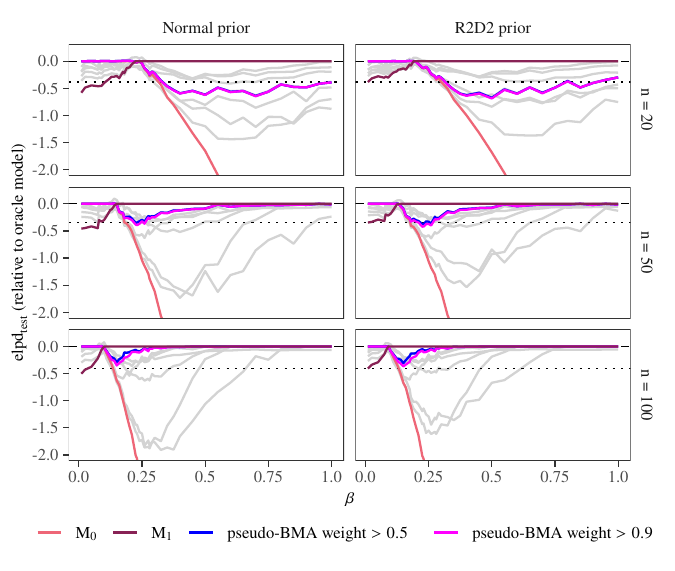}
    \vspace{-0.4cm}
    \caption{Illustrative example: Part 1. Out-of-sample performance for the intercept-only model $\mathrm{M}_0$, the full model $\mathrm{M}_1$ and  selection using pseudo-BMA weights $>0.5$ or $0.9$. The other methods considered in Section 3 are plotted as gray lines in the background. The dotted lines indicate $\mathrm{M}_1$'s worst-case performance with R2D2 prior.}
    \label{fig:app-illustrative-example-selection-with-pbma-weights}
\end{figure}

\begin{figure}[htp!]
    \centering
    \includegraphics[width=0.7\linewidth]{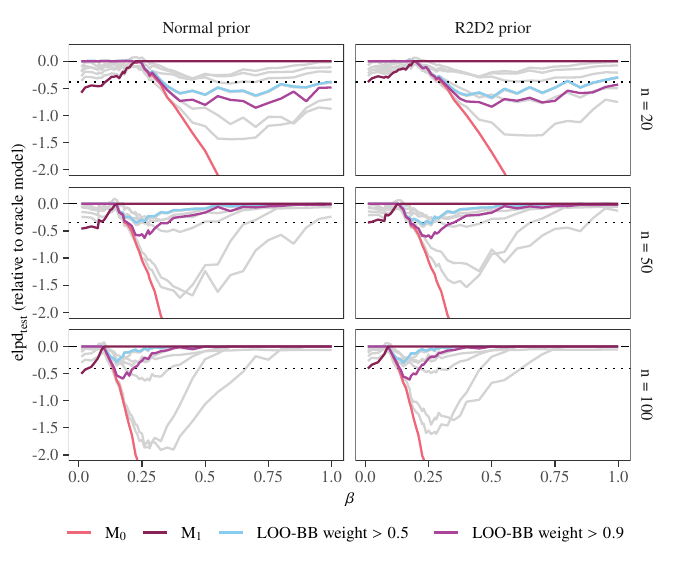}
    \vspace{-0.4cm}
    \caption{Illustrative example: Part 1. Out-of-sample performance for the intercept-only models $\mathrm{M}_0$, the full model $\mathrm{M}_1$ and  selection using LOO-BB weights $>0.5$ or $0.9$. The other methods considered in Section 3 are plotted as gray lines in the background. The dotted lines indicate $\mathrm{M}_1$'s worst-case performance with R2D2 prior.}
    \label{fig:app-illustrative-example-selection-with-loo-bb-weights}
\end{figure}

\clearpage

\section{Examples of predictively consistent priors}\label{app:examples-pred-consistent-priors}

\subsection{Other global-local priors}\label{app:section-examples-pred-consistent-priors-other-global-local}
For any global-local prior with local weights satisfying $\sum_j \phi_j = 1$, Lemma~\ref{lem:gl-characterisation} reduces second order predictive consistency to requiring $\mathbb{E}_{p_d}[\tau^2] < \infty$, independently of model size.
The \emph{Dirichlet--Laplace} (DL) prior \citep{bhattacharya_dirichlet-laplace_2015} places a $\Dirichlet(a, \ldots, a)$ distribution on the local scales $(\phi_1, \ldots, \phi_d)$, so $\sum_j \phi_j = 1$ by construction, and assigns a $\mathrm{Gamma}\left(na, \tfrac{1}{2}\right)$ prior to the global scale $\tau$, ensuring $\mathbb{E}[\tau^2] < \infty$ for fixed $n$.
\footnote{\citet{bhattacharya_dirichlet-laplace_2015} recommend $a = n^{-(1+\iota)}$ for any $\iota > 0$, which yields minimax-optimal posterior contraction rates and ensures second order predictive consistency when $n\rightarrow \infty$. For any $a > 0$, $\tau^2$ has finite second moments for fixed $n$.}
\footnote{\citet{zhang_variable_2018} further derived calibration of $a$ to a target $R^2$, while retaining the Laplace marginal prior structure.}
\emph{Penalised complexity} (PC) priors \citep{simpson_penalising_2017} define a \emph{base model} and penalise departures from it. The PC prior places an exponential distribution on the Kullback--Leibler distance from the base, which, for the base being the null model, results in $\tau \sim \mathrm{Exponential}(\lambda)$ (i.e.\ a type-2 exponential prior on the standard deviation), so $\mathbb{E}_{p_d}[\tau^2] = 2/\lambda^2 < \infty$, and predictive consistency follows.
When local weights are \emph{not} constrained to sum to one (e.g., independent half-Cauchy scales as in the horseshoe prior), predictive consistency depends on the joint behaviour of $\tau^2$ and $\sum_j \phi_j$ and must be verified case by case.

\clearpage

\subsection{Illustrative example: Part 3. 
Bayesian $R^2$ based on prior and realised Bayesian $R^2$}

\begin{figure}[htp!]
    \centering
    \includegraphics[]{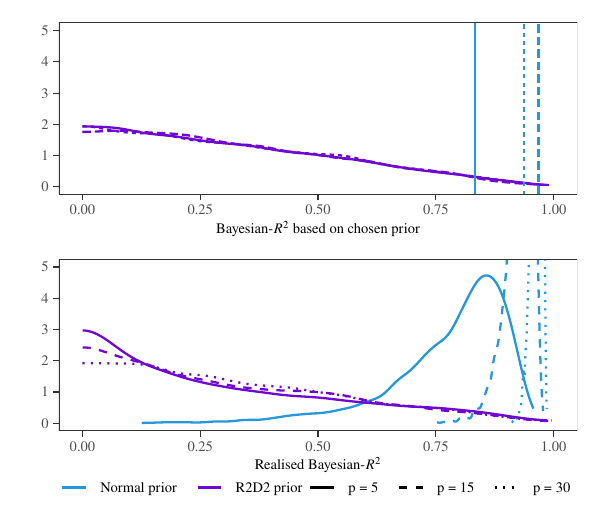}
    \caption{Illustrative example: Part 3. Implied prior Bayesian-$R^2$ for increasing number of covariates $p \in \{5, 15, 30\}$. 
    We assume $\sigma^2=1$ and standardised and orthogonal predictors, that is, $\frac{1}{n} X^TX = I$. 
    The first row shows Bayesian-$R^2$ based on (1) the Beta prior with mean $\mu_{R^2} = \sfrac{1}{3}$ and scale $\varphi_{R^2} = 3$ in the R2D2 specification and (2) constants based on $R^2 = \sfrac{p \sigma_{\beta}}{(p \sigma_{\beta} + \sigma^2)}$ for the independent normal priors $\beta \sim \normal(0, \sigma_{\beta})$ with $\sigma_{\beta} = 1$. 
    The second row presents realised prior Bayesian-$R^2$, making use of the fact that $ \Var\left(X^T \beta \mid \beta\right)=||\beta||^2$ and $||\beta||^2 \sim \chi^2(p)$ for the normal priors, since $\frac{1}{n} X^TX = I_p$.}
    \label{fig:ex-1-linear-regression-prior-r2}
\end{figure}

\subsection{Illustrative example: Part 3. Generating data with fixed signal-to-noise ratio}\label{app:illustrative-example-part-3-data-generation}

Section \ref{subsec:examples-pred-consistent-priors-illustrative-example-part-3} uses simulated data from a linear model where Bayesian-$R^2$ (Equation \eqref{eq:def-bayesian-r2-gelman}) and the residual variance $\sigma^2$ are fixed. 
Predictors are $x_i \in \mathbb{R}^p$ with $x_i \sim \normal(0,\Sigma)$, where each diagonal element of $\Sigma$ equals 1 and off-diagonal elements equal $\rho$. 
We seek either a fixed coefficient value $b$ (equal across coefficients) or a coefficient variance $\tau^2$ such that the realised $R^2$ (\ref{eq:def-bayesian-r2-gelman}) matches the target, given $\sigma^2$.
We consider two cases:
\begin{enumerate}[nosep]
  \item Fixed coefficients with $\beta_j=b$ for $j=1,\dots,p$ that is, $\beta$ is fixed and treated as a constant;
  \item Random coefficients with $\beta_j \sim \normal(0,\tau^2)$, independent across $j$.
\end{enumerate}
\paragraph{Fixed $\beta_j=b$.}
We assume known $R^2$, $\sigma^2$, and $\rho$ and use $\Var(\beta^T x)=\beta^T \Var(x) \beta=\beta^T \Sigma \beta$ to write 
\begin{align*}
\Var(\beta^T x)
&= \sum_{j=1}^p \sum_{k=1}^p \Sigma_{jk}\beta_j\beta_k
 = \sum_{j=1}^p b^2 + 2\sum_{j=1}^p \sum_{k<j} \rho\,b^2 = p\,b^2 \left(1+\rho(p-1)\right),
\end{align*}
so
\begin{align*}
R^2
= \frac{\Var(\beta^T x)}{\Var(\beta^T x)+\sigma^2}
= \frac{p\,b^2 \left(1+\rho(p-1)\right)}{p\,b^2 \left(1+\rho(p-1)\right)+\sigma^2}.
\end{align*}
Solving for $b$ gives
\begin{align}\label{eq:beta-bar}
b
= \pm \sqrt{\frac{R^2\,\sigma^2}{(1-R^2) p (1+p\rho-\rho)}}.
\end{align}
\paragraph{Random $\beta_j \sim \normal(0,\tau^2)$.}
For the second case, we assume that the value of $\beta \in \mathbb{R}^p$ is random, with $\mathbb{E}[\beta_j] = 0$, $\Var(\beta_j) = \tau^2$ and $\mathrm{Cov}(\beta_j, \beta_k) = 0$ for $j \neq k$.
We assume $\beta_j$ and $x_k$ are independent for any $j, k$.
Then, we want to find the value of $\tau^2$ that fixes the $R^2$ of the data generating process given $\sigma^2$. 
By the law of total variance,
\begin{align*}
\Var(\beta^T x)
= \mathbb{E}_{\beta}\left[\Var(\beta^T x \mid \beta)\right] + \Var_{\beta}\left(\mathbb{E}[\beta^T x \mid \beta]\right).
\end{align*}
The second term is zero since $\mathbb{E}[x]=0$ and $\mathbb{E}[\beta^T x \mid \beta]=0$. 
For the first term, $\Var(\beta^T x \mid \beta) = \beta^T \Sigma \beta$ and 
\begin{align*}
    \mathbb{E}_{\beta}[\beta^T \Sigma \beta]
    = \text{tr}\left(\Sigma\,\mathbb{E}\left[\beta\beta^T\right]\right)
    = \text{tr}\left(\Sigma \tau^2 I_p\right)
    = \tau^2 \text{tr}(\Sigma)
    = \tau^2 p,
\end{align*}
because $\text{tr}(\Sigma)=\sum_{j=1}^p \Sigma_{jj}=p$ regardless of $\rho$. Hence,
\begin{align}\label{eq:tau}
    R^2 = \frac{p\,\tau^2}{\sigma^2 + p\,\tau^2}
    \rightarrow
    \tau^2 = \frac{\sigma^2 R^2}{p(1-R^2)}.
\end{align}
Equivalently, component-wise, 
$\Var\left(\beta_j x_j\right)=\mathbb{E}\left[\beta_j^2\right]\mathbb{E}\left[x_j^2 \right]=\tau^2$ and, for $j\neq k$,
$\mathrm{Cov}\left(\beta_j x_j,\beta_k x_k\right)=\mathbb{E}\left[\beta_j\beta_k\right]\mathbb{E}\left[x_j x_k\right]=0$, so
\begin{align*}
    \Var\left(\beta^T x\right)=\Var\left(\sum_{j=1}^p \beta_j x_j\right) = \sum_{j=1}^p \Var\left(\beta_j x_j\right) + 2 \sum_{j=1}^p \sum_{k < j}^p \mathrm{Cov}\left(\beta_j x_j, \beta_k x_k\right) = p\tau^2.
\end{align*}

\clearpage

\subsection{Illustrative example: Part 3. Prior and posterior implied $R^2$}\label{app:illustrative-example-part-3-prior-posterior-r2}
\begin{figure}[htp!]
    \centering
    \includegraphics[width=\textwidth]{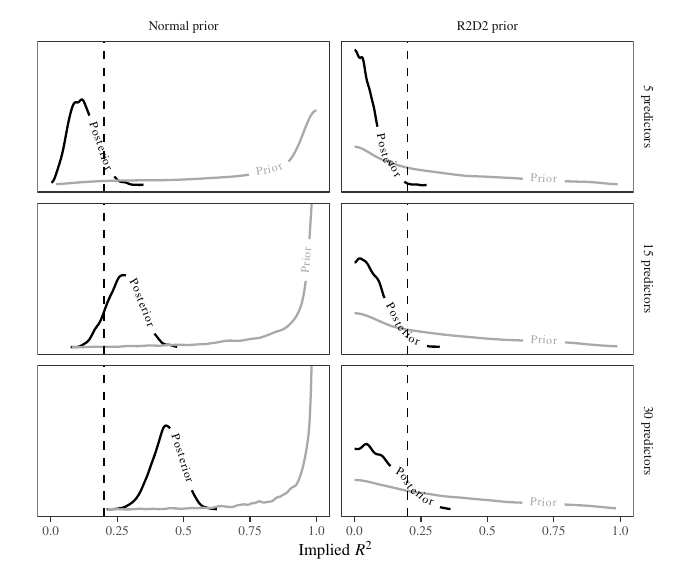}
    \vspace{-0.5cm}
    \caption{Illustrative example: Part 3. Implied prior and posterior $R^2$ for models with independent normal prior or R2D2 prior (columns) and $p\in \{5, 15, 30\}$ predictors (rows). The true $R^2=0.2$, indicated by a vertical dashed line. With increasing number of predictors, independent normal priors concentrate a-priori at $R^2$ values close to one which also pulls the posterior $R^2$ values away from the true value. Results are more stable and overall closer to the true $R^2$ for the R2D2 prior.}
    \label{fig:app-ex-1-adding-covariates-prior-posterior-r2}
\end{figure}

\clearpage
 
\subsection{Illustrative example: Part 3. Out-of-sample predictive performance for different true $R^2$ and uncorrelated predictors}
\begin{figure}[htp!]
    \begin{subfigure}{\textwidth}
        \centering
        \includegraphics[]{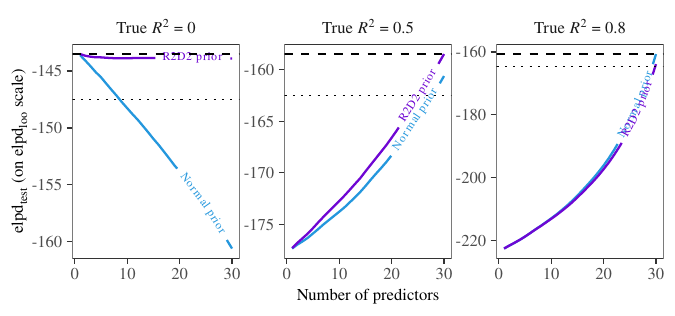}
        \caption{$p=30$}
        \label{fig:app-illustrative-example-rho-0-p-30-subfigure-1}
    \end{subfigure}
    \hfill
    \begin{subfigure}{\textwidth}
        \centering 
        \includegraphics[]{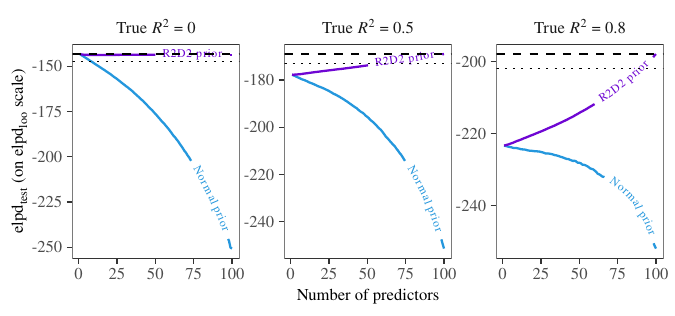}
        \caption{$p=100$}
        \label{fig:app-illustrative-example-rho-0-p-100-subfigure-2}
    \end{subfigure}
    \caption{Illustrative example: Part 3. 
    We compare out-of-sample predictive performance ($\mathrm{elpd}_{\text{test}}$) with normal priors or R2D2 prior with true $R^2 \in \{0, 0.5, 0.8\}$ (columns). 
    The subfigures show results for a DGP with $p \in \{30, 100\}.$
    Results are averaged over $500$ repetitions and each based on $100$ observations with uncorrelated predictors ($\rho = 0$) and independent test data with $n_{\text{test}} = 2000$. 
    Results are on $\mathrm{elpd}_{\text{loo}}$ scale and a dotted line indicates a difference of 4 to the highest performing model across number of predictors. 
    The R2D2 prior allows to increase model complexity up to the largest model without declining test performance. 
    Meanwhile, predictive performance on independent test data mostly drops when using independent normal priors, except when $p=30$ and $R^2\in\{0.5, 0.8\}$.}
    \label{fig:app-illustrative-example-adding-covariates-elpd-paths-r2-0-05-08-rho-0-p-30-100}
\end{figure}

\clearpage

\subsection{Illustrative example: Part 3. Out-of-sample predictive performance for different true $R^2$ and  correlation $\rho=0.5$}
\begin{figure}[htp!]
    \begin{subfigure}{\textwidth}
        \centering
        \includegraphics[]{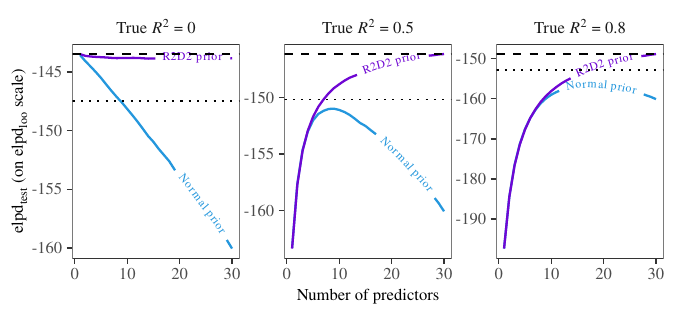}
        \caption{$p=30$}
        \label{fig:app-illustrative-example-rho-05-p-30-subfigure-1}
    \end{subfigure}
    \hfill
    \begin{subfigure}{\textwidth}
        \centering 
        \includegraphics[]{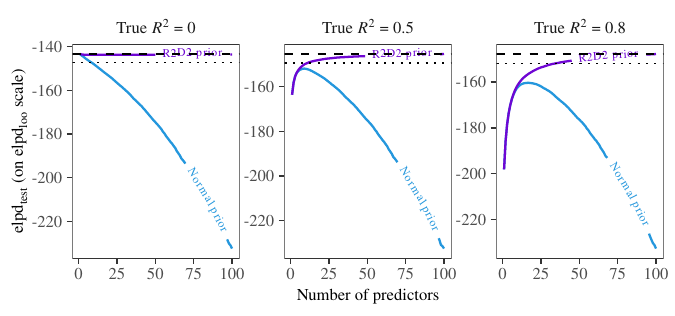}
        \caption{$p=100$}
        \label{fig:app-illustrative-example-rho-05-p-100-subfigure-2}
    \end{subfigure}
    \caption{Illustrative example: Part 3. We compare out-of-sample predictive performance ($\mathrm{elpd}_{\text{test}}$) with normal priors or R2D2 prior with true $R^2 \in \{0, 0.5, 0.8\}$ (columns). 
    The subfigures show results for a DGP with $p \in \{30, 100\}.$
    Results are averaged over $500$ repetitions and each based on $100$ observations with correlated predictors with $\rho = 0.5$ and independent test data with $n_{\text{test}} = 2000$. 
    Results are on $\mathrm{elpd}_{\text{loo}}$ scale and a dotted line indicates a difference of 4 to the highest performing model across number of predictors.
    The R2D2 prior allows to increase model complexity up to the largest model without declining test performance while it drops with increasing number of predictors when using independent normal priors.}
    \label{fig:app-illustrative-example-adding-covariates-elpd-paths-r2-0-05-08-rho-05-p-30-100}
\end{figure}

\clearpage 

\subsection{Illustrative example: Part 3. Out-of-sample predictive performance for different true $R^2$ and correlation $\rho=0.9$}
\begin{figure}[htp!]
    \begin{subfigure}{\textwidth}
        \centering
        \includegraphics[]{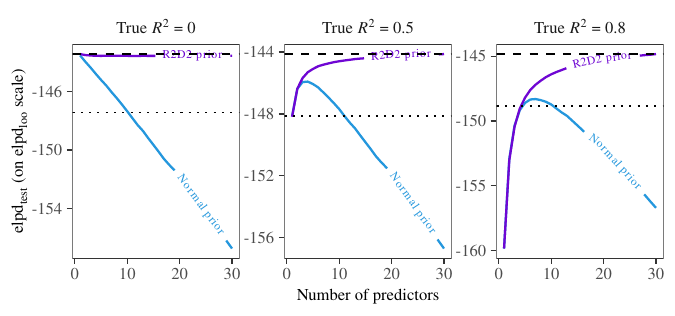}
        \caption{$p=30$}
        \label{fig:app-illustrative-example-rho-09-p-30-subfigure-1}
    \end{subfigure}
    \hfill
    \begin{subfigure}{\textwidth}
        \centering 
        \includegraphics[]{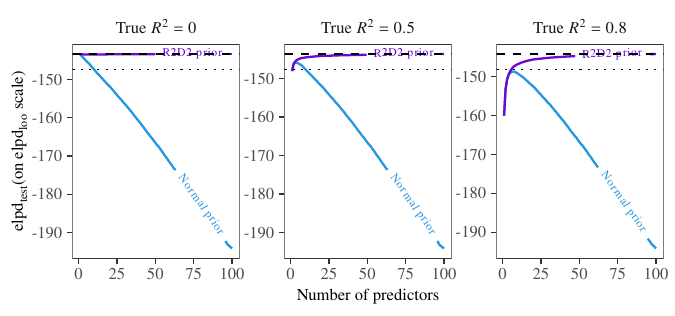}
        \caption{$p=100$}
        \label{fig:app-illustrative-example-rho-09-p-100-subfigure-2}
    \end{subfigure}
    \caption{Illustrative example: Part 3.
    We compare out-of-sample predictive performance ($\mathrm{elpd}_{\text{test}}$) with normal priors or an R2D2 prior with true $R^2 \in \{0, 0.5, 0.8\}$ (columns). 
    The subfigures show results for a DGP with $p \in \{30, 100\}.$
    Results are averaged over $500$ repetitions and each based on $100$ observations with correlated predictors with $\rho = 0.9$ and independent test data with $n_{\text{test}} = 2000$. 
    Results are on $\mathrm{elpd}_{\text{loo}}$ scale and a dotted line indicates a difference of 4 to the highest performing model across  number of predictors.
    The R2D2 prior allows to increase model complexity up to the largest model without declining test performance while it declines with increasing number of predictors when using independent normal priors.}
    \label{fig:app-illustrative-example-adding-covariates-elpd-paths-r2-0-05-08-rho-09-p-30-100}
\end{figure}

\clearpage

\subsection{Illustrative example: Part 3. Selected model sizes with different prior specification}
\begin{figure}[htp!] 
    \centering
    \includegraphics[]{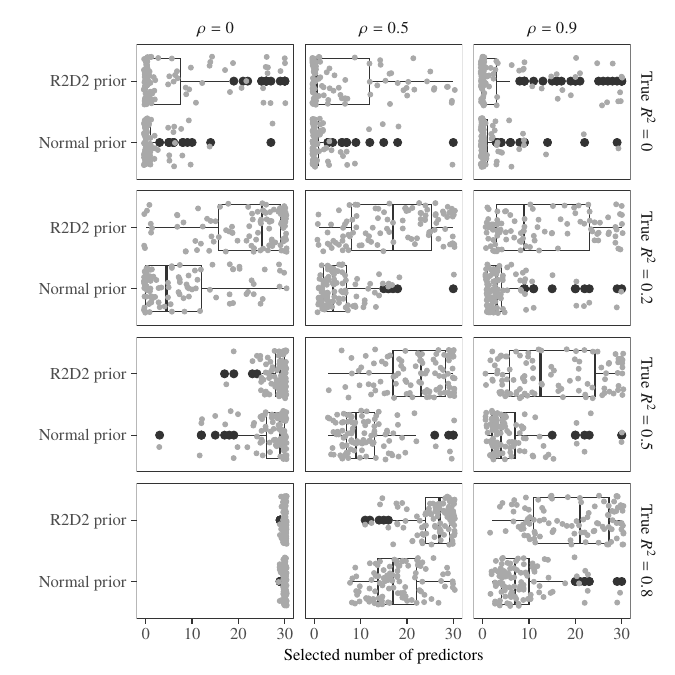}
    \vspace{-0.3cm}
    \caption{Illustrative example: Part 3. Adding covariates in linear regression. We compare selected model sizes across $100$ repetitions when using $\max\left(\mathrm{elpd}_{\text{loo}} \right)$ as a selection rule.
    With independent normal priors, model sizes tend to be smaller than with R2D2 prior and more concentrated across repetitions.
    With R2D2 prior, the results are more varied across the range of model sizes. We can also select smaller models, especially if true $R^2$ is small and correlation $\rho$ is high, but we tend to select larger models, especially when $R^2 \neq 0$ and $\rho \in \{0, 0.5\}$.}
    \label{fig:app-multiple-covariates-selection-max-elpd-loo-fixed-beta}
\end{figure}

\section{Experiments}
\subsection{Experiment 2: Forward variable selection}\label{app:ex-2-forward-search}
\begin{figure}[tp!]
    \begin{subfigure}{\textwidth}
        \centering
        \includegraphics[]{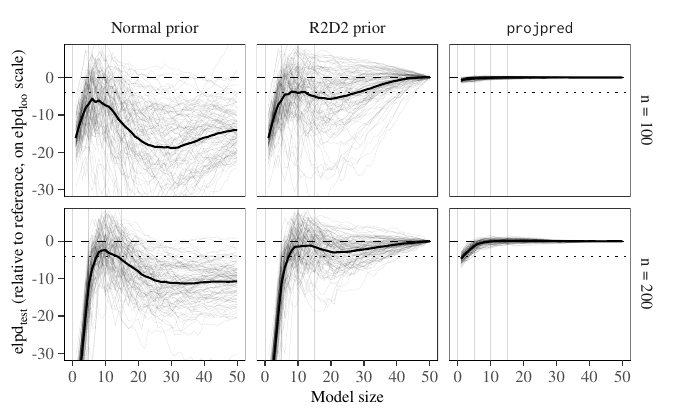}
        \vspace{-0.3cm}
        \caption{Block correlation $\rho = 0$}
        \label{fig:forward-search-rho-0-subfigure-1}
    \end{subfigure}
    \hfill
    \begin{subfigure}{\textwidth}
        \centering 
        \includegraphics[]{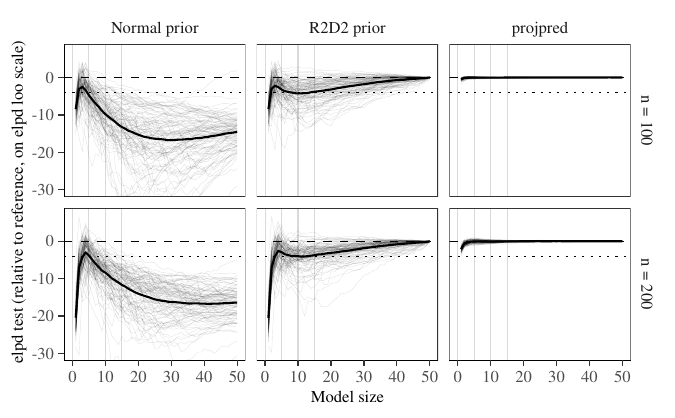}
        \vspace{-0.3cm}
        \caption{Block correlation $\rho = 0.9$}
        \label{fig:forward-search-rho-09-subfigure-2}
    \end{subfigure}
    \vspace{-0.7cm}
    \caption{Experiment 2: Forward selection. 
    $\mathrm{elpd}_{\text{test}}$ (average: thicker line, repetitions: thinner lines) with $n_{\text{test}}=2000$ 
    for forward search with (1) independent normal priors, (2) R2D2 prior and (3) projection predictive selection via \texttt{projpred} (columns), excluding the intercept-only model.
    Training data is generated using $p=50$, true $R^2=0.5$, $n \in \{100, 200\}$ (rows) 
    and correlation $\rho \in \{0,0.9\}$ (subfigures (a) and (b)).
    We observe larger variation across repetitions when using independent normal priors.
    Models with R2D2 prior reach reference model performance for $p=50$, while $\mathrm{elpd}_{\text{test}}$ remains clearly below with normal priors.}
    \label{fig:app-ex-2-forward-search-elpd-test-paths-block-structure-with-reps}
\end{figure}

\clearpage

\subsubsection{Out-of-sample performance with data generated with equally weakly relevant predictors}

\begin{figure}[htp!]
    \centering
    \includegraphics[]{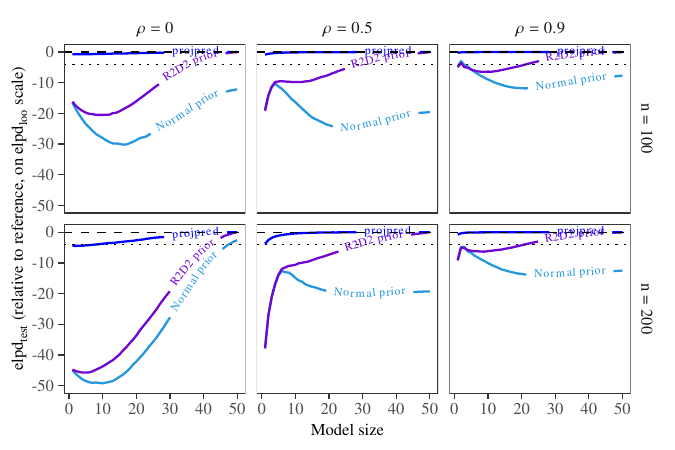}
    \caption{Experiment 2: Forward selection. Predictive performance on independent test data ($\mathrm{elpd}_{\text{test}}$) with $n_{\text{test}}=2000$ on $\mathrm{elpd}_{\text{loo}}$ scale, relative to the reference model and averaged over $100$ repetitions. 
    We compare forward search with (1) independent normal priors, (2) R2D2 prior and (3) projection predictive inference via \texttt{projpred}.  
    We show results for $n \in \{100, 200\}$ (rows) and $p=50$ equally weakly relevant covariates with correlation $\rho \in \{0, 0.5, 0.9\}$ (columns), excluding the intercept-only model. 
    True $R^2=0.5$ and a dotted line points out a difference of 4 relative to the best-performing model.
    Forward search with R2D2 prior consistently reaches \texttt{projpred}-level performance, while models with independent normal priors do not in most cases, unless $\rho=0$ and $n=200$ and model size is large or $\rho=0.9$ and $n=100$ and model size is small.}
    \label{fig:app-ex-2-forward-search-average-elpd-test-paths-flat-structure}
\end{figure}

\begin{figure}[htp!]
    \begin{subfigure}{\textwidth}
        \centering
        \includegraphics[]{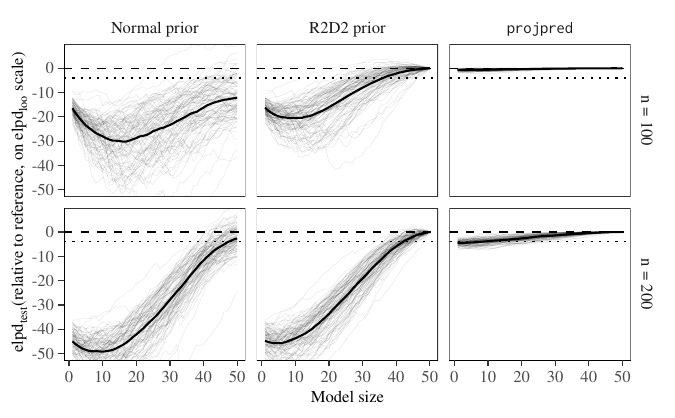}
        \vspace{-0.3cm}
        \caption{Correlation $\rho = 0$}
        \label{fig:app-ex-4-subfigure-1}
    \end{subfigure}
    \hfill
    \begin{subfigure}{\textwidth}
        \centering 
        \includegraphics[]{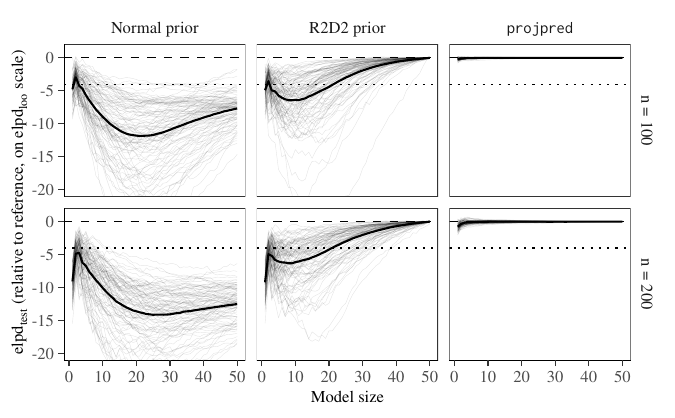}
        \vspace{-0.3cm}
        \caption{Correlation $\rho = 0.9$}
        \label{fig:app-ex-4-subfigure-2}
    \end{subfigure}
    \vspace{-0.7cm}
    \caption{Experiment 2: Forward selection. $\mathrm{elpd}_{\text{test}}$ (average: thicker line, repetitions: thinner lines) with $n_{\text{test}}=2000$ on $\mathrm{elpd}_{\text{loo}}$ scale, relative to the reference model. 
    We compare forward search with (1) independent normal priors, (2) R2D2 prior and (3) projection predictive inference via \texttt{projpred} (columns).
    We show results with $n \in \{100, 200\}$ (rows) and $p=50$ equally weakly relevant covariates with correlation $\rho \in \{0,0.9\}$ and true $R^2=0.5$, excluding the intercept-only model. 
    Note the different range of values on the $y$-axis for different correlation. 
    Variation across repetitions is larger for $\rho=0$ and overall largest for forward search with independent normal priors.}
    \label{fig:app-ex-2-forward-search-mean-elpd-test-elpd-loo-paths-data-structure-flat}
\end{figure}

\clearpage 

\subsection{Experiment 3: Increasing the complexity of nonlinear models}\label{app:ex-3-nonlinear}
\subsubsection{Considered nonlinear modelling approaches}\label{app:ex-3-nonlinear-modelling-approaches}
We assume $y_i \sim \normal\left(f(x_i), \sigma^2\right)$ with $i=1, \cdots, n$ and a-priori $\beta_j \sim \normal(0, 1)$ with $j=0, \cdots 19$. 
We compare five modelling approaches: 
\begin{enumerate}[\textbullet]
    \item Polynomial regression with degree $k$, that is, $f(x_i) = \sum_{j = 0}^k \beta_j x_i^{j}$. We consider 
    \begin{enumerate}[label = (\arabic*), nosep]
        \item raw polynomial regression, 
        \item orthogonal polynomials, 
        where we first obtain orthonormalised columns $b_j(x)$ from the raw components $\{1, x,x^2, \cdots, x^k\}$ such that $\sum_{i=1}^n b_j(x_i) b_{j'}(x_i) = 0 \ \text{for} \ j \neq j'$ and then $f(x_i) = \sum_{j=0}^{k} \beta_j b_j(x_i)$ \citep{kennedy_statistical_1980}.
    \end{enumerate}
    \item Thin plate spline (TPS) regression \citep[][]{wood_thin_2003, wood_generalized_2017, wood_generalized_2025, hastie_elements_2009} with rank $k$, that is, $f(x_i) = \sum_{j=1}^q \beta_j \nu_j(x_i) + \sum_{j=1}^r u_j \omega_j(x_i)$ where $\nu_j$ are the non-penalised polynomial components (here: $\nu_1(x_i) = 1$ and $\nu_2(x_i) = x_i$) and $\omega_j$ form the penalised basis. 
    We use \texttt{mgcv::s()} \citep[][]{wood_thin_2003, wood_generalized_2017} in \texttt{brms::brm()} \citep[][]{burkner_brms_2017}. 
    Smoothness is controlled by penalising $u$ which corresponds to a prior $u \mid \sigma_s \sim \normal(0, \sigma_s^2 S^{-1})$ with $r\times r$ penalty matrix $S$.
    We consider 
    \begin{enumerate}[nosep]
        \item[(3)] TPS with $\sigma_s=\infty$ (i.e., no penalisation, fixed degrees of freedom), and
        \item[(4)] TPS with prior for $\sigma_s \sim \text{Student-t}^+(3, 0, 2.5)$ (i.e., non-fixed penalisation).
    \end{enumerate}
    \item Hilbert space approximated Gaussian process regression
    \citep[HSGP,][]{solin_hilbert_2020, riutort-mayol_practical_2022}  as introduced before in Section \ref{subsec:methods-function-priors}, that is, 
    \begin{equation}
        \begin{aligned}
            f(x_i) \approx \sum_{j=1}^k S_{\theta} \left( \sqrt{\lambda_j} \right)^{\frac{1}{2}} \phi_j(x_i) \beta_j 
            \quad \text{     with } & \beta_j \sim \normal(0, 1),  \\
        \end{aligned}
        \label{eq:hilbert-approx-2}
    \end{equation}
    and eigenvectors $\phi_j$ and eigenvalues $\lambda_j$ of the covariance operator of the kernel. 
    We assume $\ell \sim \text{InvGamma}(2,1)$, $\sigma^2 \sim \text{InvGamma}(2, 1)$
    and use the HSGP implementation in \texttt{brms} \citep[][]{burkner_brms_2017}. 
\end{enumerate}

\clearpage

\subsubsection{Observed data versus posterior predictions}
\begin{figure}[htp!]
    \centering
    \includegraphics[width=\linewidth]{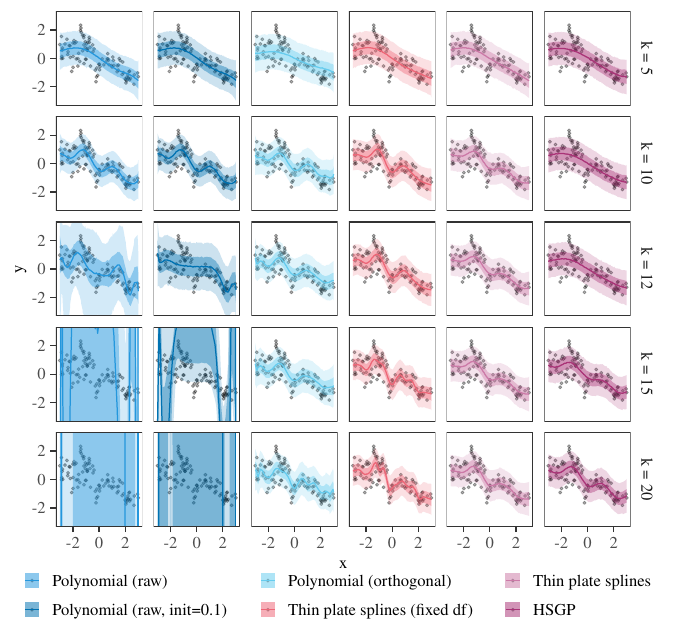}
    \caption{Experiment 3: Increasing the complexity of nonlinear models. 
    Observed data ($n=100$) and posterior results for six different modelling approaches (columns) with increasing degree/number of basis functions $k \in\{5, 10, 12, 15, 20\}$ (rows). 
    We compare a raw polynomial (1) without and (2) with improved initial value, (3) an orthogonal polynomial, as well as thin plate splines (4) with fixed degrees of freedom and (5) with non-fixed penalisation, and (6) an HSGP. 
    The data is generated similar to the drowning data set used, for example, by \citet[][]{gelman_bayesian_2015}. }
    \label{fig:app-ex-3-nonlinear-obs-posterior-results}
\end{figure}

\clearpage

\subsubsection{Basis values for different modelling approaches}

\begin{figure}[htp!]
    \centering
    \includegraphics[]{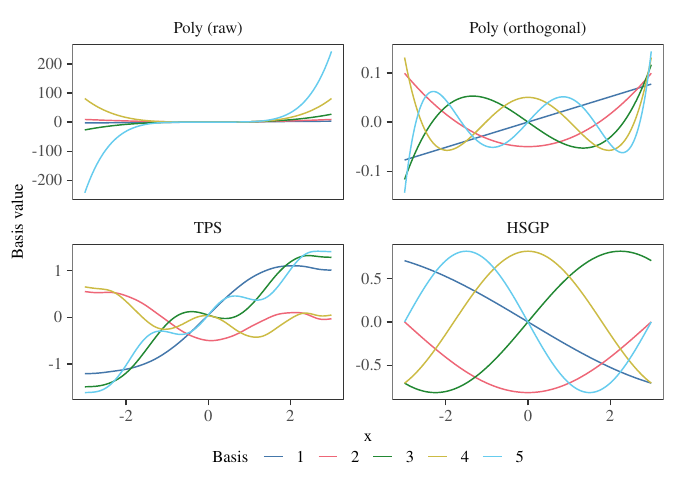}
    \vspace{-0.3cm}
    \caption{Experiment 3: Increasing the complexity of nonlinear models. We compare basis values of the first five bases for a raw polynomial (Poly (raw)), an orthogonal polynomial (Poly (orthogonal)), as well as thin plate splines (TPS) and HSGPs. The range of the basis values for the raw polynomial is much larger and increases much faster with more bases than the basis values for the other modelling approaches (compare y-axis scales).}
    \label{app:fig-ex-3-nonlinear-basis-function-values}
\end{figure} 

\clearpage

\subsubsection{Test performance for $n\in\{20, 50, 100\}$}

\begin{figure}[htp!]
    \centering
    \includegraphics[]{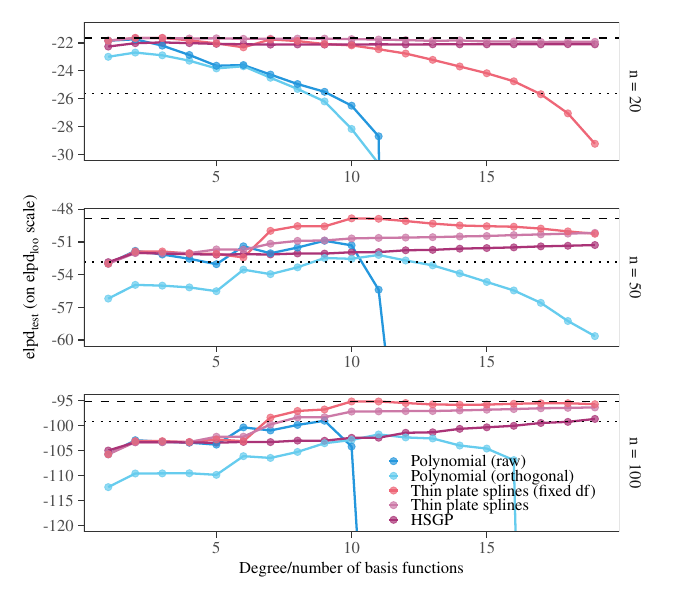}
    \caption{Experiment 3: Increasing the complexity of nonlinear models. For raw and orthogonal polynomials, $\mathrm{elpd}_{\text{test}}$ drops at higher degrees.
    For $n=20$, also non-penalised TPS (with fixed degrees of freedom) deteriorate for higher degrees. 
    Only penalised TPS and HSGP remain stable as $k$ increases.}
    \label{fig:app-ex-3-nonlinear-obs-elpd-paths-n-20-50-100}
\end{figure}

\clearpage

\subsubsection{Computational issues for the different modelling approaches}
\begin{figure}[htp!]
    \centering
    \includegraphics[width=\linewidth]{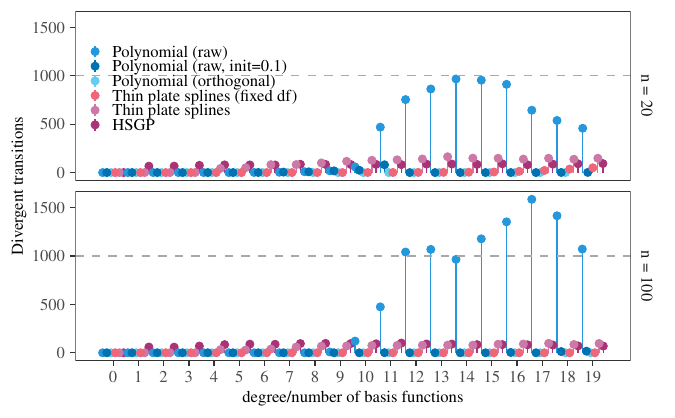}
    \caption{Experiment 3: Increasing the complexity of nonlinear models. Number of divergent transitions for each of the modelling approaches averaged across repetitions for each condition. We run four chains with $1000$ iterations. 
    The raw polynomials have large numbers of divergent transitions from degree $11$ and higher. 
    Better initialisation or orthogonalisation seems to resolve issues with divergent transitions, but, as illustrated in Figure \ref{fig:app-ex-3-nonlinear-obs-elpd-paths-n-20-50-100}, the models still exhibit sharp drops in predictive performance on independent test data after some level of model complexity.}
    \label{fig:app-ex-3-nonlinear-div-transitions}
\end{figure}

\clearpage

\subsection{Experiment 4: Treatment Effect}\label{app:treatment-study}

\subsubsection{Proofs for the Treatment Effect Experiment}\label{app:treatment-proofs}

We provide formal derivations for the results stated in Section~\ref{subsec:ex-4-rct-studies}. Consider the model in Equation~\eqref{eq:rct-model} in vector form: $y = \alpha z + X\beta + \varepsilon$ with $\varepsilon \sim \normal(0, \sigma^2 I_n)$, and independent priors $\alpha \sim \normal(0, \tau_\alpha^2)$ and $\beta \sim \normal(0, \Sigma_\beta)$.

The following assumptions are maintained throughout all propositions in this section unless stated otherwise.

\paragraph{Maintained assumptions.}
\begin{enumerate}[label=(A\arabic*),nosep]
    \item \label{ass:model} \textbf{Well-specified model.} The data-generating process is $y = \alpha^* z + X\beta^* + \varepsilon$ with $\varepsilon \sim \normal(0, \sigma^2 I_n)$, and the model is the normal linear model in Equation~\eqref{eq:rct-model}. In particular, the likelihood is correctly specified.
    \item \label{ass:treatment} \textbf{Randomised treatment.} Treatment assignments $z_i \iidsim \mathrm{Bernoulli}(0.5)$ independently of $X$ and $\varepsilon$.
    \item \label{ass:covariates} \textbf{Covariate regularity.} The columns of $X \in \mathbb{R}^{n \times p}$ are drawn independently from a sub-Gaussian distribution with zero mean and unit variance (the standard normal is the leading case). In particular, columns of $X$ are independent of $z$.
    \item \label{ass:priors} \textbf{Gaussian conditional priors.} Conditional on all hyperparameters, $\alpha \sim \normal(0, \tau_\alpha^2)$ and $\beta \sim \normal(0, \Sigma_\beta)$ independently, with $\Sigma_\beta$ positive definite. All derivations condition on $\sigma^2$.
    \item \label{ass:rank} \textbf{Rank condition.} $\mathrm{rank}(X) = \min(n,p)$ almost surely, which holds with probability one under Assumption~\ref{ass:covariates} \citep{vershynin_high-dimensional_2018}.
\end{enumerate}

\begin{proposition}[Marginal posterior variance of $\alpha$]\label{prop:marginal-var}
The conditional posterior $(\alpha, \beta) \mid y, X, z, \sigma^2$ is multivariate normal. The marginal posterior variance of $\alpha$ is
\begin{equation}
    \Var(\alpha \mid y, X, z, \sigma^2) = \left[ \frac{1}{\sigma^2} z^T (I_n - H_\beta) z + \tau_\alpha^{-2} \right]^{-1},
\end{equation}
where $H_\beta = X(X^TX + \sigma^2\Sigma_\beta^{-1})^{-1}X^T$, and the posterior mean is
\begin{equation}
    \mathbb{E}[\alpha \mid y, X, z, \sigma^2] = \Var(\alpha \mid y, X, z, \sigma^2) \cdot \frac{1}{\sigma^2} z^T(I_n - H_\beta) y. \label{eq:posterior-mean-alpha}
\end{equation}
\end{proposition}

\textit{Proof.}
Write $y \sim \normal(W\vartheta, \sigma^2 I_n)$ with $W = [z \; X] \in \mathbb{R}^{n \times (p+1)}$ and $\vartheta = (\alpha, \beta^T)^T$. The prior is $\vartheta \sim \normal(0, \Sigma_\vartheta)$ with $\Sigma_\vartheta = \mathrm{diag}(\tau_\alpha^2, \Sigma_\beta)$. The posterior precision matrix is
\begin{equation}
    \Lambda = \sigma^{-2} W^T W + \Sigma_\vartheta^{-1}
    = \begin{pmatrix} \sigma^{-2} z^Tz + \tau_\alpha^{-2} & \sigma^{-2} z^T X \\ \sigma^{-2} X^T z & \sigma^{-2} X^T X + \Sigma_\beta^{-1} \end{pmatrix}.
\end{equation}
Denote the blocks as $\Lambda_{11} = \sigma^{-2} z^Tz + \tau_\alpha^{-2}$, $\Lambda_{12} = \sigma^{-2} z^T X$, $\Lambda_{22} = \sigma^{-2} X^T X + \Sigma_\beta^{-1}$. By the Schur complement formula for block matrices, the $(1,1)$ entry of $\Lambda^{-1}$ is
\begin{equation}
    \Var(\alpha \mid \cdot) = \left(\Lambda_{11} - \Lambda_{12}\Lambda_{22}^{-1}\Lambda_{21}\right)^{-1}.
\end{equation}
Computing the off-diagonal product:
\begin{align}
    \Lambda_{12}\Lambda_{22}^{-1}\Lambda_{21}
    &= \sigma^{-4} z^T X \left(\sigma^{-2} X^TX + \Sigma_\beta^{-1}\right)^{-1} X^T z \nonumber\\
    &= \sigma^{-4} z^T X \cdot \sigma^2 \left(X^TX + \sigma^2 \Sigma_\beta^{-1}\right)^{-1} X^T z \nonumber\\
    &= \sigma^{-2} z^T \underbrace{X \left(X^TX + \sigma^2 \Sigma_\beta^{-1}\right)^{-1} X^T}_{= H_\beta} z.
\end{align}
Substituting into the Schur complement:
\begin{equation}
    \Lambda_{11} - \Lambda_{12}\Lambda_{22}^{-1}\Lambda_{21} = \sigma^{-2} z^Tz - \sigma^{-2} z^T H_\beta z + \tau_\alpha^{-2} = \sigma^{-2} z^T(I_n - H_\beta)z + \tau_\alpha^{-2}.
\end{equation}
For the posterior mean, the standard formula for partitioned normal posteriors gives:
\begin{align}
    \mathbb{E}[\alpha \mid \cdot]
    &= \Var(\alpha \mid \cdot) \left[\sigma^{-2} z^T y - \Lambda_{12}\Lambda_{22}^{-1} \cdot \sigma^{-2} X^T y\right] \nonumber\\
    &= \Var(\alpha \mid \cdot) \cdot \sigma^{-2} \left[z^T y - z^T H_\beta y \right] \nonumber\\
    &= \Var(\alpha \mid \cdot) \cdot \sigma^{-2} z^T(I_n - H_\beta) y. \tag*{$\square$}
\end{align}

Throughout the remainder of this section we use the following notation: $H_\beta = X(X^TX + \sigma^2 \Sigma_\beta^{-1})^{-1}X^T$, $d_z = \sigma^{-2} z^T(I_n - H_\beta)z$, and $s = d_z/(d_z + \tau_\alpha^{-2})$. We write $\nu_1,\dots,\nu_{\min(n,p)}$ for the non-zero eigenvalues of $X^TX$ with associated left singular vectors $u_1,\dots,u_{\min(n,p)}$, and $P_X$ for the projection onto $\mathrm{col}(X)$. The statements that follow consider two regimes:
\begin{enumerate}[label=(\roman*),nosep]
    \item \textbf{Finite-sample}: $n$ fixed, $p$ increasing towards $n$ --- analysed non-asymptotically via exact identities, monotonicity in $p$, and explicit finite-$n$ bounds;
    \item \textbf{Proportional}: $p = \lfloor \gamma n \rfloor$ with $\gamma \in (0,1)$, $n \to \infty$ --- where the random-matrix $\Theta/o_p$ rates are well defined.
\end{enumerate}

\begin{lemma}[Basis-free identity for the residualised data precision]\label{lem:exact-data-prec}
Under Assumptions~\ref{ass:treatment} and~\ref{ass:priors},
\begin{equation}\label{eq:exact-data-prec}
    \mathbb{E}[d_z \mid X, \Sigma_\beta] = \frac{1}{4\sigma^2}\bigl[\bigl(n - \mathrm{tr}(H_\beta)\bigr) + \mathbf{1}_n^T(I_n - H_\beta)\mathbf{1}_n\bigr],
\end{equation}
with $\mathrm{tr}(H_\beta) \in [0, p)$ strictly whenever $\Sigma_\beta^{-1} \succ 0$.
\end{lemma}

\textit{Proof.} By Assumption~\ref{ass:treatment}, $z_i \sim \mathrm{Bernoulli}(0.5)$ independently, so $\mathbb{E}[z] = \frac{1}{2}\mathbf{1}_n$ and $\mathrm{Cov}(z) = \frac{1}{4}I_n$. Applying $\mathbb{E}[z^T A z] = \mathbb{E}[z]^T A \mathbb{E}[z] + \mathrm{tr}(A \cdot \mathrm{Cov}(z))$ with $A = I_n - H_\beta$ yields \eqref{eq:exact-data-prec}. For any $\Sigma_\beta \succ 0$, $H_\beta = X(X^TX + \sigma^2 \Sigma_\beta^{-1})^{-1}X^T$ satisfies $0 \preceq H_\beta \prec P_X$ (the strict upper bound follows from $\sigma^2 \Sigma_\beta^{-1} \succ 0$), so all non-zero eigenvalues of $H_\beta$ lie in $[0,1)$ and $\mathrm{tr}(H_\beta) < p$ strictly. The argument is basis-free and does not require $\Sigma_\beta^{-1}$ and $X^TX$ to commute. \hfill $\square$

\subsubsection*{Finite-sample regime: \texorpdfstring{$n$ fixed, $p \to n$}{n fixed, p to n}}\label{app:treatment-proofs-i}

The propositions in this subsection characterise $\mathbb{E}[d_z \mid X, \phi]$, the prior precision $\tau_\alpha^{-2}$, the shrinkage factor $s = d_z/(d_z + \tau_\alpha^{-2})$, the posterior precision $d_z + \tau_\alpha^{-2}$, and the posterior variance $(d_z + \tau_\alpha^{-2})^{-1}$ in regime (i).

\begin{proposition}[Normal prior, regime (i)]\label{prop:rct-normal}
Suppose $\Sigma_\beta = I_p$ and $\tau_\alpha^{-2} = \mathcal{O}(1)$. The eigendecomposition of $H_\beta$ gives
\begin{equation}\label{eq:normal-d_z-decomp}
    \mathbb{E}[d_z \mid X] = \frac{1}{4\sigma^2}\sum_{j=1}^p \frac{\sigma^2}{\nu_j+\sigma^2}\bigl[1 + (u_j^T \mathbf{1}_n)^2\bigr] + \frac{n-p}{4\sigma^2} + \frac{\|(I_n-P_X)\mathbf{1}_n\|^2}{4\sigma^2}.
\end{equation}
As $p \to n$, the second and third terms vanish but the regularisation sum (first term) is strictly positive: the ridge penalty $\sigma^2 I_p$ keeps $H_\beta \prec I_n$ even at $p = n$. Hence
\begin{enumerate}[label=(\alph*),nosep]
    \item $\mathbb{E}[d_z \mid X]$ is strictly positive and bounded above by $n/(2\sigma^2)$;
    \item the posterior precision $d_z + \tau_\alpha^{-2}$ lies in $(\tau_\alpha^{-2}, \infty)$ bounded;
    \item $s = d_z/(d_z + \tau_\alpha^{-2})$ lies in $(0,1)$ bounded away from both endpoints;
    \item the posterior variance $(d_z + \tau_\alpha^{-2})^{-1}$ lies in $(0, \tau_\alpha^2)$ without attaining $\tau_\alpha^2$.
\end{enumerate}
\end{proposition}

\textit{Proof.} For $\Sigma_\beta = I_p$, $\Sigma_\beta^{-1} = I_p$ commutes with $X^TX$, so $H_\beta$ is co-diagonalisable with $X^TX$ and its non-zero eigenvalues equal $\nu_j/(\nu_j+\sigma^2)$. Hence $\mathrm{tr}(H_\beta) = \sum_j \nu_j/(\nu_j+\sigma^2)$ and $\mathbf{1}_n^T(I_n - H_\beta)\mathbf{1}_n = \sum_j [\sigma^2/(\nu_j+\sigma^2)](u_j^T \mathbf{1}_n)^2 + \|(I_n - P_X)\mathbf{1}_n\|^2$. Substituting into Lemma~\ref{lem:exact-data-prec} gives \eqref{eq:normal-d_z-decomp}. Each summand $\sigma^2/(\nu_j+\sigma^2)$ lies in $(0,1]$ and is strictly positive whenever $\nu_j < \infty$, so the regularisation sum is $> 0$ for every $p \leq n$. As $p \to n$, $(n-p) \to 0$ and $\|(I_n - P_X)\mathbf{1}_n\|^2 \to 0$, so $\mathbb{E}[d_z \mid X] > 0$ throughout. The upper bound $\mathbb{E}[d_z \mid X] \le n/(2\sigma^2)$ follows from $\mathrm{tr}(H_\beta) \ge 0$ (so $n-\mathrm{tr}(H_\beta)\le n$) and $\mathbf{1}_n^T(I_n-H_\beta)\mathbf{1}_n \le \mathbf{1}_n^T\mathbf{1}_n = n$ (as $H_\beta \succeq 0$). The conclusions on posterior precision, $s$, and posterior variance follow by direct substitution into $s = d_z/(d_z + \tau_\alpha^{-2})$ and $(d_z + \tau_\alpha^{-2})^{-1}$, using $\tau_\alpha^{-2} = \mathcal{O}(1)$. \hfill $\square$

\begin{proposition}[Joint R2D2 prior, regime (i)]\label{prop:rct-joint}
Suppose $(\alpha,\beta)$ share an R2D2 prior with weights $\phi = (\phi_\alpha,\phi_1,\dots,\phi_p) \sim \mathrm{Dirichlet}(a,\dots,a)$ of dimension $p+1$ for some $a > 0$, so $\tau_\alpha^2 = \tau^2 \phi_\alpha$ and $\beta_j \sim \normal(0,\tau^2 \phi_j)$. By Dirichlet aggregation, $\phi_\alpha \sim \mathrm{Beta}(a, pa)$. Then:
\begin{enumerate}[label=(\alph*),nosep]
    \item the conditional prior precision $\tau_\alpha^{-2} = 1/(\tau^2 \phi_\alpha)$ grows linearly in $p$: for $a > 1$, the exact Beta inverse moment gives $\mathbb{E}[\tau_\alpha^{-2}] = (a(p+1)-1)/[(a-1)\tau^2] = \Theta(p/\tau^2)$; for any $a > 0$, the in-probability statement $\tau_\alpha^{-2} = \Theta_p(p/\tau^2)$ holds (for $a \le 1$ the moment $\mathbb{E}[\tau_\alpha^{-2}]$ diverges due to mass near $\phi_\alpha = 0$, but the in-probability statement is unaffected);
    \item conditional on $\phi$, the prior on $\beta$ provides ridge-type regularisation with effective per-coordinate penalty $\sigma^2/(\tau^2 \phi_j)$, so the eigenvalues of $H_\beta = X(X^TX + \sigma^2 \mathrm{diag}(1/(\tau^2\phi_j)))^{-1}X^T$ are strictly below $1$, and $\mathbb{E}[d_z \mid X, \phi]$ is strictly positive and bounded as $p \to n$;
    \item the shrinkage factor $s = d_z/(d_z + \tau_\alpha^{-2})$ is stochastically decreasing in $p$ at stochastic scale $s = \Theta_p(1/p)$ (no limit is taken at fixed $n$): quantitatively, $1 - s \geq c\,p/(c\,p + d_{z,\max})$, where $c = \Theta_p(1/\tau^2)$ and $d_{z,\max} = \sup_{p \le n} \mathbb{E}[d_z \mid X, \phi] < \infty$, so $s$ attains its strictly positive minimum at $p = n$;
    \item the posterior precision $d_z + \tau_\alpha^{-2} = \Theta_p(p)$ is dominated by the prior contribution, so the posterior variance is $(d_z + \tau_\alpha^{-2})^{-1} = \Theta_p(1/p)$. The posterior mean $s\,\hat\alpha_{\text{OLS}_\beta}$ is centred at a fraction $s$ of $\hat\alpha_{\text{OLS}_\beta}$ that decreases stochastically in $p$. In the proportional regime this fraction stays bounded away from $1$ in probability (Proposition~\ref{prop:rct-joint-prop}), leaving a non-vanishing multiplicative bias rather than a collapse to zero.
\end{enumerate}
\end{proposition}

\textit{Proof.} For (a), $\mathbb{E}[\phi_\alpha] = 1/(p+1)$ and, for $a > 1$, the standard Beta inverse moment $\mathbb{E}[1/\phi_\alpha] = (a + pa - 1)/(a - 1)$ gives $\mathbb{E}[\tau_\alpha^{-2}] = (a(p+1)-1)/[(a-1)\tau^2]$. For any $a > 0$, $\Var(\phi_\alpha) = p/[(p+1)^2(pa+a+1)]$, so the squared coefficient of variation is $\Var(\phi_\alpha)/\mathbb{E}[\phi_\alpha]^2 = p/(pa+a+1) \to 1/a$, a positive constant. The weight $\phi_\alpha$ therefore does \emph{not} concentrate at its mean for fixed $a$. Instead $(p+1)\phi_\alpha \to \mathrm{Gamma}(a,a)$ in distribution, a non-degenerate positive limit, so $(p+1)\phi_\alpha = \Theta_p(1)$ and $\tau_\alpha^{-2} = 1/(\tau^2\phi_\alpha) = \Theta_p(p/\tau^2)$. The order statement holds for every $a > 0$, reflecting that the limit is bounded away from $0$ and $\infty$ in probability rather than concentration of $\phi_\alpha$ at $1/(p+1)$. For (b), conditional on any draw of $\phi$ with $\phi_j > 0$, $\Sigma_\beta^{-1} = \mathrm{diag}(1/(\tau^2\phi_j)) \succ 0$. Since $\Sigma_\beta^{-1}$ is heterogeneous, it does not commute with the dense $X^TX$, so $H_\beta$ does not share the eigenbasis of $X^TX$ at finite $p$ and the per-coordinate formula $\nu_j/(\nu_j + \sigma^2/(\tau^2\phi_j))$ is not literally an eigenvalue identity. We use a Weyl-type sandwich (Weyl's monotonicity inequality for self-adjoint matrices): setting $c_{\min} = \min_j 1/(\tau^2\phi_j)$ and $c_{\max} = \max_j 1/(\tau^2\phi_j)$, the operator inequality $c_{\min} I_p \preceq \Sigma_\beta^{-1} \preceq c_{\max} I_p$ implies
\begin{equation}\label{eq:weyl-sandwich}
    \frac{\nu_j}{\nu_j + \sigma^2 c_{\max}} \;\leq\; \lambda_j(H_\beta) \;\leq\; \frac{\nu_j}{\nu_j + \sigma^2 c_{\min}} \quad (j = 1,\dots,p),
\end{equation}
where $\lambda_j(H_\beta)$ are the non-zero eigenvalues of $H_\beta$ in $\mathrm{col}(X)$ (the remaining $n-p$ eigenvalues are $0$). Since $\sigma^2 c_{\min} > 0$, every $\lambda_j(H_\beta) < 1$ strictly, so by Lemma~\ref{lem:exact-data-prec} $\mathbb{E}[d_z \mid X, \phi]$ is strictly positive and bounded as $p \to n$. For (c) and (d), since $p \le n$ is bounded and $\tau_\alpha^{-2}$ is stochastically increasing in $p$ (from (a)) while $d_z$ is bounded (from (b)), $s = d_z/(d_z + \tau_\alpha^{-2})$ is stochastically decreasing in $p$. The lower bound $1 - s = \tau_\alpha^{-2}/(\tau_\alpha^{-2} + d_z) \ge \tau_\alpha^{-2}/(\tau_\alpha^{-2} + d_{z,\max})$ holds deterministically; substituting $\tau_\alpha^{-2} = \Theta_p(p/\tau^2)$ gives the explicit form $1 - s \ge cp/(cp + d_{z,\max})$ with $c = \Theta(1/\tau^2)$. The posterior precision and variance rates follow by direct substitution. \hfill $\square$

\begin{proposition}[Heavy-shrinkage tightening for the joint R2D2 prior]\label{prop:rct-joint-approx}
Under the setting of Proposition~\ref{prop:rct-joint}, suppose the heavy-shrinkage condition
\begin{equation}\label{eq:heavy-shrinkage}
    \frac{p+1}{n} \;\gg\; \frac{R^2}{1 - R^2}
\end{equation}
holds, with $R^2 := \tau^2/(\tau^2+\sigma^2)$ the implied prior $R^2$. Then, provided \eqref{eq:heavy-shrinkage} holds up to a $\log p$ factor (made precise in the proof), $H_\beta \to 0$ in operator norm in probability over $\phi$, and Lemma~\ref{lem:exact-data-prec} reduces to
\begin{equation}\label{eq:d_z-heavy-shrinkage}
    \mathbb{E}[d_z \mid X, \phi] \;=\; \frac{n}{2\sigma^2} + o_p(n/\sigma^2).
\end{equation}
Combined with $\tau_\alpha^{-2} = \Theta_p(p/\tau^2)$ from Proposition~\ref{prop:rct-joint} and evaluation at the prior mean $\mathbb{E}[\phi_\alpha] = 1/(p+1)$, the shrinkage factor satisfies
\begin{equation}\label{eq:shrinkage-ratio-appx}
    s \;\approx\; \frac{n\tau^2}{n\tau^2 + 2\sigma^2(p+1)} \;=\; \frac{n R^2}{n R^2 + 2(1-R^2)(p+1)}.
\end{equation}
The approximation is valid for any $(n,p)$ for which \eqref{eq:heavy-shrinkage} holds; in particular it applies in regime (i) and, by Proposition~\ref{prop:rct-joint-prop} in Appendix~\ref{app:treatment-proofs-prop}, also in regime (ii) at any $\gamma \in (0,1)$ for which $R^2$ is small enough that \eqref{eq:heavy-shrinkage} is satisfied.
\end{proposition}

\textit{Proof.} We bound the eigenvalues of $H_\beta$ through the Weyl sandwich \eqref{eq:weyl-sandwich} rather than through concentration of the Dirichlet weights, which does not hold for fixed $a$ (Proposition~\ref{prop:rct-joint}(a)). With $c_{\min} = \min_j 1/(\tau^2\phi_j) = 1/(\tau^2\max_j\phi_j)$, \eqref{eq:weyl-sandwich} gives $\lambda_j(H_\beta) \le \nu_j/(\nu_j + \sigma^2 c_{\min})$ for every $j$. For a symmetric Dirichlet of dimension $p+1$, the largest weight satisfies $\max_j\phi_j = \Theta_p((\log p)/p)$ by standard order-statistics concentration of the Beta marginals, so $c_{\min} = \Theta_p(p/(\tau^2\log p))$ and the effective penalty is $\sigma^2 c_{\min} = \Theta_p\!\left(\frac{1-R^2}{R^2}\cdot\frac{p}{\log p}\right)$. Under Assumption~\ref{ass:covariates}, the largest eigenvalue of $X^TX$ satisfies $\nu_{\max} \le C n$ with high probability for some absolute constant $C$ depending only on the sub-Gaussian norm \citep[][Thm.~4.6.1]{vershynin_high-dimensional_2018}. Hence, when \eqref{eq:heavy-shrinkage} holds up to a $\log p$ factor, that is $(p+1)/n \gg (R^2/(1-R^2))\log p$, the penalty $\sigma^2 c_{\min}$ dominates $\nu_{\max}$ with probability tending to one, so $\lambda_{\max}(H_\beta) = o_p(1)$ and $H_\beta \to 0$ in operator norm. It follows that $\mathrm{tr}(H_\beta) \le p\,\lambda_{\max}(H_\beta) = o_p(p) = o_p(n)$ and $\mathbf{1}_n^T H_\beta \mathbf{1}_n \le n\,\lambda_{\max}(H_\beta) = o_p(n)$. Substituting into Lemma~\ref{lem:exact-data-prec} yields \eqref{eq:d_z-heavy-shrinkage}, since $\mathbf{1}_n^T \mathbf{1}_n = n$. For the shrinkage factor, $\tau_\alpha^{-2} = \Theta_p(p/\tau^2)$ does not concentrate (Proposition~\ref{prop:rct-joint}(a)), so \eqref{eq:shrinkage-ratio-appx} is obtained by evaluating $s = d_z/(d_z + \tau_\alpha^{-2})$ at the prior mean $\mathbb{E}[\phi_\alpha] = 1/(p+1)$, that is with $\tau_\alpha^{-2}$ replaced by $(p+1)/\tau^2$. This plug-in value is an upper bound on the expected shrinkage, since $s$ is concave in $\phi_\alpha$ and Jensen's inequality gives $\mathbb{E}[s(\phi_\alpha)] \le s(\mathbb{E}[\phi_\alpha])$. \hfill $\square$

\begin{proposition}[Split prior (R2D2 on $\beta$, independent Normal on $\alpha$), regime (i)]\label{prop:rct-split}
Suppose $\beta_j \sim \normal(0,\tau^2\phi_j)$ with $\phi \sim \mathrm{Dirichlet}(a,\dots,a)$ of dimension $p$, and $\alpha \sim \normal(0,\tau_\alpha^2)$ independently with $\tau_\alpha^{-2} = \mathcal{O}(1)$. Then $H_\beta$ has the same form as in Proposition~\ref{prop:rct-joint} (the $\beta$-weights here are Dirichlet on the $p$-simplex versus marginals of the $(p{+}1)$-simplex there, which leaves the eigenvalue orders unchanged), so the Weyl sandwich \eqref{eq:weyl-sandwich} keeps $\mathbb{E}[d_z \mid X, \phi]$ strictly positive and bounded as $p \to n$. Combined with $\tau_\alpha^{-2} = \mathcal{O}(1)$ fixed:
\begin{enumerate}[label=(\alph*),nosep]
    \item the posterior precision $d_z + \tau_\alpha^{-2}$ lies in $(\tau_\alpha^{-2}, \infty)$ bounded, $s$ is bounded in $(0,1)$, and the posterior variance is bounded above by $\tau_\alpha^2$ without attaining it;
    \item under the heavy-shrinkage condition \eqref{eq:heavy-shrinkage}, the eigenvalue analysis of Proposition~\ref{prop:rct-joint-approx} applies verbatim and tightens to $\mathbb{E}[d_z \mid X, \phi] = n/(2\sigma^2) + o_p(n/\sigma^2)$, so the posterior precision is $\Theta(n/\sigma^2)$ and the posterior variance attains the parametric rate $\Theta(\sigma^2/n)$.
\end{enumerate}
The posterior mean is centred on $\hat\alpha_{\text{OLS}_\beta}$, which is unbiased for $\alpha^*$ up to the $O_p(\sqrt{p/n})$ confounding term of Proposition~\ref{prop:coverage} (and exactly unbiased in the flat-prior limit); since $\tau_\alpha^{-2} = \mathcal{O}(1)$ keeps $s$ bounded away from $0$, $s < 1$ controls essentially only variance, not bias.
\end{proposition}

\textit{Proof.} $H_\beta$ has the same form as in Proposition~\ref{prop:rct-joint} (same eigenvalue orders), so the Weyl sandwich \eqref{eq:weyl-sandwich} gives $\mathbb{E}[d_z \mid X, \phi]$ strictly positive and bounded; the independent prior on $\alpha$ keeps $\tau_\alpha^{-2}$ fixed in $p$. Substitution into $s = d_z/(d_z + \mathcal{O}(1))$ yields (a). Part (b) follows from Proposition~\ref{prop:rct-joint-approx}: the heavy-shrinkage substitution $\Sigma_\beta^{-1} \approx (p+1)/\tau^2\cdot I_p$ and the eigenvalue-domination argument apply identically because they only use the prior on $\beta$, not on $\alpha$. \hfill $\square$

\subsubsection*{Proportional asymptotic regime: \texorpdfstring{$p = \lfloor\gamma n\rfloor$, $n \to \infty$}{p = floor(gamma n), n to infinity}}\label{app:treatment-proofs-prop}

The propositions in this subsection treat the regime $p = \lfloor\gamma n\rfloor$ with $\gamma \in (0,1)$ fixed and $n \to \infty$, complementing the finite-sample analysis of Appendix~\ref{app:treatment-proofs-i}. The qualitative pathologies of regime (i) become more pronounced: only the joint R2D2 prior produces non-vanishing shrinkage bias.

\begin{proposition}[Normal prior, regime (ii)]\label{prop:rct-normal-prop}
Under the setting of Proposition~\ref{prop:rct-normal} and Assumption~\ref{ass:covariates}, the empirical spectral distribution of $n^{-1}X^TX$ converges to the Marchenko--Pastur law on $[(1-\sqrt\gamma)^2,(1+\sqrt\gamma)^2]$ \citep{marchenko_distribution_1967, bai_spectral_2010}, so $\nu_j = \Theta(n)$ uniformly in $j$ and the regularisation sum in \eqref{eq:normal-d_z-decomp} is $O(p/n) = O(\gamma)$. Two terms of \eqref{eq:normal-d_z-decomp} are of order $n$: the residual-degrees-of-freedom term $(n-p)/(4\sigma^2)$ and the treatment-mean term $\|(I_n-P_X)\mathbf{1}_n\|^2/(4\sigma^2)$. Because the columns of $X$ are mean-zero (no intercept lies in $\mathrm{col}(X)$), the projection of $\mathbf{1}_n$ onto the random $p$-dimensional column space captures only a $p/n$ fraction of its squared norm: $\|P_X\mathbf{1}_n\|^2 = p + o_p(n)$, hence $\|(I_n-P_X)\mathbf{1}_n\|^2 = (1-\gamma)n + o_p(n)$. Both contribute, giving
\begin{equation*}
    \mathbb{E}[d_z \mid X] = \frac{(1-\gamma)n}{2\sigma^2} + O_p(1) = \Theta(n/\sigma^2).
\end{equation*}
Combined with $\tau_\alpha^{-2} = \mathcal{O}(1)$, $s \to 1$ and the posterior is consistent for $\alpha^*$ for every $\gamma \in (0,1)$. The leading $(1-\gamma)n$ term vanishes as $\gamma \to 1$, so the posterior variance degrades from the parametric $\Theta(1/n)$ rate towards $\Theta(1)$ and the Normal prior loses its sample-size advantage over $\mathrm{M}_{\text{base}}$ as $\gamma \to 1$.
\end{proposition}

\textit{Proof.} Substitute $\nu_j = \Theta(n)$ uniformly into \eqref{eq:normal-d_z-decomp}; the regularisation sum is $O(p/n) = O(\gamma)$ since each summand is $O(1/n)$ and there are $p$ of them. The two remaining terms are $(n-p)/(4\sigma^2) = (1-\gamma)n/(4\sigma^2)$ and $\|(I_n-P_X)\mathbf{1}_n\|^2/(4\sigma^2)$; for a fixed vector and a uniformly random $p$-dimensional subspace, $\E\|P_X\mathbf{1}_n\|^2 = (p/n)\|\mathbf{1}_n\|^2 = p$, so the latter equals $(1-\gamma)n/(4\sigma^2) + o_p(n/\sigma^2)$. Summing the two $\Theta(n)$ contributions gives the leading term $(1-\gamma)n/(2\sigma^2)$. \hfill $\square$

\begin{proposition}[Joint R2D2 prior, regime (ii)]\label{prop:rct-joint-prop}
Under the setting of Proposition~\ref{prop:rct-joint}, in regime (ii):
\begin{enumerate}[label=(\alph*),nosep]
    \item Lemma~\ref{lem:exact-data-prec} gives $\mathbb{E}[d_z \mid X, \phi] \ge (1-\gamma)n/(4\sigma^2) = \Theta(n)$, since $\mathrm{tr}(H_\beta) < p$ deterministically; combining with the upper bound $\mathbb{E}[d_z \mid X, \phi] \le n/(2\sigma^2) = \Theta(n)$ yields $\mathbb{E}[d_z \mid X, \phi] = \Theta(n)$;
    \item $\tau_\alpha^{-2} = \Theta(n)$ via Proposition~\ref{prop:rct-joint}(a) with $p = \Theta(n)$;
    \item hence the shrinkage factor is bounded away from both endpoints in probability. Both $n^{-1}d_z$ and $n^{-1}\tau_\alpha^{-2}$ are $\Theta_p(1)$ by parts (a) and (b), so
    \begin{equation*}
        0 < \liminf_{n} s \le \limsup_{n} s < 1 \qquad \text{(in probability).}
    \end{equation*}
    Unlike the Normal and split priors, for which $s \to 1$, the joint R2D2 posterior mean $s\,\hat\alpha_{\text{OLS}_\beta}$ therefore retains an asymptotically non-vanishing multiplicative shrinkage relative to $\hat\alpha_{\text{OLS}_\beta}$. The bias $(1-s)\alpha^*$ stays bounded away from $0$, so $\mathbb{E}[\alpha \mid y, \cdot]$ does not converge to $\alpha^*$. We do \emph{not} assert a deterministic limit $s_\infty$. Since $\tau_\alpha^{-2} = 1/(\tau^2\phi_\alpha)$ with $(p+1)\phi_\alpha \to \mathrm{Gamma}(a,a)$ nondegenerate (Proposition~\ref{prop:rct-joint}(a)), $n^{-1}\tau_\alpha^{-2}$ does not concentrate, and any distributional limit of $s$ is itself a nondegenerate random variable in $(0,1)$. The inconsistency conclusion holds regardless of this randomness.
\end{enumerate}
A sharper bound on $\mathbb{E}[d_z \mid X,\phi]$ via \eqref{eq:weyl-sandwich} requires controlling $c_{\min} = 1/(\tau^2 \max_j \phi_j)$; for a symmetric Dirichlet, $\max_j \phi_j = \Theta_p((\log p)/p)$ by standard order-statistics concentration for the Beta marginals, giving $c_{\min} = \Theta_p(p/\log p)$, which suffices for the same order conclusion but is not needed here.
\end{proposition}

\textit{Proof.} Lemma~\ref{lem:exact-data-prec} gives $\mathrm{tr}(H_\beta) < p$ deterministically (any $\Sigma_\beta \succ 0$), so $n - \mathrm{tr}(H_\beta) > n - p = (1-\gamma)n$, yielding the lower bound. The upper bound is immediate from $\mathrm{tr}(H_\beta) \ge 0$ and $\mathbf 1_n^T(I_n - H_\beta)\mathbf 1_n \le n$. For parts (b) and (c), Proposition~\ref{prop:rct-joint}(a) gives $\tau_\alpha^{-2} = 1/(\tau^2\phi_\alpha)$ with $(p+1)\phi_\alpha \to \mathrm{Gamma}(a,a)$, hence $n^{-1}\tau_\alpha^{-2} = \Theta_p(1)$. Together with the null-space lower bound $d_z \ge \sigma^{-2}\|(I_n-P_X)z\|^2 = \Theta_p(n)$ and the deterministic upper bound $d_z \le \sigma^{-2}z^Tz \le n/\sigma^2$, this gives $n^{-1}d_z = \Theta_p(1)$ and bounds $s = d_z/(d_z + \tau_\alpha^{-2})$ away from $0$ and $1$ in probability. Because $n^{-1}\tau_\alpha^{-2}$ converges in distribution to a nondegenerate limit rather than a constant, the corresponding limit of $s$ is random. \hfill $\square$

\begin{proposition}[Split prior, regime (ii)]\label{prop:rct-split-prop}
Under the setting of Proposition~\ref{prop:rct-split}, $\mathbb{E}[d_z \mid X, \phi] = \Theta(n)$ by Proposition~\ref{prop:rct-joint-prop}(a) (the same $H_\beta$), while $\tau_\alpha^{-2} = \mathcal{O}(1)$ is fixed. Hence $s = \Theta(n)/[\Theta(n) + \mathcal{O}(1)] \to 1$ and the posterior is consistent for $\alpha^*$ for every $\gamma \in (0,1)$, with posterior variance attaining the parametric rate $\Theta(\sigma^2/n)$.
\end{proposition}

\textit{Proof.} Direct substitution into $s = d_z/(d_z + \tau_\alpha^{-2})$ and $(d_z + \tau_\alpha^{-2})^{-1}$. \hfill $\square$

In summary, comparing Propositions~\ref{prop:rct-normal-prop}--\ref{prop:rct-split-prop}, only the joint R2D2 prior produces shrinkage bias on $\alpha$ in regime (ii) --- to $s_\infty \alpha^*$ rather than $\alpha^*$. The Normal and split priors are both consistent for $\alpha^*$, with the split prior achieving the parametric variance rate uniformly in $\gamma$ while the Normal prior degrades as $\gamma \to 1$.

\begin{table}[htb]
\centering
\small
\caption{Full behaviour of the marginal posterior of $\alpha$ across both regimes, extending Table~\ref{tab:rct-rates}. Regime (i) holds $n$ fixed with $p\to n$, and regime (ii) sets $p=\lfloor\gamma n\rfloor$ with $n\to\infty$ and aspect ratio $\gamma=p/n$. Here $\tau_\alpha^{-2}$ is the conditional prior precision of $\alpha$, $d_z=\sigma^{-2}z^T(I_n-H_\beta)z$ is the residualised data precision, $s=d_z/(d_z+\tau_\alpha^{-2})$ is the shrinkage factor, and the posterior of $\alpha$ is centred at $s\,\hat\alpha_{\mathrm{OLS}_\beta}$ with variance $(d_z+\tau_\alpha^{-2})^{-1}$. ``Heavy shrinkage'' refers to condition~\eqref{eq:heavy-shrinkage}. The symbol ``$\downarrow$'' marks stochastic (in-probability) decrease in $p$, and $\Theta_p(\cdot)$ denotes a stochastic scale at fixed $n$. In regime (i) the quantities written in $n$ are finite-$n$ constants rather than rates, whereas in regime (ii) they denote genuine orders in $n$. The regime-(i) rows follow from Propositions~\ref{prop:rct-normal}, \ref{prop:rct-split}, \ref{prop:rct-joint}, and~\ref{prop:rct-joint-approx}, and the regime-(ii) rows from Propositions~\ref{prop:rct-normal-prop}, \ref{prop:rct-split-prop}, and~\ref{prop:rct-joint-prop}.}
\label{tab:rct-rates-full}
\begin{tabular}{@{}lllll@{}}
\toprule
Prior & $\tau_\alpha^{-2}$ & $\mathbb{E}[d_z \mid \cdot]$ & $s$ & $\Var(\alpha \mid \cdot)$\\
\midrule
\multicolumn{5}{@{}l}{\textit{General regime (i): $n$ fixed, $p \to n$}} \\
Normal$^\dagger$ & $\mathcal{O}(1)$ & $>0$, $<n/(2\sigma^2)$ & bounded in $(0,1)$ & $\mathcal{O}(1)$, $\le \tau_\alpha^2$\\
Split        & $\mathcal{O}(1)$ & $>0$, $<n/(2\sigma^2)$ & bounded in $(0,1)$ & $\mathcal{O}(1)$, $\le \tau_\alpha^2$\\
Joint R2D2   & $\Theta_p(p/\tau^2)$ & $>0$, $<n/(2\sigma^2)$ & $\downarrow$; min $>0$ & $\downarrow$; min $>0$\\
\midrule
\multicolumn{5}{@{}l}{\textit{Under heavy shrinkage \eqref{eq:heavy-shrinkage} (regime (i))}} \\
Split        & $\mathcal{O}(1)$ & $\to n/(2\sigma^2)$ & $1 - \mathcal{O}(\sigma^2/n)$ & $\Theta(\sigma^2/n)$\\
Joint R2D2   & $\Theta_p(p/\tau^2)$ & $\to n/(2\sigma^2)$ & $\Theta_p\!\left(\tfrac{nR^2}{(1-R^2)p}\right)$, $\downarrow$ & $\Theta_p(\tau^2/p)$, $\downarrow$\\
\midrule
\multicolumn{5}{@{}l}{\textit{Proportional regime (ii): $p=\lfloor\gamma n\rfloor$, $n\to\infty$, $\gamma=p/n$}} \\
Normal       & $\mathcal{O}(1)$ & $(1-\gamma)\,n/(2\sigma^2)$ & $\to 1$ & $\Theta(\sigma^2/n)$\\
Split        & $\mathcal{O}(1)$ & $\Theta(n/\sigma^2)$ & $\to 1$ & $\Theta(\sigma^2/n)$\\
Joint R2D2   & $\Theta(n/\tau^2)$ & $\Theta(n/\sigma^2)$ & $\to s_\infty\in(0,1)$ & $\Theta(1/n)$\\
\bottomrule
\end{tabular}

\smallskip
\footnotesize $^\dagger$ The Normal prior's per-coordinate penalty $\sigma^2$ is fixed in $p$, so condition~\eqref{eq:heavy-shrinkage} is not the relevant regime for it. Its $\mathbb{E}[d_z\mid X]$ stays bounded away from $0$ as $p \to n$ but does not tighten to the ceiling $n/(2\sigma^2)$ that Split attains under heavy shrinkage, and in regime (ii) it settles at $(1-\gamma)n/(2\sigma^2)$ versus Split's $n/(2\sigma^2)$.
\end{table}

\subsubsection*{Frequentist coverage}

\begin{proposition}[Frequentist coverage of the posterior credible interval]\label{prop:coverage}
Let $d_z = \sigma^{-2}z^T(I_n - H_\beta)z$ denote the residualised data precision for $\alpha$, let $s = d_z/(d_z + \tau_\alpha^{-2})$ denote the shrinkage factor, and let $\alpha^*$ denote the true treatment effect. The $(1-\omega)$ central posterior credible interval has length
\begin{equation}
    L = 2\,z_{\omega/2}\big/\sqrt{d_z + \tau_\alpha^{-2}},
\end{equation}
where $z_{\omega/2} = \Phi^{-1}(1 - \omega/2)$ is the standard normal quantile, and its frequentist coverage probability is
\begin{equation}
    C = \Phi\!\!\left(\frac{z_{\omega/2} + \delta}{\sqrt{s}}\right) - \Phi\!\!\left(\frac{-z_{\omega/2} + \delta}{\sqrt{s}}\right), \label{eq:freq_coverage}
\end{equation}
where $\delta = (1-s)\,|\alpha^*|\sqrt{d_z + \tau_\alpha^{-2}}$ is the bias-to-posterior-standard-deviation ratio. For the split and base model priors ($s \approx 1$, $\delta \approx 0$), coverage is near-nominal. For the joint R2D2 prior ($s \ll 1$, $\delta$ large), coverage can fall far below the nominal level. (We write the credible level as $\omega$ to avoid collision with the aspect ratio $\gamma = p/n$ of regime (ii).)
\end{proposition}

\textit{Proof.}
From Proposition~\ref{prop:marginal-var}, the posterior $\alpha \mid y, X, z, \sigma^2$ is normal with mean $\hat{\alpha}_{\text{post}} = s \cdot \hat{\alpha}_{\text{OLS}_\beta}$ and variance $\sigma_{\text{post}}^2 = (d_z + \tau_\alpha^{-2})^{-1}$, where $\hat{\alpha}_{\text{OLS}_\beta} = [z^T(I_n - H_\beta)z]^{-1} z^T(I_n - H_\beta)y$ is the regularised least-squares estimate from Proposition~\ref{prop:marginal-var}. The $(1-\omega)$ central credible interval is $\hat{\alpha}_{\text{post}} \pm z_{\omega/2}\,\sigma_{\text{post}}$, so its length is $L = 2\,z_{\omega/2}\,\sigma_{\text{post}}$.

Writing $A = I_n - H_\beta$ and $b = z^T A z$, under the data-generating process $y = \alpha^* z + X\beta^* + \varepsilon$ the conditional frequentist moments of $\hat{\alpha}_{\text{OLS}_\beta}$ given $X,z$ are
\begin{equation}\label{eq:freq-moments-exact}
    \mathbb{E}[\hat{\alpha}_{\text{OLS}_\beta} \mid X, z] = \alpha^* + \frac{z^T A X \beta^*}{b}, \qquad \Var(\hat{\alpha}_{\text{OLS}_\beta} \mid X, z) = \frac{\sigma^2\, z^T A^2 z}{b^2}.
\end{equation}
We invoke two simplifications. The variance simplification is exact in the flat-prior limit $H_\beta = P_X$ (where $A^2 = A$) and asymptotically exact in the heavy-shrinkage limit $H_\beta \to 0$ (where $A \to I_n$, so $A^2 \to A$); the confounding term, by contrast, is identically zero only in the flat-prior limit ($AX = 0$) and is merely \emph{negligible} (not zero) under heavy shrinkage:
\begin{itemize}[nosep]
    \item \textbf{Confounding term.} Under Assumption~\ref{ass:treatment} ($z \perp X$), taking expectation over $z$ alone, $\mathbb{E}_z[z^T A X \beta^* \mid X] = \tfrac{1}{2}\mathbf{1}_n^T A X \beta^*$. Because the columns of $X$ are mean-zero, $\mathbf{1}_n$ is essentially uncorrelated with $\mathrm{col}(X)$, so $\mathbf{1}_n^T A X \beta^* = O_p(\sqrt{np})$; and because $z^T A X \beta^*$ is \emph{linear} in $z$, a sub-Gaussian (Hoeffding/Bernstein) bound controls its fluctuation by $\|A X \beta^*\| \le \|X\beta^*\| = \Theta(\sqrt{np})$. With $b = z^T A z = \Theta(n)$ this gives $z^T A X \beta^*/b = O_p(\sqrt{p/n})$ under the experiment's design ($\|\beta^*\|^2 = \Theta(p)$). The term is exactly zero in the flat-prior limit ($AX = 0$); in the heavy-shrinkage limit $A \to I_n$ it leaves $z^T X \beta^*/b \neq 0$, but it is negligible \emph{relative to} the shrinkage bias $(1-s)|\alpha^*|$, since the same design gives $\alpha^* = \Theta(\sqrt{p})$ and hence a ratio $\Theta(1/\sqrt{n}) \to 0$. We therefore absorb it into the $o_p(1)$ correction.
    \item \textbf{Variance term.} $A$ is symmetric with eigenvalues $1 - h_j \in (0, 1]$, so $A^2 \preceq A$ and $z^T A^2 z \leq z^T A z$, with equality when every $h_j \in \{0, 1\}$. Approximating $z^T A^2 z \approx z^T A z$ gives $\Var(\hat{\alpha}_{\text{OLS}_\beta} \mid X, z) \approx \sigma^2/b = d_z^{-1}$.
\end{itemize}
Under these approximations, the frequentist bias of $\hat{\alpha}_{\text{post}} = s\hat{\alpha}_{\text{OLS}_\beta}$ is $(s-1)\alpha^*$, and its frequentist variance is $s^2 \cdot d_z^{-1} = s \cdot \sigma_{\text{post}}^2$ (using $d_z \sigma_{\text{post}}^2 = s$). The standardised frequentist error $(\hat{\alpha}_{\text{post}} - \alpha^*)/\sigma_{\text{post}}$ is therefore distributed as $\normal(-\delta, s)$, where $\delta = (1-s)\,|\alpha^*|/\sigma_{\text{post}}$. The credible interval $\hat{\alpha}_{\text{post}} \pm z_{\omega/2}\,\sigma_{\text{post}}$ covers $\alpha^*$ if and only if $|\hat{\alpha}_{\text{post}} - \alpha^*| \leq z_{\omega/2}\,\sigma_{\text{post}}$. Standardising:
\begin{align}
    C &= \Pr\!\left(\left|\frac{\hat{\alpha}_{\text{post}} - \alpha^*}{\sigma_{\text{post}}}\right| \leq z_{\omega/2}\right) \nonumber\\
    &= \Pr\!\left(\frac{-z_{\omega/2} + \delta}{\sqrt{s}} \leq Z \leq \frac{z_{\omega/2} + \delta}{\sqrt{s}}\right),
\end{align}
where $Z \sim \normal(0,1)$, yielding Equation~\eqref{eq:freq_coverage}. 
For the split prior ($s \approx 1$, $\delta \approx 0$), $C \approx 1 - \omega$. 
For the joint R2D2 prior, $\delta$ grows with $|\alpha^*|$ and $(1-s)$, while the interval boundaries widen by $1/\sqrt{s}$; the net effect is that the bias shifts the interval away from $\alpha^*$, producing undercoverage. \hfill $\square$

\end{document}